\newif\ifnotes
\notestrue

\newif\ifSTOC
\newif\ifdraft

\documentclass[12pt]{article}
\usepackage{authblk}
\usepackage[margin=1in]{geometry}
\usepackage{mathtools}
\usepackage{amsmath}
\usepackage{amssymb}
\usepackage{amsthm}
\usepackage{mathrsfs}
\usepackage{subfigure}
\usepackage{customcommands}
\usepackage{bm}
\usepackage[normalem]{ulem}
\usepackage{Nettastyle}
\usepackage{comment}
\usepackage{mathtools}

\newcommand\coolover[2]{\mathrlap{\smash{\overbrace{\phantom{%
    \begin{matrix} #2 \end{matrix}}}^{\mbox{$#1$}}}}#2}

\newcommand\coolleftbrace[2]{%
#1\left\{\vphantom{\begin{matrix} #2 \end{matrix}}\right.}

\newcommand{\research}[1]{}
\newcommand{\cut}[1]{}

\title{The Complexity of Learning (Pseudo)random Dynamics of Black Holes and Other Chaotic Systems}

\author[1]{Lisa Yang}
\author[2]{and Netta Engelhardt}
\affiliation[1]{Computer Science and Artificial Intelligence Laboratory,\\
Massachusetts Institute of Technology, Cambridge, MA 02139, USA}
\affiliation[2]{Center for Theoretical Physics,\\
Massachusetts Institute of Technology, Cambridge, MA 02139, USA}

\emailAdd{lisayang@mit.edu}
\emailAdd{engeln@mit.edu}

\date{\today}

\abstract{It has been recently proposed that the naive semiclassical prediction of non-unitary black hole evaporation can be understood in the fundamental description of the black hole as a consequence of ignorance of high-complexity information. Validity of this conjecture implies that any algorithm which is polynomially bounded in computational complexity cannot accurately reconstruct the black hole dynamics. In this work, we prove that such bounded quantum algorithms cannot accurately predict (pseudo)random unitary dynamics, even if they are given access to an arbitrary set of polynomially complex observables under this time evolution; this shows that ``learning'' a (pseudo)random unitary is computationally hard. We use the common simplification of modeling black holes and more generally chaotic systems via (pseudo)random dynamics. The quantum algorithms that we consider are completely general, and their attempted guess for the time evolution of black holes is likewise unconstrained: it need not be a linear operator, and may be as general as an arbitrary (e.g. decohering) quantum channel. 

}

\newtheorem{manualtheoreminner}{Theorem}
\newenvironment{manualtheorem}[1]{
  \renewcommand\themanualtheoreminner{#1}
  \manualtheoreminner
}{\endmanualtheoreminner}

\renewcommand{\[}{\left[}
\renewcommand{\]}{\right]}
\renewcommand{\(}{\left(}
\renewcommand{\)}{\right)}
\renewcommand{\gets}{\leftarrow}

\newcommand{\calR}{\mathcal{R}} 

\newcommand{\N}{\mathbb{N}}
\newcommand{\R}{\mathbb{R}}
\newcommand{\C}{\mathbb{C}}
\newcommand{\cD}{\mathcal{D}} 
\newcommand{\Ndistr}{\mathcal{N}}
\newcommand{\CNdistr}{\mathcal{CN}}
\newcommand{\CN}{R} 

\newcommand{\E}{\mathbb{E}}
\newcommand{\Eover}[1]{\underset{#1}{\E}}
\newcommand{\intover}[1]{\underset{#1}{\int}}
\newcommand{\Id}{\mathrm{I}} 
\newcommand{\SWAP}{\mathrm{SWAP}} 
\newcommand{\Haar}{\mu}
\newcommand{\KLD}{D_{\mathrm{KL}}} 
\newcommand{\md}{m} 
\newcommand{\TV}{\rm TV}

\newcommand{\zo}{\{0,1\}}
\newcommand{\poly}{\mathrm{poly}}

\newcommand{\secp}{\lambda} 
\newcommand{\Hn}{n} 
\newcommand{\Hd}{d} 
\newcommand{\Qgrav}{U}
\newcommand{\BH}{\psi}
\newcommand{\SCgrav}{O} 
\newcommand{\A}{\mathcal{A}} 
\newcommand{\QCC}{\mathsf{QCC}}
\newcommand{\prot}{\pi}
\newcommand{\Qalg}{\mathcal{Q}} 
\newcommand{\Qcir}{\mathcal{Q}} 

\newcommand{\SCcol}{\alpha} 
\newcommand{\Eapp}{\mathsf{APRX}} 
\newcommand{\Ugrav}{\widehat{U}} 
\newcommand{\qd}{q} 
\newcommand{\Vmat}{V} 
\newcommand{\Qmat}{Q} 

\newcommand{\Dist}{\mathsf{Dist}}
\newcommand{\Key}{\mathcal{K}}
\newcommand{\key}{k}
\newcommand{\PRU}{\mathcal{U}} 
\newcommand{\pru}{U} 
\newcommand{\PRS}{\Psi} 
\newcommand{\prs}{\psi} 
\newcommand{\p}{p} 
\newcommand{\neglf}{\eta}

\newcommand{\Hil}{\mathcal{H}}
\newcommand{\qop}{\mathcal{E}} 
\newcommand{\Kraus}{O}
\newcommand{\Knorm}{\beta}
\newcommand{\emm}{m} 

\begin{document}

\maketitle

\ifSTOC
\setcounter{page}{0}
\fi

\section{Introduction}\label{sec:intro}

The past two decades have highlighted a deep connection between computational complexity and fundamental physics (see e.g.~\cite{HarHay13, KimTan20, StaSus14, RobSta14, SusZha14, Ali15, BroRob15, BroRob15b, LehMye16, CouFis16, ChaMar16, CarMye16, CarCha17, EngWal17a, EngWal18, BroGha19, EngPen21a, EngPen21b, AkeEng22}). One such instance includes the dynamics of a large set of chaotic systems -- notably including black holes -- which are generally well-modeled by highly complex (or apparently complex) unitary time evolution~\cite{HayPre07, ChoSha21, HoCho22, RobYos17, Haa91, Pag93a}. Another such instance includes recent work~\cite{KimTan20, AkeEng22} on the black hole information paradox~\cite{Haw75}, which has leveraged high complexity unitary dynamics to show that the tension between semiclassical effective field theory and quantum gravity in an old black hole can be relaxed by limiting the regime of validity of the former using complexity. 

One of the upshots of these recent developments on complexity within gravity is a precise notion of \textit{complexity coarse-graining}. The idea of complexity coarse-graining in gravity in general and black holes in particular was initially described in~\cite{EngWal17b} as a way of quantifying the ignorance of a computationally bounded\footnote{Consistently with quantum information theoretic conventions, we use the term ``computationally bounded'' to mean bounded at polynomial (or subexponential) computational complexity (i.e. running time) in $\log\dim {\cal H}$, or $S$.} observer about the black hole interior. To be precise,~\cite{EngWal17b, EngWal18} (see also the refinement of~\cite{EngPen21a}) defined an entropy called the Simple Entropy for a state $\psi$: define the set ${\cal S}$ of states $\phi$ (where we include density matrices in the definition of a state) on the black hole Hilbert space ${\cal H}$ which agree with $\psi$ on all simple observables ${\cal O}$, i.e. 
\be
\langle {\cal O}\rangle_{\psi}= \langle {\cal O}\rangle_{\phi}
\ee
where simple means that ${\cal O}$ is at most polynomially complex in $\log \dim {\cal H}$. The Simple Entropy of $\psi$ is defined:
\be
S^{\rm simple}[\psi]\equiv \max\limits_{\phi\in {\cal S}} S_{\rm vN}[\phi].
\ee
This Jaynesian coarse-graining~\cite{Jay57a, Jay57b} can be thought of as a protocol that forgets any fine-grained -- meaning high complexity -- information in the state beyond the expectation values of simple operators. In particular, it replaces a (possibly pure) state $\psi$ with a more mixed state $\phi$.\footnote{A notable exception is the vacuum: rigidity theorems protect the gravitational vacuum, so in that case the set $S$ has only a single state in it: $\psi=\ket{0}$~\cite{EngWal17a}.} In AdS/CFT~\cite{Mal97}, the simple entropy is computed by the generalized entropy of the outermost quantum extremal surface~\cite{EngWal14} (QES)~\cite{EngWal17b, EngWal18, BouCha19}. 

In the QES calculation of the Page curve~\cite{Pag93b}~\cite{Pen19, AEMM}, the outermost QES for the radiation is nonminimal after the Page time. In particular, it does not dominate the computation of $S_{\rm vN}$ of the actual state, and its entropy does not follow the Page curve, but rather the so-called  ``Hawking curve''. The Simple Entropy may thus potentially be interpreted as the quantity in quantum gravity that corresponds to the entropy calculation in the non-unitary Hawking analysis of black hole evaporation. This perspective is supported by the Python's Lunch proposal~\cite{BroGha19} that reconstruction of the region between the outermost QES and the minimal QES is exponentially complex. In the case of the evaporating black hole, this corresponds to the entire interior: as predicted by Harlow-Hayden~\cite{HarHay13}, it is exponentially hard to reconstruct the interior from the radiation. The coarse-grained state $\phi$ whose fine-grained entropy computes $S^{\rm simple}[\psi]$ is precisely the ``apparent state'' that a computationally bounded reconstruction procedure should arrive at. The fact that this state also computes Hawking's non-unitarily evolving entropy was demonstrated rigorously in the recent work~\cite{AkeEng22} in the context of a class of non-isometric quantum code models of the bulk-to-boundary map for an evaporating black hole.

How generally does this paradigm extend to the discrepancy between time evolution via a decohering quantum channel predicted by Hawking's analysis and via unitary  dynamics as in the fundamental theory? Does complexity coarse-graining within the fundamental theory \emph{imply} an inconsistency with the fundamental time evolution in a black hole setting?

Consider a computationally bounded quantum algorithm which is given access to the expectation values of any operators in the fundamental theory; these expectation values may be computed including sources, as in the reconstruction protocol of~\cite{EngWal17b, EngWal18, EngPen21a}. Can such a quantum algorithm use this information to accurately predict the time evolution of a black hole (including its interior)? If the Hawking analysis is indeed emergent from complexity coarse-graining, the answer to this question must be \textit{no}. For the closest match to the Hawking calculation, a computationally bounded quantum algorithm's best guess for the time evolution of an evaporating black hole must be a decohering channel rather than a unitary operator. 

What makes black holes special? We expect that bulk reconstruction in spacetimes without horizons\footnote{Assuming some version of cosmic censorship that allows evaporating black holes but forbids more violent naked singularities formed from collapse, see e.g. ~\cite{EngFol20}.} is simple~\cite{BanDou98, HamKab05, HamKab06, HamKab06b, HeeMar, BouFre12, EngPen21b}. Horizons thus appear to be a primary obstacle to a simple reconstruction. Furthermore, black hole dynamics are understood to be chaotic (see e.g.~\cite{SekSus09, MalShe15}) and are often modeled with quantum codes using unitary operators with some degree of randomness, as in~\cite{HayPre07, HarHay13}. It is not uncommon to use Haar random unitaries (see e.g.~\cite{CotGur16, PirSun20, AkeEng22}), though these are surely unrealistic idealizations, and $k$-designs (see e.g.~\cite{HayPre07}). Due to the growing number of connections between chaotic systems and \textit{pseudorandomness} -- the property of being simultaneously efficiently generated while also  indistinguishable from Haar random by any computationally bounded algorithm -- we may also expect that pseudorandomness could be another property that sets black holes apart from other gravitational systems with less random dynamics.\footnote{The connection between chaos and horizons has been explored in a class of cases; however in the absence of a general established connection between a generic horizon and chaos, we will refrain from using the two concepts interchangeably.}

In this article, we show  that computationally bounded quantum learning algorithms cannot accurately predict the time evolution of general chaotic systems whose dynamics are pseudorandom.  Our proof relies on the high `information content' (high entropy) that such dynamics appear to have.\footnote{This refers to the average information content (also called the differential entropy) of the time evolution operator as a random variable.} We consider a learning algorithm whose task is to predict the time evolution of such a system. Given a state $\psi$, the algorithm outputs a state $\phi$ -- its best guess at the time evolution of $\psi$ under the pseudorandom unitary dynamics $U$. There are then two descriptions of the system: the actual exact pseudorandom dynamics and time-evolved state, $U$ applied to $\psi$, and the ``apparent'' low complexity dynamics and state $\phi$. We allow the states $\phi$ to be mixed density matrices and \textit{we will allow the output of the quantum algorithm to be a general quantum channel}. This makes immediate contact with the hypothesis that complexity coarse-graining produces Hawking calculation.\footnote{We leave the question of the actual nature of the output -- single linear operator or quantum channel -- to future work. Addressing this question would significantly clarify the connection between Hawking's calculation and complexity coarse-graining, but it is beyond the scope of the current work.}  Borrowing from the terminology introduced by~\cite{AkeEng22}, we call the actual description of the system the ``fundamental description''. We call the best description of the system as given by a computationally bounded algorithm aiming to produce the fundamental description the ``simple description''.~\footnote{In AdS/CFT, this would correspond to the spacetime volume that can be reconstructed from the boundary with access to arbitrary boundary time.} In the language of the black hole information problem, if Hawking's analysis can in fact be traced back to a bound in computational complexity, then such a coarse-graining should be sufficiently powerful to prevent computationally bounded observers from approximately learning the black hole's S-matrix. 

Our technical result is a set of general bounds on the accuracy of the simple description. On the most general level, we prove a bound on the accuracy in terms of the (differential) entropy of the dynamics; for (pseudo)random dynamics, this yields the desired bounds. While our motivation is inherited from gravity, our results are broadly applicable to  \textit{general} chaotic systems, with no restrictions (e.g. no particular code models) beyond the following three assumptions:
\begin{enumerate}
    \item The time evolution of the system is well-modeled by a Haar random or pseudorandom unitary operator.
    \item The state of the system is well-modeled by a Haar random or pseudorandom state.
    \item The prediction of the simple description is obtained as a result of a quantum learning process of bounded complexity, as described below.
\end{enumerate}

Assumptions 1 and 2 are generally expected to hold for many chaotic systems, in particular, for black holes in quantum gravity. We \textit{define} the simple description to be assumption 3. We model the observations and predictions that can be made within the simple description via a computationally bounded quantum \emph{learning} algorithm (assumption 3 above), as algorithms are a universal model for computation which, e.g. in the context of codes, includes any encoding procedure and (attempts to) decode. The learning algorithm has access to computationally bounded observations of the fundamental time evolution of states (modeled as a (pseudo)random unitary) and uses this to attempt to approximately predict the time evolution of the (pseudo)random state. Any ``approximate'' notions will always mean up to errors of $O(e^{-S})$, where $S=\log\dim\cal{H}$ and ${\cal H}$ is the Hilbert space (in the black hole context, $S$ is the entropy of the black hole, and we only consider black holes after scrambling). Our model uses the probably approximately correct (PAC) model of learning,\footnote{In the language of learning theory, the algorithm's goal is to learn the operator -- here the (pseudo)random unitary -- that it has access to. We use a model for learning quantum operators (the quantum analog of a function) and allow the learner to query the operator on any input states of its choice (rather than only on inputs from a fixed distribution as in the original PAC model). The formal definitions of our model are in Section~\ref{sec:prelim-comp-models}.} a foundational notion in computational learning theory introduced in~\cite{Val84} and widely studied thereafter (see e.g.~\cite{Ans22,AruWol17} for recent reviews on \emph{quantum} learning theory).\footnote{In this work, we use algorithms to refer to the formal notion of computations as is studied in computer science. 
Our results in this work are within the model of algorithms, which we suggest widely encompasses many computational processes and has been shown to include the formally defined class of ``computable functions''.
We will not use any version of the Church-Turing thesis to formally claim that \emph{any} process of logical deductions and calculations, or any physical process, can be simulated by an algorithm (which would then imply that our algorithmic model can be all encompassing). We note that~\cite{Sus20} argues that the Quantum Extended Church-Turing thesis would at least need to be modified to apply only to physical systems that stay outside the horizons of black holes.} We build upon this notion and the techniques of~\cite{AruGri21} to show our results.
We allow the learning algorithm to use any completely-positive trace preserving (CPTP) map, i.e. quantum channel, as its model for the pseudorandom unitary dynamics. This includes the possibility that the algorithm's best guess for the dynamics is a \emph{decohering} channel -- as in the non-unitary naive semiclassical analysis of Hawking's result.\footnote{We thank D. Harlow for discussions on this point.}

Our technical results show that there \textit{must} be a significant discrepancy between the simple prediction and the fundamental time evolution of chaotic systems. Prima facie, this result may seem vacuous: after all, how could a computationally bounded algorithm ever construct a Haar random (or pseudorandom) unitary? Such Haar random unitaries take an exponentially long time to build. The crucial ingredient here is that $U$ is already sampled from the Haar measure, and our algorithm -- much like an experimentalist -- has a potential shortcut: oracle access to $U$, i.e. queries to the sampled unitary which do not cost it any computational complexity.

This oracle access \emph{is} sufficient for a bounded algorithm to learn to predict many of the observables even if the time evolution operator $U$ is an exponentially complex unitary. This was shown in~\cite{HuaChe22}, a development which we discuss further in Section~\ref{sec:subsecpriorwork}. Thus it is possible for bounded algorithms to learn to predict some aspects of exponentially complex dynamics, e.g. the values of `simple' observables.
Our focus will be on whether it can learn the \emph{fine-grained} features of the correct state -- i.e. with the correct values for \emph{all} observables. This statement is not derived solely from the randomness in $\Qgrav$ but also importantly depends on the \emph{granularity} of what the algorithm aims to predict. Intuitively, the algorithm can access $O(\poly\log\dim\cal{H})$ qubits of information about $U$ through its queries; however, we show that even with oracle access its bounded complexity prevents it from learning the exponential amount of (highly entropic) information in $U$, preventing it from accurately predicting the fine-grained time evolution. 

It may seem clear nevertheless that there is a straightforward argument establishing this result. Let us briefly describe this (erroneous) argument: the columns of a Haar random unitary are frequently (naively) thought of as essentially $\Hd$ random vectors. Since our algorithm can at most query a small, $\frac{\poly\log\Hd}{\Hd}$ fraction of the matrix, it would appear clear that it cannot learn even most of the unitary. This simple argument is flawed, however: a random $U$ cannot be thought of as $\Hd$ independent random vectors. Only a small fraction of the columns are actually close to independent~\cite{Jia10,Jia09}. Nevertheless, we are able to take this into account and formalize a proof: we use~\cite{Jia10} to establish the number of columns of $U$ which are approximately independent as random vectors in order to show that the first $\Hd/\log\Hd$ columns cannot be learned with $\ll \Hd/\log\Hd$ queries. This constitutes a significant $1/\log\Hd$ fraction of the matrix $U$ yielding first a $1-\frac{1}{\log\Hd}$ bound, which can be improved to a $1-\Omega(1)$ bound (where $\Omega(1)$ denotes a quantity lower bounded by a constant independent of $\Hd$) via a re-partition of $U$ into $\log \Hd$ sets of columns. Together this yields our final results, that the columns of $U$ cannot be approximated well by any algorithm. Our full results are actually non-trivial quantitative bounds that depend on the operator $\Qgrav$; see Theorem~\ref{thm:intro-fidelity-general-learner} and the discussion thereafter for our full results.

Having established the nontriviality of our result, let us address another potential confusion: how closely tied are our conclusions to the extant results of Harlow and Hayden~\cite{HarHay13}. While the general motivation for the role of complexity in black hole evaporation is rooted in~\cite{HarHay13}, we are addressing a different question -- about the time evolution of an evaporating black hole, not about the existence of firewalls --  using different techniques -- quantum learning theory -- and with different assumptions (in particular, we do not assume the existence of one-way functions or require any assumptions on quantum statistical zero knowledge).

Let us now briefly describe the algorithm process and the definition of success in the learning process, and then state our quantitative results. As above let ${\cal H}$ be the fundamental theory Hilbert space, $\ket{\psi}$ a pseudorandom state in ${\cal H}$ (e.g. representing a black hole post scrambling), and $U$ the actual fundamental time evolution, modeled by a pseudorandom unitary $U$.  The task of our quantum algorithm is to output a state $O\ket{\psi}$ as close as possible to $U\ket{\psi}$: that is, to guess the time evolution of a (pseudo)random state. The algorithm has the following properties:
\begin{enumerate}
    \item Queries and Computations: the algorithm can make multiple queries to the fundamental description for any states in ${\cal H}$. That is, for any $\ket{\phi_i}\in {\cal H}$ it can ask for $U\ket{\phi_i}$. The algorithm can perform any computations between queries, on both the $U\ket{\phi_i}$ from its queries thus far, and on its own registers. This includes computations that entangle $U\ket{\phi_i}$ with its own registers, and measuring $U\ket{\phi_i}$. \footnote{We write $\ket{\phi_i}$ as a pure state only for notational convenience: all of our results actually allow the algorithm to query for mixed states $\rho_i$, receiving $\Qgrav\rho_i\Qgrav^\dagger$.} 
    \item Computationally Bounded: the algorithm is limited to a polynomial number of queries in $\log \dim{\cal H}$ 
    though each query can be for any state in ${\cal H}$ that it has (as ``advice'') or can prepare. The algorithm is also limited in processing power: even given $U\ket{\phi_i}$, it can only compute quantities from it that are polynomially complex in $\log \dim {\cal H}$. 
\end{enumerate}
For concreteness, we describe an example of a quantum algorithm that is allowed in this model: the algorithm has a set of states $\ket{\phi_1},\ldots,\ket{\phi_\ell}$, a combination of ``advice'' states, e.g. states already existing in the physical world, and states that it can prepare efficiently, and of which it may prepare multiple copies. ``Querying'' $\Qgrav$ on any $\ket{\phi_i}$ corresponds to $U$'s time evolution, turning the state into $\Qgrav\ket{\phi_i}$. The algorithm can iteratively query, perform efficient quantum computations on any of its states and ancilla registers, and prepare new $\ket{\phi_i}$. Its computations could be a combination of calculations in effective field theory and other calculations of bounded complexity (running time). It aims to produce a model of $\Qgrav$, e.g. it may produce a quantum circuit model for $\Qgrav$ where the $\SCgrav\in\C^{\Hd\times\Hd}$ that we discuss here is the operator that the circuit implements, although we will also present a more general result for quantum channels. (Note that $O$ stands for the operator that $\A$ learns, not an observable.) Then given a (pseudo)random state $\ket{\BH}$, the algorithm $\A$ would apply this quantum circuit to $\ket{\BH}$ to produce its prediction $\SCgrav\ket{\BH}$. We note an important (and realistic) aspect of our model: there is only a single copy of the state $\ket{\BH}$ representing the black hole.\footnote{The comparison between $\SCgrav\ket{\BH}$ and $\Qgrav\ket{\BH}$ is only in the theoretical analysis of fidelity.} This is given to the learning algorithm. Since there is only a single copy of $\ket{\BH}$, the algorithm cannot e.g. query for $\Qgrav$ applied to $\ket{\BH}$ and discern that its prediction $\SCgrav\ket{\BH}$ differs significantly e.g. using the SWAP test. Its model $\SCgrav$ is also formed before it is given $\ket{\BH}$;  afterwards it no longer has query access to $\Qgrav$ (as otherwise it could simply use $\Qgrav$ to predict $\Qgrav\ket{\BH}$ and never form its own model to begin with).

We define success in the following way: an algorithm \textit{learns} if the fidelity between $O\ket{\psi}$ and  $U\ket{\psi}$ is near maximal -- that is, exponentially close to 1 in $\log\dim\mathcal{H}=S$. 
That is, the learning algorithm is successful if \textit{on average}:
\be \label{eqn:intro-def-success}
\left | \bra{\psi}{O}^{\dagger} U\ket{\psi} \right|^{2}\geq 1-\frac{1}{{\rm exp}(\log\dim\mathcal{H})}=1-\frac{1}{{\rm exp}(S)}.
\ee
The algorithm fails at its task if on average:
\be \label{eqn:intro-def-failure}
\left | \bra{\psi}{O}^{\dagger} U\ket{\psi} \right|^{2}\leq 1-\frac{1}{{\rm poly}(\log \dim {\cal H})}=1-\frac{1}{{\rm poly}(S)}.
\ee
The above definition of failure\footnote{Note that we do not define success or failure for fidelity which is subexponentially but not polynomially bounded from $1$. This distinction is immaterial to our results, which are bounded below 1 by $\Omega(1)$.} may at first appear quite weak: after all, it seems that the simple description can still learn quite a lot about the fundamental theory if the fidelity is 1 up to polynomial corrections in $1/S$. It is indeed true that \textit{some learning} can happen (and indeed, must happen~\cite{HuaChe22}): this is consistent with our expectations that the time evolution of simple observables can be predicted by the simple description; in AdS/CFT, this corresponds to reconstruction exclusively outside of the Python's Lunch. However, as we explain in Appendix~\ref{sec:appendix-POVM}, this bound implies  that there is a positive operator-valued measure (POVM)\footnote{POVMs are the most general kind of measurement in quantum mechanics, generalizing projective measurements~\cite{NieChu11}. They take into account that a system of interest may be a part of a larger system: POVMs can describe how the system of interest is affected by projective measurements on a larger system.} -- i.e. a quantum measurement -- that with significant probability, distinguishes between the algorithm's prediction and the fundamental time evolution of the state, as we had set out to show.

Nevertheless, we may wish for a stronger bound that would definitively imply a large deviation between the simple and fundamental descriptions. Ambitiously, one might aim for an $o(1)$ bound close to `random guessing', however our focus in this work is not foremost on tightly characterizing how limited quantum learning algorithms are (just that they are significantly limited): a $1-1/\poly(\log d)$ bound would suffice for our purposes to show the simple description to be inaccurate past exponentially suppressed corrections. We do better than such a $1-1/\poly(\log d)$ bound, improving to a $1-\Omega(1)$ bound; the latter implies that there is a quantum measurement that distinguishes between the simple description and the fundamental time evolution with \emph{constant} probability. We thus show that the computationally bounded algorithm above cannot learn the time evolution operator $U$. Quantitatively, our technical result for Haar random $U$ is:

\begin{manualtheorem}{1}[Hardness of Learning Haar Random Unitaries]\label{thm:intro-fidelity-random}
Let $\Qgrav$ be a Haar random unitary and $\ket{\BH}$ a Haar random state in a Hilbert space of dimension $\Hd$.
For any quantum algorithm $\A$ that makes $o\(\frac{\Hd^2}{(\log\Hd)^{2}}\)$ quantum queries to $\Qgrav$, produces a single linear operator $\SCgrav$ (aiming to produce $\Qgrav$) and given $\ket{\BH}$ outputs $\SCgrav\ket{\BH}$,  
$$\underset{\Qgrav,\SCgrav,\ket{\BH}}{\mathrm{avg}}\[\left|\bra{\BH}\SCgrav^\dagger \Qgrav\ket{\BH}\right|^2\]\leq 1-\Omega\(1\)$$
where $\underset{\Qgrav,\SCgrav,\ket{\BH}}{\mathrm{avg}}$ denotes averaging over $\Qgrav,\SCgrav,\ket{\BH}$. 
This bound holds for any $\SCgrav$ with column norms upper bounded by $1+\frac{1}{\poly (d)}$. 
\end{manualtheorem} 
Note that the purposes of the averaging above is only to describe the behavior of typical operators: we do \textit{not} invoke fundamental averaging. The fundamental dynamics are described by a \textit{single} unitary operator. Note that the predicted time evolution operator $\SCgrav$ can be highly nonunitary even if the column norms of $O$ are close to $1$, $O^{\dagger}O$ need not even be approximately close to the identity:\footnote{We thank C. Akers for discussions on this point.} a highly non-unitary predicted time evolution. We comment that in all of our theorems, the complexity of the black hole dynamics is of the same order as the complexity of the algorithm: here the algorithm is allowed exponential complexity, when we discuss pseudorandom dynamics below, the algorithm will have subexponential complexity, and analogously for Theorem 3 regarding arbitrary distributions of $U$. Thus the algorithm in principle always has the computational resources we would expect it to require in order to produce $U$.

A close examination of Theorem~\ref{thm:intro-fidelity-random} reveals that it holds even for computationally unbounded learners $\A$ as long as they make a bounded number of queries. This is unsurprising: after all, Haar random dynamics are too random and too complex to match all of the expected behaviors of chaotic systems in general and black holes in particular. A simple way to see this is that Haar random unitaries cannot be efficiently generated, while  we do expect the fundamental dynamics of our system (in particular a black hole) to be efficiently generated. Therefore, once we prove the hardness of learning for random unitaries and states, we relax the restriction to pseudorandom unitaries and states. This is motivated by the connection between pseudorandom states and chaotic systems~\cite{HoCho22} (see also~\cite{CotHun17} and \cite{ChoSha21}) and the expectation that black holes are maximally chaotic systems~\cite{SheSta13, MalShe15}.\footnote{Admittedly, the notion of pseudorandomness developed in~\cite{HoCho22} is not necessarily the same as the pseudorandom notions, for unitaries and states~\cite{JiLiu18}, that we use in this work. However, the heuristic connection between the two still provides reasonable motivation for identifying chaotic systems with the pseudorandom notions that we use.} We note that measures of chaos can even be directly related to relaxations of Haar randomness: ~\cite{RobYos17} shows a (tight) quantitative relationship between out-of-time-order correlators and $k$-designs, which, like the pseudorandom notions~\cite{JiLiu18} that we use, capture some but not all of the properties of Haar random unitaries. For pseudorandom unitaries and states, we show: 

\begin{manualtheorem}{2}[Hardness of Learning Pseudorandom Unitaries]\label{thm:intro-fidelity-pseudorandom}
Let $\Qgrav$ be a pseudorandom unitary and $\ket{\BH}$ a pseudorandom state in a Hilbert space of dimension $\Hd$. For any family of $\poly(\log\Hd)$-size quantum circuits $\A$ with quantum query access to $\Qgrav$ that produces a unitary operator $\SCgrav$, and given $\ket{\BH}$ outputs $\SCgrav\ket{\BH}$,
$$\underset{\Qgrav,\SCgrav,\ket{\BH}}{\mathrm{avg}}\[\left|\bra{\BH}\SCgrav^\dagger \Qgrav\ket{\BH}\right|^2\]\leq 1-\Omega\(1\).$$
\end{manualtheorem} 
Thus the hardness of learning result \textit{still applies} to any unitary and state drawn from ensembles which are pseudorandom as defined in~\cite{JiLiu18}. Indeed, now that $\Qgrav$ is pseudorandom, we \textit{do} require a bound on the learning algorithm's computational complexity, which is consistent with the complexity hypothesis for semiclassical gravity starting with~\cite{HarHay13}. If the unitary and state in Theorem~\ref{thm:intro-fidelity-pseudorandom} are pseudorandom for sub-exponential size circuits,\footnote{Such constructions are possible using a sub-exponential hardness assumption e.g. sub-exponential LWE (Learning with Errors) which is commonly used in this area.} then Theorem~\ref{thm:intro-fidelity-pseudorandom} holds for any family $\A$ of sub-exponential size quantum circuits.  The reader may be surprised to learn that we have here restricted $\SCgrav$ to be unitary, in contrast with the previous result about Haar random $\ket{\psi}$ and $U$. While this may prima facie appear to be a significant restriction, it is an artifact of the proof technique: we will generalize the permitted output of the algorithm to include general quantum channels below.\footnote{This is a consequence of this proof requiring a circuit implementation of the quantum algorithm's output, which we know exists for operators obeying the axioms of quantum mechanics -- unitaries and quantum channels.}

We note a few constructions of pseudorandom states:~\cite{BouFef19} gives a construction of pseudorandom states specifically in the context of AdS black holes. See also the very recent work~\cite{BouFef22} for another construction of pseudorandom states. As pseudorandom unitaries and states are recent notions~\cite{JiLiu18}, there are currently only candidate constructions of pseudorandom unitaries, no constructions proven to be pseudorandom. The notion of pseudorandomness that we and recent works use is specifically for \emph{computationally bounded} algorithms, and seems a fitting model for the fundamental dynamics.

Having relaxed the requirement of Haar random to pseudorandom, we might wonder whether the chaotic or random nature of the system is a sine qua non for our results; that is, is even apparent randomness truly essential for our result? We expect that some amount of apparent randomness \textit{is} in fact necessary. When the system is simple -- e.g. low-energy perturbations of the vacuum on short time scales -- we expect that the coarse-grained theory \textit{can} be successful at this learning task. Indeed, some pseudorandomness of the state that we average over is critical: if $\ket{\psi}$ were not sufficiently (apparently) random (e.g. sampled from an ensemble of polynomially many states), then the algorithm could sample a state $\ket{\psi'}$ from this distribution on its own and with $1/\poly(\log\Hd)$ probability, obtain $\ket{\psi'}=\ket{\psi}$. Then the quantum algorithm can query the oracle $U$ on $\ket{\psi'}$ and simply use its output $U\ket{\psi}$ as the prediction. Similarly, some pseudorandomness of the unitary dynamics is also critical here since if $\Qgrav$ is a fixed (efficient) unitary, then the quantum algorithm can simply apply $\Qgrav$ to $\ket{\psi}$. On a technical level, the results rely on a high differential entropy of $U$; in fact, the bounds presented above can be derived for a unitary drawn from any distribution with sufficiently high differential entropy. We explain in Section~\ref{sec:tech-overview} how such a bound is derived as part of proving these theorems. Thus even if a given chaotic system is not well modeled by any pseudorandom unitary, so long as its time evolution operator can be modeled from a distribution with high differential entropy, our results still hold for the system. 

To genuinely make contact with our original motivation of reproducing Hawking's analysis by coarse-graining over complexity, we must consider a further relaxation of Theorem~\ref{thm:intro-fidelity-pseudorandom}. Rather than requiring $\A$ to output a single linear operator, we would like to allow for a more general quantum operation (completely positive trace-preserving map -- CPTP) $\qop$ as its best guess for $\Qgrav$. This would allow $\qop$ to be a decohering channel. Our learning results can accommodate this generalization: we explicitly prove the fidelity bound for any distribution of unitaries for $\Qgrav$ (which may be taken to be e.g. the distribution of a (pseudo)random unitary) based on differential entropy. We thus prove the following very general result, in which we do not assume that the unitary dynamics are drawn from any particular distribution, and we permit the quantum algorithm to output any quantum channel as its best guess for the time evolution unitary. Although our focus in this work is on the hardness of learning quantum gravity dynamics modeled by (pseudo)random unitaries, we are able to prove a bound for a general distribution of unitaries, which may be of interest both in physics and in learning theory (e.g. to model the dynamics of other systems).

\begin{manualtheorem}{3}[Hardness of Learning for Algorithms Predicting with Quantum Channels]\label{thm:intro-fidelity-general-learner}
Let $\Hil$ be any Hilbert space of dimension $\Hd$.
Let $\Qgrav$ be a unitary sampled from any distribution of unitary operators  on $\mathcal{H}$ and let $\ket{\BH}$ be a Haar random state\footnote{With a similar statement for pseudorandom $\ket{\psi}$.} in $\Hil$.
For any quantum algorithm $\A$ that makes $\ell$ quantum queries to $\Qgrav$, learns a quantum channel denoted as $\qop:\Hil\to\Hil'$ mapping from $\Hil$ to any Hilbert space $\Hil'$, and given $\ket{\BH}$ the algorithm outputs $\qop(\ket{\BH}\bra{\BH})$, we have that for any fixed\footnote{By fixed, we mean the map $\mathcal{M}$ is independent of the instance of $\Qgrav$, $\ket{\BH}$, and $\qop$ that is sampled or output.} quantum channel $\mathcal{M}:\Hil'\to\Hil$: 
$$\underset{\Qgrav,\qop,
\ket{\BH}}{\mathrm{avg}}\[F\(\mathcal{M}\(\qop(\ket{\BH}\bra{\BH})\),\Qgrav\ket{\BH}\bra{\BH}\Qgrav^\dagger\)\]\leq 1-\frac{1}{8\pi e}\cdot\frac{\Hd^2}{e^{\frac{2}{\Hd^4}\[-h\(\left|\Qgrav\otimes\Qgrav\right|\)+2\ell \log_{2}d\]}},
$$
where $F$ is the fidelity between two density operators, and $h(|\Qgrav\otimes\Qgrav|)$ is the joint differential entropy of the elements of $\Qgrav\otimes\Qgrav$ after taking their norms.
\end{manualtheorem}
Note that the bound in Theorem~\ref{thm:intro-fidelity-general-learner} is not in the same format -- i.e. $1-\Omega(1)$ -- as the previous two theorems. This is inevitable for a general distribution of $U$; Theorem~\ref{thm:intro-fidelity-general-learner} can most accurately be thought of as the generalization of  an intermediate result derived within the proof of Theorem~\ref{thm:intro-fidelity-random}, that for any distribution of $\Qgrav$ (not necessarily Haar random), the average fidelity is $\leq 1-e^{h(|U|)...}$ (note that $h(|U|)$ can be, and is, negative in our setting). In Theorem~\ref{thm:intro-fidelity-random} we instantiate this intermediate bound for Haar random $\Qgrav$ using their properties to obtain the $1-\Omega(1)$ bound (Theorem~\ref{thm:intro-fidelity-pseudorandom} for pseudorandom $\Qgrav$ follows thereafter). The parallelism that exists between our proof techniques for Theorems~\ref{thm:intro-fidelity-random}-\ref{thm:intro-fidelity-pseudorandom}, and our proof techniques for Theorem~\ref{thm:intro-fidelity-general-learner}, suggests that, by using these properties mutatis mutandis for Haar random and pseudorandom unitaries to instantiate Theorem~\ref{thm:intro-fidelity-general-learner}, the same also holds for algorithms predicting with a quantum channel $\qop$.

This final result raises an interesting question: under what circumstances is the output of $\A$ a general quantum operation $\qop$ which is not a unitary, or more generally not a single linear operator? More specifically, under what assumptions on $U$ is the output of $\A$ a decohering channel? Our results allow for the \textit{possibility} that $\A$ outputs a decohering channel; when is that possibility actually realized? It may be sufficient for $U$ and $\psi$ to be pseudorandom, or there may be additional conditions on $U$ without which $\A$ may output $O$ or even another unitary $U_{A}\neq U$. We leave this question for future work, remarking only that the additional conditions on $U$ (if any) that are both sufficient and necessary for a decohering output $\qop$ are likely to be required for a maximally accurate quantum code model of Hawking's calculation as a coarse-graining in the fundamental description.

Let us briefly comment on the relation to the non-isometric codes of~\cite{AkeEng22}. In those code models, the primary (though not exclusive) use of complexity bounds was to restrict the regime of validity of the effective, semiclassical description to subexponentially complex measurements. As a proof of principle that such a restriction can be successful in resolving critical conflicting aspects of the black hole information problem,~\cite{AkeEng22} constructed explicit models of non-isometric maps encoding the effective description in the fundamental description. These authors then showed that within these code models, complexity coarse-graining in the fundamental description does reproduce Hawking's entropy calculation via the simple entropy. Our work here is a more general question about generic reconstruction of the semiclassical calculation via algorithms subject to complexity coarse-graining.

Finally, let us note a caveat on applying our results to the gravitational context: this is admittedly a toy model and does not incorporate numerous aspects of gravity (e.g. diffeomorphism invariance). Nevertheless, it serves as a bare-bones model for the role computational complexity may play in semiclassical gravitational dynamics.

\paragraph{Organization.} In Section~\ref{sec:subsecpriorwork} we discuss relations of our results with prior works in learning theory. In Section~\ref{sec:tech-overview} we give a technical overview (walk-through) of the proofs of our main results. The remaining sections contain the formal definitions and proofs. In Section~\ref{sec:prelim} we define the learning model and state preliminaries. In Section~\ref{sec:fidelity-upperbd-random} we state the full version of Theorem~\ref{thm:intro-fidelity-random} and give its proof, parts of which are placed in Appendix~\ref{sec:appendix-tech} for readability. Section~\ref{sec:fidelity-upperbd-pseudorandom} is on Theorem~\ref{thm:intro-fidelity-pseudorandom}: in Section~\ref{sec:fidelity-defn-PRU-PRS} we give the definitions of pseudorandom unitaries and states, prove lemmas in Sections~\ref{sec:fidelity-PRU} and~\ref{sec:fidelity-PRS} which we use in Section~\ref{sec:PR-thm-proof} to prove Theorem~\ref{thm:intro-fidelity-pseudorandom}. Finally in Section~\ref{sec:general-learner} we state and prove Theorem~\ref{thm:intro-fidelity-general-learner} for algorithms predicting with quantum channels. We give further context and discuss implications of our results throughout.

\subsection{Relation to Prior Work in Learning Theory}\label{sec:subsecpriorwork}

The capabilities of quantum algorithms and the potential for quantum advantage in machine learning is of active interest in fields of computing (see e.g.~\cite{HuaBro21} for a recent discussion). We note that our technical results may be extended to prove that it is hard for any quantum algorithm to learn an unknown \emph{shallow} quantum circuit (an analogue of the results of~\cite{AruGri21} for learning classical shallow circuits). It suffices to construct a shallow pseudorandom unitary or to simply use circuits for producing pseudorandom states (perhaps~\cite{JiLiu18,BraShm20,BouFef22}) in place of the pseudorandom unitary in Theorem~\ref{thm:intro-fidelity-pseudorandom} and consider quantum queries to this circuit.

Below we discuss the relation of our results to recent and prior works in learning theory, and comment on (other) implications of our results for computer science.
\paragraph{Relation to the learning algorithm in~\cite{HuaChe22}.} From the theoretical quantum machine learning end, the very recent work of~\cite{HuaChe22} shows, positively, the capabilities of quantum machine learning algorithms for predicting observables of an unknown quantum process. Their suggested applications include efficiently predicting the outcomes of physical experiments and simulating complex systems faster than real-time. They provide an efficient quantum algorithm that, for any unknown quantum process $\cal{E}$ and state $\rho$ from any ``locally flat'' distribution, can learn to predict any ``bounded-degree'' (e.g. local) observable of $\cal{E}(\rho)$ (on average over $\rho$). This includes exponentially complex $\cal{E}$, e.g. corresponding to chaotic dynamics, and a wide range of distributions for $\rho$. Although our conceptual idea of modeling semiclassical gravity as a learning algorithm is unrelated to~\cite{HuaChe22}, our technical result is complementary in nature to their result. We show that when $\cal{E}$ and $\rho$ are pseudorandom as in a chaotic system,  
no efficient quantum learning algorithm is able to predict a state that agrees, within $o\(\frac{1}{\log\dim\cal{H}}\)$, with $\cal{E}(\rho)$ for every POVM (generalized measurement). Thus while for highly complex processes $\cal{E}$ and a wide range of states $\rho$, the learning algorithm of~\cite{HuaChe22} is able to predict many (``bounded-degree'') observables, we give evidence against the ability of any efficient learning algorithm to predict for every observable.\footnote{To accommodate observables of high complexity in our model, we consider the algorithm to predict a state and then can consider, without efficiency considerations, the observable on the predicted state vs. on the fundamental state.}

\medskip
\cite{AkeEng22} and~\cite{HuaChe22} together with this article illustrate that a complexity cutoff for semiclassical gravity would simultaneously permit accurate computations of the expectation values of general classes of observables, while forbidding the predictions of arbitrary, high complexity observables when the system is chaotic. This together supports the hypothesis of a complexity cutoff: the former permits the validity of semiclassical gravity predictions in the cases where it is expected to be accurate; the latter is precisely where semiclassical gravity fails to match the actual evolution of the system even approximately, e.g. in the black hole information paradox. 

\paragraph{Prior works in learning theory.}
We do not know of any prior works proposing our conceptual idea -- modeling the learning of effective physical dynamics, e.g. semiclassical gravity, as a computationally bounded learning algorithm and showing its limitation in learning fundamental dynamics, e.g. quantum gravity.

Here we discuss works within learning theory related to our technical results. Recently, Arunachalam, Grilo, and Sundaram~\cite{AruGri21} proved that it is hard for quantum algorithms to learn \emph{classical} (Boolean) functions, extending the proof techniques of~\cite{Val84,Kha93,Kha92} to quantum learners. They then transitioned from random Boolean functions to pseudorandom functions to show the hardness of learning $\mathsf{AC}^0$ and $\mathsf{TC}^0$ circuits. Beyond theoretical interest, a motivation of their work was to show the limitations of quantum machine learning for learning the weights of neural networks (a general class of which can be implemented by $\mathsf{TC}^0$ circuits). The high level structure of our proof follows~\cite{AruGri21}; our contribution is in overcoming the challenges in working with quantum operators as the object to learn. We give an overview of the technical ideas in Section~\ref{sec:tech-overview}.

Most of the literature focuses on learning classical functions. To our knowledge, the works on learning quantum operators are in quantum state tomography: learning or reconstructing an unknown quantum state from measurement results on many copies of the state~\cite{Aar06}. Our learning model and results focus on learning quantum operators, and allows the learner to make arbitrary queries and post processing. See~\cite{HuaChe22} and the references therein for the literature on quantum (process) tomography.
Additionally,~\cite{ChuLin18} show that polynomially many samples suffice to learn polynomial-size quantum circuits. Our hardness results are complementary: random unitaries are inefficient to generate, and while pseudorandom unitaries are efficiently computable, we are then interested in the time (not sample) complexity of the learner.

\paragraph{Implications of our results for learning theory and cryptography.}
As is frequently the case in the interdisciplinary field of quantum gravity and quantum information, progress motivated by one leads to new insights about its counterpart. Our results may also be considered within quantum learning theory with no reference to the role and motivation that it has in physics. We prove the hardness of learning (pseudo)random unitaries acting on (pseudo)random states. This shows that if there exists a pseudorandom unitary construction, then there are efficient quantum processes that cannot be efficiently learned. This gives evidence that for chaotic systems, quantum machine learning has limitations in learning and predicting.
We note that our results also have relevance to quantum cryptography: we show that pseudorandom unitaries have a concrete ``real-world'' security guarantee from the otherwise abstract  pseudorandomness guarantee.\footnote{The pseudorandomness guarantee is that the adversary (algorithm) cannot distinguish between query access to the pseudorandom unitary, or to a hypothetical instantiation of a Haar random unitary. This is a ``real-ideal'' guarantee between the real-world pseudorandom instantiation and a hypothetical ``ideal world'' Haar random instantiation. Such guarantees are used in cryptography to give flexible and robust security definitions (see~\cite{Lin16} for some context). Theorem~\ref{thm:intro-fidelity-pseudorandom} gives a concrete unpredictability guarantee for the real-world pseudorandom instantiation alone (using, in part, this ``real-ideal'' guarantee).} 
Thus pseudorandom unitaries, which may be used in place of classical pseudorandom functions on quantum computers, are (1) efficiently computable, and (2) have outputs which cannot be predicted by quantum adversaries, even if they can make adaptive queries to the unitary themselves.

\section{Technical Overview}\label{sec:tech-overview}

In this section we give an overview of the proofs of Theorems~\ref{thm:intro-fidelity-random}-\ref{thm:intro-fidelity-general-learner}, showing the hardness of learning (pseudo)random unitaries. This section is intended to be a walk through of how we derive our results, familiarizing readers with the main tools and techniques that we use in the formal proofs of Theorems~\ref{thm:intro-fidelity-random}-\ref{thm:intro-fidelity-general-learner}. The purpose of this overview is pedagogical and explanatory, and we thus focus on an initial simpler result here -- deriving a fidelity bound of $1-\frac{1}{\poly(\log\Hd)}$ for any algorithm that aims to learn a Haar random unitary $\Qgrav$ and produces a single linear operator $\SCgrav$ as its hypothesis for $\Qgrav$. (Note that an upper bound of $1-\frac{1}{\poly(\log\Hd)}$ already establishes the hardness of learning for such algorithms.)
The derivation of this initial result allows us to give an illustration of our primary techniques. Further ideas and techniques are used to derive Theorems~\ref{thm:intro-fidelity-random}-\ref{thm:intro-fidelity-general-learner}, including the $1-\Omega(1)$ bound, for which we only include brief high level explanations here. Proofs of our full results Theorems~\ref{thm:intro-fidelity-random}-\ref{thm:intro-fidelity-general-learner} are in Sections~\ref{sec:fidelity-upperbd-random}-\ref{sec:general-learner} and Appendix~\ref{sec:appendix-tech}.

We focus on the proof of Theorem~\ref{thm:intro-fidelity-random}, which is where most of the technical work lies. We then briefly explain how to prove Theorem~\ref{thm:intro-fidelity-pseudorandom} and describe how bounding the fidelity shows there exist operators that distinguish the semiclassical prediction from the quantum gravity evolution of a system. We remark that even if one is only interested in proving the hardness of learning for pseudorandom unitaries (e.g. as an \emph{efficiently} computable model of scrambling in quantum gravity), the definition of pseudorandom is that it is computationally indistinguishable from Haar random. To use the pseudorandom property, it is natural, perhaps necessary, to first show the result for a Haar random unitary.

Our proof of Theorem~\ref{thm:intro-fidelity-random} follows the blueprint by Arunachalam, Grilo, and Sundaram used in their work to show the hardness of quantum algorithms learning \emph{classical} functions (Lemma 4.1 in~\cite{AruGri21}), which in turn builds on~\cite{Kha92,Kha93} considering classical learning algorithms. The bulk of our work here will be in carrying out this blueprint for quantum operators (the quantum analog of a function).

Since the proof for classical functions is simpler, we include a walk through of~\cite{AruGri21}'s proof (modifying the presentation to create analogies with our quantum setting) at the end of this overview. The reader is welcome to refer to it to help clarify the structure of our proof.

Central to our results is a notion of entropy defined for real variables called the differential entropy. To use this notion for the fundamental time evolution operator $\Qgrav$ and the operator ($O$ or $\qop$) that the learning algorithm predicts with, we will implicitly use fixed basis for the underlying Hilbert space. The choice of basis does not matter for our results; it simply allows us to refer to these operators via matrix representation.

\paragraph{Learning quantum operators.}

Consider two quantum algorithms $\calR$ and $\A$ where $\calR$ is given a Haar random unitary $\Qgrav\in\C^{\Hd\times\Hd}$ (unknown to $\A$) and $\calR$ and $\A$ interact sending quantum information (qubits) back and forth. Specifically, this allows $\A$ to query $\Qgrav$ on quantum states $\ket{\phi_i}$ of its choice and obtain $\Qgrav\ket{\phi_i}$.
We consider the computational complexity of $\A$ and its communication complexity -- how many qubits are exchanged between $\calR$ and $\A$ for $\A$'s queries.
$\A$'s task is to ``learn'' $\Qgrav$ by producing a hypothesis $\SCgrav\in\C^{\Hd\times\Hd}$.
Thus we measure $\A$'s success probability by asking it to predict $\Qgrav\ket{\BH}$ where $\ket{\BH}$ is a Haar random state.
We give the following intuition which is more precise in the setting of learning classical Boolean functions but still gives some guidance here. Since the state $\ket{\BH}$ is random, to succeed, $\A$'s hypothesis $\SCgrav$ should be close to $\Qgrav$ on most of the Hilbert space of states (with respect to the Haar measure). Since $\Qgrav$ consists of exponentially many unknown entries, to learn a good approximation, $\A$ should need to query $\Qgrav$ on a set of states that spans almost the complete Hilbert space which would require exponentially many queries.

$\A$ now has quantum oracle access to $\Qgrav$ which allows it to query for any state $\ket{\phi_i}$ and obtain $\Qgrav\ket{\phi_i}$. Physically this corresponds to time evolving the chosen state $\ket{\phi_i}$ under the dynamics $\Qgrav$ and allows $\A$ to do arbitrary measurements and post processing on the evolved system. This is our model for querying and learning $\Qgrav$. $\A$ is asked to produce a circuit implementation of a linear operator $\SCgrav$. We define its success probability as the average fidelity $\left|\bra{\BH}\SCgrav^\dagger\Qgrav\ket{\BH}\right|^2$ between its prediction $\SCgrav\ket{\BH}$ and the state $\Qgrav\ket{\BH}$ for a chaotic system $\ket{\BH}$ modeled as a Haar random state. We allow $\SCgrav$ to be any matrix where the square of its column norms is bounded above by 1, allowing it to be highly nonunitary. Formally, let $\mathcal{H}$ be a Hilbert space of dimension $\Hd=2^\Hn$ where $\Hn$, the number of qubits, equals the entropy $S$ of $\ket{\psi}$. For any $\SCgrav$ whose column norms are bounded above by $\sqrt{\alpha}\leq 1$, the definition of failure is as follows:
\be\label{eqn:techover-quantum-pr-general}
\xi:=\Pr_{\substack{\Qgrav\gets\Haar\\ \ket{\BH}\gets\Haar}}\[\left|\bra{\BH}\SCgrav^\dagger\Qgrav\ket{\BH}\right|^2\]\leq \frac{1+\alpha}{2}-\frac{1}{\poly(\log\Hd)}
\ee 
where $\Haar$ is the Haar measure and the probability is over random $\Qgrav$ and $\ket{\BH}$. When the norm of the columns of $\SCgrav$ are bounded above by $1+o(\frac{1}{\log d})$, this reduces to the bound: 
\be\label{eqn:techover-quantum-pr}
\xi:=\Pr_{\substack{\Qgrav\gets\Haar\\ \ket{\BH}\gets\Haar}}\[\left|\bra{\BH}\SCgrav^\dagger\Qgrav\ket{\BH}\right|^2\]\leq 1-\frac{1}{\poly(\log\Hd)}.
\ee

\medskip\noindent
We now proceed to describe our proof, which will ultimately give a \textit{strictly stronger} $1-\Omega(1)$ result than the necessary bound above. We focus on the proof of the bound, Eq.~\ref{eqn:techover-quantum-pr-general} above, and then discuss our stronger $1-\Omega(1)$ bound. 
The guide for this proof will be to use a communication complexity bound on $\A$'s queries to $\Qgrav$, shown in~\cite{AruGri21} for such protocols, to formally bound the (mutual) information $\A$ can learn about $\Qgrav$. To upper bound the average fidelity, we will want to relate it to the mutual information in order to use the communication complexity bound. To do so, we first show the average fidelity can be related to the entropy remaining in $\Qgrav$ that $\A$ does not learn, i.e. the conditional entropy $h(\Qgrav|\SCgrav)$.

\paragraph{Averaging over a random quantum state.} First we relate the average fidelity to the elements of $\SCgrav$ and $\Qgrav$ as matrices. \footnote{In the proof for classical functions, such a relation, Equation~\ref{eqn:techover-classical-avg-inputs}, holds trivially by the independence of the inputs. However for quantum operators, the inputs are ``dependent'': on input a random state $\ket{\BH}$, the output $\Qgrav\ket{\BH}$ depends on all the columns of the operator.}

We fix $\SCgrav$ and $\Qgrav$ and take the average over a Haar random state $\ket{\BH}$. By calculation and using a Haar integral result from~\cite{Har13},

\begin{align*}
\Eover{\ket{\BH}\gets\Haar}\[\left|\bra{\BH}\SCgrav^\dagger \Qgrav\ket{\BH}\right|^2\]
\approx\frac{1}{\Hd(\Hd+1)}\Tr\[(\SCgrav^\dagger \Qgrav\otimes (\SCgrav^\dagger \Qgrav)^\dagger)\].
\end{align*}
where $\Eover{A \gets B}$ denotes the average over $A$ sampled from $B$. We use an approximation here to simplify the calculations in this overview (so the equations in what follows will be approximations).
By calculation, the trace here equals the squared inner product of $\SCgrav$ and $\Qgrav$ as matrices. Using this and taking the expectation over $\Qgrav$ and $\SCgrav$ from our context,
\be\label{eqn:techover-quantum-exp-F-expr}
\Eover{\substack{\Qgrav\gets\Haar\\ \SCgrav\gets\A^\Qgrav\\ \ket{\BH}\gets\Haar}}\[\left|\bra{\BH}\SCgrav^\dagger \Qgrav\ket{\BH}\right|^2\]\leq \frac{1}{\Hd}\Eover{\substack{\Qgrav\gets\Haar\\ \SCgrav\gets\A^\Qgrav}}\left|\sum_{i,j\in[\Hd]} \SCgrav^*_{ij}\Qgrav_{ij}\right|.
\ee

\paragraph{Translating from inner product to conditional entropy.}
Now we seek to translate from this algebraic expression to information theoretic quantities in order to use the communication complexity bound on mutual information.
We will use a result about estimators\footnote{In~\cite{AruGri21}, Fano's inequality is used but it only holds for discrete random variables. Here we find an analogous inequality for continuous variables (note that this is not identical to the `quantum Fano' inequality).}, which are used in statistics to calculate an estimate of a quantity $X\in\R$ using observed data $Y$.
Estimators in statistics have a similar function to decoders in coding and information theory. The `estimator error bound' states that the mean squared error of any estimator is lower bounded by the conditional \emph{differential} entropy\footnote{The differential entropy $h$ is an analog of Shannon entropy for continuous random variables. It does not have some of the properties of Shannon entropy including nonnegativity which is why $e^h$ often appears instead of $h$.} in $X$ given $Y$, $h(X|Y)$. For us, the quantity to estimate or decode will be (the norm of) $\Qgrav_{ij}$ where $\Qgrav\gets\Haar$ and the observed data will be the matrix $\SCgrav$ from the learner $\A$. We will consider the estimator that simply outputs (the norm of) $\SCgrav_{ij}$. The estimator error bound gives us
$$\E\[\(|\Qgrav_{ij}|-|\SCgrav_{ij}|\)^2\]\geq \frac{1}{2\pi e}e^{2h(|\Qgrav_{ij}||\SCgrav)}$$
which is equivalent to the following upper bound on $|\SCgrav_{ij}||\Qgrav_{ij}|$,
$$\frac{1}{2}\E\[|\Qgrav_{ij}|^2\]+\frac{1}{2}\E\[|\SCgrav_{ij}|^2\]-\frac{1}{4\pi e}e^{2h(|\Qgrav_{ij}||\SCgrav)}\geq \E\[|\SCgrav_{ij}||\Qgrav_{ij}|\].$$
Since the column norms of $\SCgrav$ and $\Qgrav$ are bounded, we can sum this inequality for every $i,j\in[\Hd]$ and obtain an upper bound on the inner product between $\SCgrav$ and $\Qgrav$ and then use Jensen's inequality. This translates Equation~\ref{eqn:techover-quantum-exp-F-expr} into an information theoretic bound,
\begin{equation}\label{eqn:techover-quantum-EF-bound-single-e^h}
\Eover{\substack{\Qgrav\gets\Haar\\ \SCgrav\gets\A^\Qgrav\\ \ket{\BH}\gets\Haar}}\[\left|\bra{\BH}\SCgrav^\dagger \Qgrav\ket{\BH}\right|^2\]\leq \frac{1+\SCcol}{2}-\frac{\Hd}{4\pi e}\cdot e^{\frac{2}{\Hd^2}\sum\limits_{i,j\in[\Hd]} h\(|\Qgrav_{ij}||\SCgrav\)}.
\end{equation}

Here $\sqrt{\SCcol}$ is the maximum norm of $\SCgrav$'s columns (for simplicity in reading this overview, the reader may assume $\SCgrav$ is unitary and $\SCcol=1$).
Unlike Shannon entropies, differential entropies can be negative (and conditioning will make them more negative)\footnote{Differential entropies are also not invariant under scaling. If the $\Hd$ in the numerator is puzzling: a rough approximation for $h(|\Qgrav_{ij}|)$ is $-\frac{\log\Hd}{2}$.}. To bound the average fidelity below $1$, we need a lower bound on how negative the sum of entropies is. This corresponds with what we intuitvely want to show: the entropy of $\Qgrav$, modeling the quantum gravity time evolution operator, is still significant even
given any $\SCgrav$ from any bounded quantum algorithm.

\paragraph{Lower bounding the entropy of $\Qgrav$.}
For now, let us forgo considering how taking the absolute value of $\Qgrav_{ij}$ and how conditioning by $\SCgrav$ affects the entropies $h\(|\Qgrav_{ij}||\SCgrav\)$ and try to lower bound $\sum\limits_{i,j\in[\Hd]} h(\Qgrav_{ij})$. This is at least the differential entropy of a Haar random unitary $h(\Qgrav)$.
We note that for a random classical function $f:\zo^\Hn\to\zo$, it is simple to calculate the analogous entropy quantity $H(f)=2^\Hn$ where $H$ is the Shannon entropy for discrete variables. For a random quantum unitary, $h(\Qgrav)$ seems like a reasonable quantity to ask about, but we do not know of any lower bounds on $h(\Qgrav)$ nor the explicit probability density function of a Haar random unitary to calculate its entropy.

We can, however, think about how to generate a random unitary to get a sense of its distribution. Generating a matrix of $\Hd^2$ random (complex) Gaussians and orthogonalizing (and normalizing) its columns using the Gram-Schmidt procedure produces a Haar random unitary. Intuitively, this orthogonalization should not remove many degrees of freedom so we would expect the entropy of a random unitary to be close to the entropy of a matrix of Gaussians.
Fortunately we can calculate the entropy of the latter: since the entries are \emph{independent} Gaussian variables, the entropy of the matrix is simply the sum of the entropies of its elements.

\paragraph{Overcoming dependence of entries using a random matrix approximation.}
This leads us to consider if we can use the matrix of Gaussians to somehow lower bound $h(\Qgrav$). We use the result of~\cite{Jia10} that many entries of a random unitary can be well approximated by complex Gaussian variables. Formally,~\cite{Jia10} shows that for a $\Hd\times\Hd$ matrix $\CN=\(\CN_{ij}\)$ of complex Gaussians with variance $\frac{1}{2\Hd}$, and the unitary $U=\(U_{ij}\)$ produced by orthonormalizing $\CN$, the difference between $\Qgrav_{ij}$ and $\CN_{ij}$ is small for $\md(\Hd)=\Theta\(\frac{\Hd}{\log\Hd}\)$ columns, 
$$\Pr\[\max_{i\in[\Hd],j\in[\md(\Hd)]}\left|\Qgrav_{ij}-\CN_{ij}\right|\leq O\(\frac{1}{\sqrt{\Hd}}\)\]\geq 1-O\(\frac{1}{\Hd}\).$$
In this overview, we will make the convenient simplifying assumption that \emph{with probability 1}, these entries $\Qgrav_{ij}$ are \emph{equal} to $\CN_{ij}$.
Our goal in this overview is to present a simplified outline of the proof so we will distill away many of the technical steps and calculations. We refer the reader to Section~\ref{sec:fidelity-upperbd-random} for the full proof.
We return to rederive Equation~\ref{eqn:techover-quantum-EF-bound-single-e^h}. 

Throwing away the entropies of the entries we cannot approximate (by using $e^{h}\geq 0$) and replacing the remaining $\Qgrav_{ij}$ by $\CN_{ij}$,
\begin{align}
\Eover{\substack{\Qgrav\gets\Haar\\ \SCgrav\gets\A^\Qgrav\\ \ket{\BH}\gets\Haar}}\[\left|\bra{\BH}\SCgrav^\dagger \Qgrav\ket{\BH}\right|^2\]
&\leq \frac{1+\SCcol}{2}-\md(\Hd)\cdot\frac{1}{4\pi e}\cdot e^{\frac{2}{\Hd\cdot\md(\Hd)}h\(\{|\CN_{ij}|\}_{i\in[\Hd],j\in[\md(\Hd)]}\middle|\SCgrav\)}.\label{eqn:techover-quantum-E-F-approx-e^h}
\end{align}
We now arrive at lower bounding $h\(\{|\CN_{ij}|\}\middle|\SCgrav\)$ but  this time, better equipped. By the independence of $\CN_{ij}$, we have $h\(\{|\CN_{ij}|\}_{i\in[\Hd],j\in[\md(\Hd)]}\)=\Hd\cdot\md(\Hd)\cdot h_\CN$ where $h_\CN$ is the entropy of a single complex Gaussian (after taking its norm). All that remains is to consider how conditioning by $\SCgrav$, the model that $\A$ learns (e.g. semiclassical gravity), can affect the entropy.

\paragraph{Entropy is retained after bounded queries (bounded $\QCC$).}
Recall the intuition is if $\A$ makes a bounded number of queries to $\Qgrav$ as it produces $\SCgrav$, then most of the entropy in $\Qgrav$ should remain. The entropy loss from knowing $\SCgrav$ is the mutual information between $\{|\CN_{ij}|\}_{i\in[\Hd],j\in[\md(\Hd)]}$ and $\SCgrav$, which we show is at most the mutual information $I(\CN;\SCgrav)$ between the matrices $\CN$ and $\SCgrav$, by using a few information theory lemmas.
\cite{AruGri21} shows that a corollary of~\cite{Tou15,KerLau16} is that the number of qubits communicated for $\A$'s queries (the quantum communication complexity\footnote{If $\A$ makes $\ell$ queries, then $\QCC=\ell\cdot 2\Hn=\ell\cdot 2\log\dim{\cal H}$.} $\QCC$) upper bounds $I(\CN;\SCgrav)$. Together, we have our desired lower bound
\begin{align*}
h\(\{|\CN_{ij}|\}_{i\in[\Hd],j\in[\md(\Hd)]}\,\middle|\,\SCgrav\)&\geq h\(\{|\CN_{ij}|\}_{i\in[\Hd],j\in[\md(\Hd)]}\)-\QCC\\
&=\Hd\cdot\md(\Hd)\cdot h_\CN -\QCC
\end{align*}
where, by explicit calculation, the differential entropy of a complex Gaussian variable is $h_\CN=1-\log 2+\frac{\gamma_E}{2}-\frac{\log\Hd}{2}$ where $\gamma_E$ is Euler's constant. Using this lower bound in Equation~\ref{eqn:techover-quantum-E-F-approx-e^h} and that we can approximate $\md(\Hd)=\Theta\(\frac{\Hd}{\log\Hd}\)$ columns,
\begin{align}
\Eover{\substack{\Qgrav\gets\Haar\\ \SCgrav\gets\A^\Qgrav\\ \ket{\BH}\gets\Haar}}\[\left|\bra{\BH}\SCgrav^\dagger \Qgrav\ket{\BH}\right|^2\]
&\leq \frac{1+\SCcol}{2}-\Theta\(\frac{1}{\log\Hd}\)\cdot\frac{e^{1+\gamma_E}}{16\pi}\cdot e^{\frac{-2\cdot\QCC}{\Hd\cdot\md(\Hd)}}.
\end{align}
For $\QCC$ at most $o\(\frac{\Hd^2}{\log\Hd}\)$ and $\SCcol=1+o\(\frac{1}{\log\Hd}\)$, the average fidelity (left side) is upper bounded by $1-\Omega\(\frac{1}{\log\Hd}\)$.
This establishes an initial bound in the fidelity for random unitaries which (already) suffices to show the hardness of learning random unitaries. Here it was crucial that the number of columns $\md(\Hd)$ which can be approximated by complex Gaussians is at least $\frac{\Hd}{\poly\log\Hd}$. As advertised in Section~\ref{sec:intro}, we actually obtain a stronger bound. Let us briefly describe the modus operandi behind the improved bound. 

\paragraph{For a fidelity bound of $1-\Omega(1)$.}
We arrived at Equation~\ref{eqn:techover-quantum-E-F-approx-e^h} by discarding the entropy contributions of the entries of $\Qgrav$ that we cannot approximate by complex Gaussians. While it is true that for a large subset of $\Qgrav$'s columns, the column vectors are not independent, we can show that \emph{any subset} of $\md(\Hd)=\Theta\(\frac{\Hd}{\log\Hd}\)$ columns are in fact independent vectors of complex Gaussians. This uses the  translation invariance of the Haar measure. Armed with the freedom to choose any set of $\md(\Hd)$ vectors, we partition $\Qgrav$ into $\lfloor\frac{\Hd}{\md(\Hd)}\rfloor$ sets of $\md(\Hd)$ columns each. We can then approximate the \emph{marginal} distribution of each subset of columns by complex Gaussians and show that the algorithm $\A$ cannot learn the columns in a \emph{single} subset well. The degree to which $\A$ can learn the entries of $\Qgrav$ is linear in the amounts that it can learn about each subset, so we can combine the hardness of learning these $\lfloor\frac{\Hd}{\md(\Hd)}\rfloor$ subsets into the hardness of learning all of $\Qgrav$. By utilizing the entropy in (nearly) all of the entries in $\Qgrav$ via marginals, we obtain a $1-\Omega(1)$ bound as stated in Theorem~\ref{thm:intro-fidelity-random}.

\paragraph{A fidelity bound for general unitaries (and states).} This proof also gives an upper bound on the average fidelity for any distribution of $\Qgrav$ that has sufficiently high (i.e. less negative) differential entropy, not strictly Haar random $\Qgrav$. (Recall the differential entropy ranges over $\mathbb{R}$, both positive and negative.) The proof of this bound for general $\Qgrav$ with sufficiently high differential entropy is the same through Equation~\ref{eqn:techover-quantum-EF-bound-single-e^h} which we can then combine with $\QCC\geq I(\Qgrav;\SCgrav)$ to obtain

\begin{equation}
\Eover{\substack{\Qgrav\gets\Haar\\ \SCgrav\gets\A^\Qgrav\\ \ket{\BH}\gets\Haar}}\[\left|\bra{\BH}\SCgrav^\dagger \Qgrav\ket{\BH}\right|^2\]\leq \frac{1+\SCcol}{2}-\frac{\Hd}{4\pi e}\cdot e^{\frac{2}{\Hd^2}\(\,h\(|\Qgrav|\)-\QCC\,\)}.
\end{equation}
where $|\Qgrav|=\(|\Qgrav_{ij}|\)_{i,j\in[\Hd]}$.
This bound depends on the differential entropy of $|\Qgrav|$ and the query (communication) complexity of $\A$. The distribution of $\Qgrav$ needs to have sufficiently high differential entropy (lower bounded), and the communication complexity needs to be upper bounded in order to lower bound 
$h\(|\Qgrav|\)-\QCC$ and obtain a fidelity bound less than $1$. The Haar random state $\ket{\BH}$ can also be relaxed to a $2$-design since the derivation of Equation~\ref{eqn:techover-quantum-exp-F-expr} only depends on the degree-$2$ Haar integral of $\(\ket{\BH}\bra{\BH}\)$.

\paragraph{From random to pseudorandom.} 
Lastly we describe how we show Theorem~\ref{thm:intro-fidelity-pseudorandom}, the hardness of learning pseudorandom unitaries (PRU) applied to a pseudorandom state (PRS).
We mention that although most of the literature, until~\cite{BouFef19} and~\cite{KimTan20}, uses quantum $k$-designs  as models of quantum gravity and black holes, PRU and PRS are also suitable models which are specifically defined to be indistinguishable from Haar random to any algorithm bounded in computational complexity. This may be fitting for modeling the proposed regime of validity of semiclassical gravity with a complexity cutoff.

Our hardness result for random unitaries and states holds for any quantum algorithm $\A$ that makes a subexponential number of queries to $\Qgrav$. We now consider $\A$ with polynomial computational complexity (runtime). We will use the indistinguishability property of PRU and PRS from their Haar random counterparts to show Theorem~\ref{thm:intro-fidelity-pseudorandom}. We give a high level description of the argument here since it is standard within cryptography.

We first show the hardness of learning a PRU applied to a Haar random state, a `hybrid' distribution, and then replace the random state by a PRS to show the hardness of learning a PRU applied to a PRS. If we replace the Haar random unitary (RU) by a PRU, then any such $\A$ should not be able to tell the difference. If some $\A$ manages to learn better for a PRU (with non-negligible improvement), then we can use $\A$ to construct a distinguishing algorithm $\Dist$ that breaks and contradicts the indistinguishability of the PRU from a Haar random unitary. The $\Dist$ runs $\A$ and simulates the oracle for $\A$, i.e. whenever $\A$ queries, it uses the RU or PRU $\widetilde{\Qgrav}$ to compute the responses to $\A$. Then at the end it gives $\A$ a state $\ket{\BH}$ and receives $\SCgrav\ket{\BH}$. $\Dist$ also computes $\widetilde{\Qgrav}\ket{\BH}$. It performs the SWAP test between these two states which passes with a probability corresponding to the fidelity between the two states. If $\A$ learns better for a PRU, the $\Dist$'s SWAP test will pass with higher probability when $\widetilde{\Qgrav}$ is a PRU than when it is a RU. If the PRU ensemble is indeed pseudorandom, then no such $\A$ can exist. We can now replace the random state by a PRS and use a similar proof to show $\A$ again cannot learn any better. This is the outline of how we prove Theorem~\ref{thm:intro-fidelity-pseudorandom}.

\paragraph{A fidelity bound implies distinguishing operators.}
While a natural choice, the fidelity also has physical significance: if two states $\rho$ and $\sigma$ have fidelity $F$, then there exists a positive operator-valued measure (POVM) with measurement probabilities $(p_{1},p_{2},\ldots,p_{\ell})$ and $(q_{1},q_{2},\ldots,q_{\ell})$ for $\rho$ and $\sigma$ respectively such that $\left(\sum_{k\in[\ell]}{\sqrt{p_{k}q_{k}}}\right)^{2}=F$~\cite{NieChu11}. In other words, states with fidelity significantly less than $1$ have a quantum measurement that distinguishes them with some significant probability(see Appendix~\ref{sec:appendix-POVM} for an explanation of this relationship). By showing Equation~\ref{eqn:techover-quantum-pr} (and furthermore by showing a $1-\Omega(1)$ bound), we show that on average there exists a POVM that distinguishes semiclassical gravity's prediction $\SCgrav\ket{\BH}$ from quantum gravity's evolution $\Qgrav\ket{\BH}$ with a significant probability. This connects with the apparent inconsistency between semiclassical gravity's prediction for evaporating black holes and quantum unitarity, showing that if the algorithm computing the effective description is of bounded complexity, then some POVM will distinguish its predictions from those of the fundamental description.

\paragraph{For algorithms predicting with quantum channels.} In Theorem~\ref{thm:intro-fidelity-general-learner} we consider algorithms $\A$ that can produce a quantum channel $\qop$ and consider the average fidelity between its prediction (which may now be a mixed state) and $\Qgrav\ket{\BH}$. To show our upper bound on the average fidelity, by Kraus's theorem, we have an operator-sum representation of $\qop$ as a set of operators $\{E_k\}$ where $\qop(\rho)=\sum_{k} E_k\rho E^\dagger_k$. We show that the average fidelity between $\qop(\ket{\BH}\bra{\BH})$ and $\Qgrav\ket{\BH}\bra{\BH}\Qgrav^\dagger$ is a sum of the `fidelity' of each operator $E_k$ with $\Qgrav$. We apply the techniques and tools we used to show the hardness of learning when $\A$ produces a single operator $\SCgrav$ to each of the operators $E_k$ here, and use these bounds together to show Theorem~\ref{thm:intro-fidelity-general-learner} (see Section~\ref{sec:general-learner} for the full result).

\medskip
We conclude this section by giving a walk-through of the proof of the quantum hardness of learning classical functions in~\cite{AruGri21} since the classical case is simpler but still contains a blueprint and some of the tools that we use in our proof. We modify the presentation of this proof to draw analogies with our proof. See~\cite{AruGri21} for their motivation and context, and the original proofs.

\paragraph{The model for learning classical functions.}
Consider two quantum algorithms $\calR$ and $\A$ where $\calR$ is given a uniformly random Boolean function $f:\zo^\Hn\to\zo$ (unknown to $\A$) and $\calR$ and $\A$ interact sending quantum information (qubits) back and forth. Specifically, this allows $\A$ to query $\calR$ on quantum states of its choice and obtain $f$ applied in superposition on an (entangled) register\footnote{More formally, if $\A$ queries $\calR$ on $\sum_{x\in\zo^\Hn, b\in\zo}\alpha_{x,b}\ket{x}\ket{b}$ then it obtains $\sum_{x\in\zo^\Hn, b\in\zo}\alpha_{x,b}\ket{x}\ket{b\oplus f(x)}$.} (i.e. $\A$ has quantum oracle access to $f$).
We consider the computational complexity of $\A$ and its communication complexity -- how many qubits are exchanged between $\calR$ and $\A$.
$\A$'s task is to `learn' $f$ by producing a hypothesis $g:\zo^\Hn\to\zo$ such that $g(x)=f(x)$ for most $x\in\zo^\Hn$. Formally, we define $\A$ to successfully learn\footnote{We give a strong definition of successful learning to be analogous to Equation~\ref{eqn:intro-def-success}. In~\cite{AruGri21} and other works in learning theory, success is often defined as performing better than a random guess: $\Pr[g(x)=f(x)]\geq \frac{1}{2}+\frac{1}{\poly(\Hn)}$.} $f$ if the success probability
$$\xi:=\Pr_{\substack{f\gets\zo^{\Hd}\\ x\gets\zo^\Hn}}\[g(x)=f(x)\]\geq 1-\frac{1}{\exp(\Hn)}=1-\frac{1}{\exp(\log\Hd)}$$
where the probability is over the randomness of $f$ and $x$.
The functions $f$ and $g$ can be represented by their truth tables, vectors in $\zo^{\Hd}$ where $\Hd=2^\Hn$ (for $x\in\zo^\Hn$, the $x$'th element is $f_x=f(x)$).
Intuitively, since $\A$ does not know what $x$ it will be asked to predict for, to succeed, it should produce $g$ that agrees with $f$ on almost all inputs. Since $f$ consists of exponentially many bits of randomness, to learn most of $f$, $\A$ should need to query $\calR$ on almost every input $x\in\zo^\Hn$ requiring exponentially many qubits exchanged. (We omit directly addressing why the ability to query in superposition does not reduce the query complexity here.) Formally we would like to show that for any $\A$ making a subexponential number of queries,
$$\xi:=\Pr_{\substack{f\gets\zo^\Hd\\ x\gets\zo^\Hn}}\[g(x)=f(x)\]\leq 1-\frac{1}{\poly(\Hn)}=1-\frac{1}{\poly(\log\Hd)}.$$

\paragraph{Hardness of learning classical functions \cite{AruGri21}.}
To prove this,~\cite{AruGri21} use information theory (this accommodates the queries being quantum). By the above intuition, the correlations between $f$ and $g$ should be bounded because the information communicated between $\calR$ and $\A$ is bounded. Indeed~\cite{AruGri21} show that a corollary of \cite{Tou15,KerLau16} is that the number of qubits communicated between $\calR$ and $\A$ (the quantum communication complexity $\QCC$) upper bounds the mutual information $I(g;f)$ between $g$ and $f$, i.e. $\QCC\geq I(g;f)$. To use this to upper bound $\xi$, we need to interpret between $g$ and $f$'s mutual information and the probability of $g_x=f_x$. Fano's inequality is exactly a tool for doing so. By Fano's inequality, for \emph{discrete} random variables $f_x\in\zo$ and correlated random variable $g\in\zo^\Hn$, the conditional entropy
\be\label{eqn:techover-classical-Fano}
H(f_x|g)\leq H_{b}(\Pr[g_x\neq f_x])=H_{b}(\Pr[g_x= f_x])
\ee
where $H$ is the Shannon entropy and the function $H_b(p)=-p\log_{2}p-(1-p)\log_{2}(1-p)$ (note $H_b$ has maxima at $p=\frac{1}{2}$ and decreases to $H_b(1)=0$)\footnote{$H_b$ is the binary entropy function; $H_b(p)=H(X)$ where $X$ is a Bernoulli variable that is $0$ with probability $p$.}.
Next we translate between conditional entropy and mutual information: first, since $f_x$ are independent bits, $H(f|g)=\sum\limits_{x\in\zo^\Hn} H(f_x|g)$.
Next, $I(g;f)=H(f)-H(f|g)$ which equals $2^\Hn-H(f|g)$ since $f$ is a uniformly random vector in $\zo^{2^\Hn}$. Combining these steps,
$$I(g;f)=2^\Hn-\sum\limits_{x\in\zo^\Hn} H(f_x|g)\geq 2^\Hn-\sum\limits_{x\in\zo^\Hn} H_b(\Pr[g_x=f_x]).$$
If $\QCC$ is subexponential then $\sum\limits_{x\in\zo^\Hn} H_b(\Pr[g_x=f_x])\geq \Omega(2^\Hn)$.
Since $H_b(1)=0$, this shows that on average $\Pr[g_x=f_x]$ is bounded away from $1$ so the success probability
\be\label{eqn:techover-classical-avg-inputs}
\Pr_{\substack{f\gets\zo^\Hd\\ x\gets\zo^\Hn}}\[g(x)=f(x)\]=\frac{1}{2^\Hn}\sum\limits_{x\in\zo^\Hn} \Pr\[g_x=f_x\]
\ee
is bounded away from $1$. Hence if $\A$ makes subexponentially many queries to $f$, it cannot learn $f$. This concludes our review of the proof in \cite{AruGri21}.

\section{Preliminaries}\label{sec:prelim}
In this section, we cover some of the background used in this work, which may not be familiar to many of our readers given the interdisciplinary nature of our results. Additional background may be found in Appendix~\ref{sec:appendix-info-theory}.

\paragraph{Notation.}
We begin by defining the basic notation used throughout this work (other notation is defined as we use it). The set of integers from $1$ to $k$ is denoted by $[k]$ and the set of integers from $k$ to $\ell$ is denoted by $[k,\ell]$. We use $\log$ to denote the natural logarithm, unless we are taking the logarithm of $\dim\cal{H}$ for a Hilbert space $\cal{H}$ in which case $\log\dim\cal{H}$ refers to $\log_2\dim\cal{H}$. The expected value (or average) of a random variable $X$ is the Lesbegue integral of $X$ with respect to its probability measure and is denoted by $\Eover{}[X]$ (see Definition~\ref{defn:expected-value} for an explicit definition). When the probability distribution of $X$ depends on the joint distribution of variables $(Y_1,\ldots,Y_i)$, we may write $\Eover{Y_1,\ldots,Y_i}[X]$. Similarly, ${\displaystyle \Pr_{Y_1,\ldots,Y_i}}$ denotes taking the probability over the joint distribution of variables $(Y_1,\ldots,Y_i)$.

We model systems by quantum states of $\Hn:=S$ qubits where $\Hd=2^\Hn$ is the dimension of the corresponding Hilbert space ${\cal H}$.
We use $F(\rho,\sigma)$ to denote the fidelity between quantum (mixed) states (density operators) $\rho$ and $\sigma$:
${\displaystyle F(\rho ,\sigma )=\left(\Tr {\sqrt {{\sqrt {\rho }}\sigma {\sqrt {\rho }}}}\right)^{2}}$. When $\rho$ and $\sigma$ represent pure states $\ket{\psi}$ and $\ket{\phi}$, this reduces to their squared inner product: $F(\ket{\psi},\ket{\phi})=|\braket{\psi}{\phi}|^2$. The Haar measure on the $\Hd$ dimensional unitary group is denoted by $\Haar_\Hd$. We will also use $\Haar_\Hd$ to denote the distribution of $\Hn=\log_2\Hd$ qubit Haar random states in ${\cal H}$. We sometimes use $\Haar$ instead of $\Haar_\Hd$ as the dimension $\Hd$ is often clear from context.
The identity operator is denoted by $I$.

For a distribution $\cD$, we write $X\gets\cD$ to denote $X$ is sampled from $\cD$ and $X\sim\cD$ to denote $X$ is distributed according to $\cD$. For an algorithm $\A$, we will write $X\gets\A$ to denote $X$ is output by $\A$.
We use $X\sim\A^\Qgrav$ to refer to random variables $X$ that $\A^\Qgrav$ produces after the query and learn phase. Wee write $\A_n$ to denote the algorithm $\A$ on inputs of length $n$ (or on states of $n$ qubits if the input is quantum). We use the analogous notation for other algorithms e.g. for the algorithm $\Key$ and ``security parameter'' $\lambda\in\mathbb{N}$ in Section~\ref{sec:fidelity-upperbd-pseudorandom} we use $\Key_\lambda$.

In addition to $o$ and $O$, we will use the asymptotic notations of $\Theta$ and $\Omega$. For our purposes, a function $f(x)=\Theta(g(x))$ if there exists constants $\eta_1,\eta_2 >0$ such that asymptotically (i.e. in the limit of large $x$), $\eta_1 g(x)\leq f(x)\leq \eta_2 g(x)$. A function $f(x)=\Omega(g(x))$ if there exists a constant $\eta >0$ such that asymptotically, $f(x) \geq \eta g(x)$. 

We work with (absolutely) continuous random variables, i.e. random variables over $\R$ with a probability density function (PDF). Random variables is denoted by uppercase letters, e.g. $X$, their support by the corresponding calligraphic letter, $\mathcal{X}\subseteq \R$, and their PDF by $p_X(x):\mathcal{X}\to\R_{\geq 0}$.
When we have a set of jointly distributed random variables, e.g. $X,Y,Z$, we denote their joint PDF by $p_{X,Y,Z}(x,y,z)$. Marginal and conditional PDFs are denoted as usual: the marginal PDF of $X,Y$ is $p_{X,Y}(x,y):=\int_{\mathcal{Z}} p_{X,Y,Z}(x,y,z)\;dz$, the conditional PDF of $X,Y$ conditioned on $Z$ is $p_{X,Y|Z}(x,y|z):=\frac{p_{X,Y,Z}(x,y,z)}{p_{Z}(z)}$, and similarly for others.

\paragraph{Complex Gaussian variables.}
To show the hardness of learning a Haar random unitary, we use the result in \cite{Jia10} (see Section~\ref{sec:prelim-approx-unitary}) to approximate a random unitary matrix by a matrix of complex Gaussian variables. A \emph{complex Gaussian random variable} is a random variable over $\C$ with real and imaginary parts that are jointly normal. We let $\Ndistr(\mu,\sigma^2)$ denote the (real) Gaussian distribution with mean $\mu$ and variance $\sigma^2$. If $\CN_1,\CN_2\sim\Ndistr(\mu,\sigma^2)$, then $\CN:=\CN_1+i\CN_2\in\C$ is a complex Gaussian with distribution denoted by $\CNdistr(\mu,2\sigma^2)$.
We note that a complex-valued random variable is equivalent to a vector-valued random variable in $\R^2$. Thus we will use traditional results for continuous random variables over $\R$ for complex Gaussians.

\subsection{Learning Models and Communication Complexity}\label{sec:prelim-comp-models}

In this section we define the learning model and state the communication complexity result that we use in this work. The learning model builds upon the probably approximately correct (PAC) model introduced in~\cite{Val84}. We briefly describe the PAC model and mainly focus on describing our learning model which is a variant of the PAC model for \emph{learning quantum operators}. We refer to \cite{AruGri21} and \cite{AruWol17} for background in quantum learning theory (mainly for classical functions) and to~\cite{Ans22} for a short review on learning quantum states and operators. (For context, recall that in the physics setup of interest, $\A$ stands for a computationally-bounded reconstruction algorithm, $U$ for the chaotic quantum gravity dynamics, and the learning process is a toy model for how we as semiclassical observers might learn about quantum gravity dynamics under ignorance of high complexity information and processing.)

\paragraph{Probably approximately correct (PAC) model of learning~\cite{Val84}.} A class of Boolean functions $\mathcal{C}$ with respect to a distribution of inputs $\mathcal{D}$ is said to be $(\epsilon,\delta)$-PAC-learnable if there is an algorithm $\A$ such that for any $f\in\mathcal{C}$, $\A$ with access to an oracle that gives it random labeled examples $(x,f(x))$ for $x\gets\mathcal{D}$, with probability $\geq 1-\delta$, outputs a `hypothesis' Boolean function $h$ that agrees with $f$ with probability $\geq 1-\epsilon$ for a random $x\gets\mathcal{D}$.

\paragraph{A model for learning quantum operators.}
We generalize the PAC model to \emph{learning quantum operators} by allowing $\mathcal{C}$ to consist of unitary operators, $\cD$ to be a distribution of quantum states, and the learner to be a quantum algorithm. Then for a unitary operator $U\in\mathcal{C}$, $\A$ is given examples $(\ket{\psi},U\ket{\psi})$ where $\ket{\psi}\gets\cD$ and its goal is to output a unitary $\widehat{U}$ such that $\widehat{U}$ and $U$ agree with high probability for a random $\ket{\psi}\gets\cD$. We use the fidelity as the metric for agreement, as stated in~\ref{defn:learning-model}.
Our model builds upon the notion of PAC learning, though notably differs in the following ways: it is adapted for learning quantum operators (instead of Boolean functions). It also allows the learner to make arbitrary quantum queries as we now define:

\medskip
\medskip
We now define quantum queries and our quantum learning model.
\begin{defn}[Quantum Queries]\label{defn:quantum-queries}
Let $U$ be a unitary. A quantum algorithm $\A$ has \emph{quantum query access} to the `oracle' $U$, denoted $\A^U$, if $\A$ can obtain the result of $U$ applied to arbitrary quantum states of its choice that it can prepare. More concretely, $\A$ can be any quantum circuit as follows: $\A$ performs its computations, and when we say that it `queries' $U$ on any state $\ket{\phi_1}$ that it has in its registers, we mean that it obtains $U\ket{\phi_1}$ in that register. $\A$ can continue to perform any computations of its choice, can query $U$ on another state $\ket{\phi_2}$ in its registers, will then have $U\ket{\phi_2}$ in that register, and so forth until it reaches the end of its computation\footnote{In a concrete model of computation, there are components corresponding to the end of the computation and the output. In a circuit model, the computation proceeds via applying a sequence of gates to qubits. The computation stops when there are no more gates to apply. We can denote certain qubits as the output register and at the end, take those qubits to be the output (state).}.
\end{defn}

Our model allows the algorithm to make arbitrary quantum queries to the operator it is trying to learn (instead of only being given examples from a fixed input distribution). These queries were defined as `quantum membership queries' in \cite{AruGri21}.

We emphasize that in quantum computational models, a quantum algorithm receives quantum states as $\log\dim\Hil$ physical qubits, not as a vector in $\Hil$. Although naively it might appear that the quantum learning algorithm here can just query the oracle for some states $\ket{\phi_i}$, obtain $U\ket{\phi_i}$, and solve for a viable $\Qgrav$ via linear equations, the algorithm obtains and computes on physical states, not on their algebraic representation. This is essential for quantum algorithms in general to run in time $\poly(\log\dim\Hil)$ and \emph{not} $\poly(\Hd)$.

We note that if one wishes to consider a different form of quantum query,
then Theorem~\ref{thm:intro-fidelity-random} and Theorem~\ref{thm:intro-fidelity-general-learner} still hold as long as for an $\Hn$ qubit query, the oracle returns $O(\Hn)$ qubits. For Theorem~\ref{thm:intro-fidelity-pseudorandom} the PRU needs to remain indistinguishable from Haar random for the chosen form of query. 

\begin{defn}[Quantum Learning Model]\label{defn:learning-model}
Let $\mathcal{U}$ be a distribution of $\Hd\times\Hd$ unitaries and $\Psi$ be a distribution of $\Hd$ dimensional quantum states. A quantum algorithm `learner' $\A$ is said to $(\epsilon,\delta)$-learn $\mathcal{U}$ over $\Psi$ if given quantum query access to the oracle $U\gets\mathcal{U}$, $\A$ produces a hypothesis, a quantum channel $\qop$, such that with probability $\geq 1-\delta$,
$$\Eover{\substack{U\gets\mathcal{U}\\
\ket{\psi}\gets\Psi\\
\qop\gets\A^U}}\[F\(\qop(\ket{\psi}\bra{\psi}),U\ket{\psi}\bra{\psi}U^\dagger\)\]\geq 1-\epsilon$$
where $F\(\sigma,\rho\)$ denotes the fidelity between density operators $\sigma$ and $\rho$.
In the special case where $\qop$ is a single linear operator $\SCgrav$,
since $\SCgrav\ket{\psi}$ and $U\ket{\psi}$ are pure, $F(\SCgrav\ket{\psi},U\ket{\psi})=\left|\bra{\psi}\SCgrav^\dagger U\ket{\psi}\right|^2$ is the squared inner product.~\footnote{Note that $\A$ does not necessarily output its hypothesis $\qop$. It may simply produce (a circuit representation of) a quantum channel $\qop$ and given $\ket{\psi}$, apply $\qop$ and output the resulting quantum state $\qop(\ket{\psi}\bra{\psi})$.}
\end{defn}

In order to make complexity theoretic statements, we consider an ensemble of distributions of unitaries $\mathcal{U}=\{\mathcal{U}_\Hn\}$ where each $\mathcal{U}_\Hn$ is a distribution over $2^\Hn$ dimensional unitaries. 
The algorithm runs in two phases: during the first, `query and learn' phase, the algorithm does not receive any input: it can only do computations on its own registers and produce states with which to query $U$. In particular, the algorithm does \textit{not} receive $\ket{\psi}$ in this phase. In the second, the algorithm loses oracle access and gains access to $\ket{\psi}$ instead; in this `prediction phase', it outputs its best guess for $U\ket{\psi}$. We will assume that $\A$ indeed first produces some quantum channel during `query and learn' and then given $\ket{\psi}$ it simply outputs the application of that quantum channel ${\cal E}$ to $\ket{\psi}$. This is denoted as ${\cal E}\(\ket{\psi}\bra{\psi}\)\gets\A^{U}\(\ket{\psi}\)$.

Next we describe the communication complexity result we use. For any interactive protocol between two parties, we can bound their mutual information by their communication complexity. This is formally stated in the following result by Arunachalam, Grilo, and Sundaram \cite{AruGri21}, building on \cite{Tou15,KerLau16}.
\begin{lem}[Corollary 2.2 in \cite{AruGri21}]\label{cor:I-QCC} 
Let $\cD$ be any distribution of a pair of variables $(X,Y)$. Consider a protocol $\prot$ between two parties $A$ and $B$: at the start, $A$ is given $X$ and $B$ is given $Y$ where $(X,Y)\sim\cD$. Subsequently, $A$ and $B$ exchange quantum information (qubits) back and forth. At the end, $B$ outputs a random variable $Z$. The \emph{quantum communication complexity} of the protocol $\QCC(\prot)$ is the number of qubits communicated between $A$ and $B$. We have
$$I(Z;X|Y)\leq\QCC(\prot).$$
\end{lem}
We can consider $\A^U$ to be a communication protocol where the oracle is party $A$ and is given $X=U$ and the algorithm $\A$ is party $B$. Each query made by $\A$ can be `sent' to the oracle and the response can be `sent' back to $\A$ (there may be no actual physical communication when $\A$ makes its oracle queries e.g. if the oracle represents a time evolution operator). If $\A$ makes $\ell$ queries of $\Hn$ qubits each then the communication complexity is $2\ell\Hn$.

\subsection{Random Matrix Approximation of Haar Random Unitaries}\label{sec:prelim-approx-unitary}
In this section we state the result of Jiang \cite{Jia10} showing that elements of a Haar random unitary can be approximated by independent complex Gaussians. We use this extensively in Section~\ref{sec:fidelity-upperbd-random}.

The original statement of Jiang's theorem, Proposition~\ref{thm:approx-Haar-U-Jiang}, and the
proof that the result below follows from the original statement, are in Appendix~\ref{sec:appendix-tech}.

\begin{prop}[Theorem A.2 in~\cite{Jia10}]\label{thm:approx-Haar-U}
Let the matrices $Z=\(Z_{ij}\)_{i,j\in[\Hd]}$ and $U=\(U_{ij}\)_{i,j\in[\Hd]}$ be defined as follows.
Let $Z=\(Z_{ij}\)_{i,j\in[\Hd]}$ be a $\Hd\times\Hd$ matrix where $Z_{ij}\gets\CNdistr(0,1)$ are independent complex Gaussian variables\footnote{Equivalently, let $Z_{ij}=\frac{X_{ij}+iY_{ij}}{\sqrt{2}}$ where $X_{ij},Y_{ij}\gets\Ndistr(0,1)$ are independent standard Gaussian variables.}. Let $U=\(U_{ij}\)_{i,j\in[\Hd]}$ be the $\Hd\times\Hd$ matrix that results from performing the Gram-Schmidt procedure\footnote{The Gram–Schmidt procedure is a simple process that takes a finite, linearly independent set of vectors and generates an orthogonal set that spans the same subspace.} on the $\Hd$ columns vectors of $Z$ and normalizing them. $U$ is distributed according to the Haar measure on the unitary group.
For any $t\in(0,6]$ let the functions $\md(\Hd)=\frac{t^2\Hd}{72\log\Hd}$ (we will always take $\md(\Hd)=\lfloor\frac{t^2\Hd}{72\log\Hd}\rfloor$) and $\epsilon(\Hd)=\frac{3t}{\sqrt{\Hd}}$. 
For every sufficiently large $\Hd$ (such that $\Hd>\max\{e^{1/t^4},e^4\}$ and $\frac{\sqrt{\Hd}}{\log\Hd}\geq \frac{72}{t}$),
$$\Pr_{Z}\[\max_{i\in[\Hd],j\in[\md(\Hd)]}\left|\Qgrav_{ij}-\frac{Z_{ij}}{\sqrt{\Hd}}\right|\geq\epsilon(\Hd)\]\leq\delta(\Hd)$$
where $\delta(\Hd)=O\(\frac{1}{\Hd}\)$.
\end{prop}

\subsection{Information Theory}\label{sec:prelim-info-theory}

In this section we briefly review the information theoretic constructs that we use following the presentation in \cite{CovTho06}. We refer the reader to Appendix~\ref{sec:appendix-info-theory} for further results used in our proofs.

We use the information theoretic notion of \emph{differential entropy} for continuous random variables. The unit of differential entropy depends on the base of the logarithm in the definition. We use the natural logarithm in this work, so all differential entropies will be in units of \emph{nats}. 

\begin{defn}(Differential Entropy)\label{defn:h}
The differential entropy of a continuous random variable $X$ is
$$h(X)=\E[-\log(p_X(X))]=-\int_{\mathcal {X}}p_X(x)\log p_X(x)\,dx.$$

The (joint) differential entropy of jointly distributed random variables $X_1,\ldots,X_\ell$ with joint PDF $p_{X_1,\ldots,X_\ell}(x_1,\ldots,x_\ell)$ is 
$$h(X_{1},\ldots,X_{\ell})=-\intover{\mathcal{X}_1\times\ldots\times\mathcal{X}_\ell} p_{X_1,\ldots,X_\ell}(x_{1},\ldots ,x_{\ell})\log p_{X_1,\ldots,X_\ell}(x_{1},\ldots,x_{\ell})\;dx_{1}\ldots dx_{\ell}.$$
\end{defn}

We are interested in how well information $X$ can be decoded from a noisy representation $Y$. Fano's inequality gives a lower bound on the probability of error in decoding $X$ from $Y$ if $X$ and $Y$ are \emph{discrete} random variables.
To work with continuous variables, we consider a function $\widehat{X}$ that tries to predict the random variable $X$ given correlated information $Y$. This is called the estimator function. We use the following bound on the mean squared estimator error from statistics. This bound is lesser known in information theory and in the context of decoding a noisy channel, but we believe it is akin to Fano's inequality\footnote{To see the analogy, we state Fano's inequality:
$$H(X|Y)\leq H_{b}(\Pr[X\neq {\widehat{X}(Y)}])+\Pr[X\neq {\widehat{X}(Y)}]\log(|{\mathcal {X}}|-1).$$
$X$ and $Y$ are discrete variables, $H$ is the (Shannon) entropy for discrete variables, and $H_b$ is the binary entropy. See \cite{CovTho06} for further details.} and we will refer to it as continuous Fano's inequality.
\begin{prop}[Continuous Fano's Inequality, Corollary/Theorem 8.6.6 in \cite{CovTho06}]\label{thm:cts-Fano-ineq-CovTho}
For any random variable $X\in\R$, correlated random variable $Y$, and `estimator' function $\widehat{X}:\mathcal{Y}\to\R$,
$$\E\[\(X-\widehat{X}(Y)\)^2\]\geq \frac{1}{2\pi e}e^{2h(X|Y)}$$
where $h$ is measured in nats.
\end{prop}

We will use the following restatement of Proposition~\ref{thm:cts-Fano-ineq-CovTho} to upper bound the expected value of $X\cdot\widehat{X}(Y)$.
\begin{prop}\label{thm:cts-Fano-ineq}
For any random variable $X\in\R$, correlated random variable $Y$, and `estimator' function $\widehat{X}:\mathcal{Y}\to\R$,
$$\frac{1}{2}\E\[X^2\]+\frac{1}{2}\E\[\widehat{X}(Y)^2\]-\frac{1}{4\pi e}e^{2h(X|Y)}\geq \E\[X\cdot\widehat{X}(Y)\]$$
where $h$ is measured in nats.
\end{prop}

Lastly we give a definition of mutual information for continuous variables and state its relation to differential entropy. The definition using KL divergence and further properties are in Appendix~\ref{sec:appendix-info-theory}.
\begin{defn}(Mutual Information)
The mutual information of jointly distributed random variables $X,Y$ with joint PDF $p_{X,Y}(x,y)$ is
$$I(X;Y)=\int_{\mathcal{Y}}\int _{\mathcal{X}}{p_{X,Y}(x,y)\log {\({\frac{p_{X,Y}(x,y)}{p_{X}(x)\,p_{Y}(y)}}\)}}\;dx\,dy.$$
\end{defn}
\begin{lem}\label{lem:I-expr-h}
$$I(X;Y)=h(X)-h(X|Y).$$
\end{lem}

\section{Hardness of Learning Random Unitaries}\label{sec:fidelity-upperbd-random}
In this section we show our first main result: any learner bounded in complexity, modeled as a quantum algorithm $\A$, cannot accurately guess -- or ``learn'' -- a Haar random unitary. Our results in this section may be read purely at the level of quantum learning. On a physics level, this can be seen as a warmup for our next section: while Haar random dynamics is a common simplifying assumption in black hole physics, such dynamics are simply too random (and too complex) for quantum gravity even within black holes. 

Let us formalize several aspects of the assumptions stated above. Our simple ``reconstruction'' algorithm $\A$, to which we shall also refer as the quantum ``learner'', is bounded in the following sense: $\A$ is limited to a subexponential (in $\log\dim{\cal H}$)-bounded number of observations of the action of $\Qgrav$ on quantum states.\footnote{Note that this includes algorithms that make a polynomial number of queries.} These are modeled as oracle queries to $\Qgrav$. The goal of $\A$ is to ``learn'' $\Qgrav$: i.e. to predict dynamics that agree with $\Qgrav$. We measure agreement as follows: given a Haar random state $\ket{\BH}$, $\A$ outputs a prediction $\SCgrav\ket{\BH}$ which is supposed to approximately agree with the time-evolved state $\Qgrav\ket{\BH}$. We measure the agreement by the quantum state fidelity, which in this case (since in this section $O$ is a linear operator and $\ket{\psi}$ is always a pure state for us) is just the squared inner product $\left|\bra{\BH}\SCgrav^\dagger\Qgrav\ket{\BH}\right|^2$.

We show the \emph{hardness of learning} a Haar random unitary. This means that \emph{any} $\A$ bounded in communication complexity, after observing random dynamics via oracle queries and performing any quantum computations, will fail to predict the evolution of $\ket{\BH}$ -- outputting a state that has non-maximal fidelity with $\Qgrav\ket{\BH}$.
In particular, the non-maximality of the fidelity is non-negligible -- polynomial in $\log\dim{\cal H}$: the algorithm $\A$ fails to learn under the definition of quantum learning.

In this section and in Section~\ref{sec:fidelity-upperbd-pseudorandom}, we assume $\A$ outputs a linear (single operator) model $\SCgrav$ of the fundamental dynamics $\Qgrav$; the fidelity bounds here apply to any such $\A$. In Section~\ref{sec:general-learner}, we show how to remove this restriction, obtaining a fidelity bound for any $\A$ learning any quantum operation (completely positive trace preserving map) as its model of $\Qgrav$.

We note and explain why the result in this section applies to \textit{any} quantum algorithm $\A$ that makes a bounded number of queries to $\Qgrav$, but the runtime of $\A$ is not bounded and can be arbitrarily large. Intuitively, this is because Haar random unitaries have an exponential amount of random information: it thus suffices only to bound the amount of information that $\A$ can access through its queries. Hence the fidelity bound in this section is an ``information theoretic'' bound. If $\A$ has access to limited information, it will not be able to produce an operator $\SCgrav$ which is well correlated with $\Qgrav$ regardless of how (long) it processes that information. In Section~\ref{sec:fidelity-upperbd-pseudorandom} we will model the fundamental time evolution operator $\Qgrav$ by a pseudorandom unitary instead of a Haar random unitary. Pseudorandom unitaries do not have an exponential amount of random information -- they can be efficiently applied to a state given only $\poly(\log\dim\Hil)$ bits of information (the randomly sampled ``key'' of the pseudorandom unitary). However, a computationally bounded $\A$ cannot distinguish them from Haar random unitaries. One may think of this as being due to the small amount of random information (e.g. in the ``key'') requiring a very long time to compute and hence it is inaccessible to computationally bounded $\A$.\footnote{To illustrate the inaccessibility (or hardness of computing) a small amount of information, we consider the factoring task for classical algorithms. If two large primes $p$ and $q$ are randomly sampled (the ``key'' here) and a classical algorithm is given their product $N=p\cdot q$, then the algorithm has all of the information -- $p$ and $q$ are determined by $N$. Yet it takes it a long time to \emph{compute} $p$ and $q$ (assuming there are not significantly faster, i.e. polynomial time, classical factoring algorithms).}

Let us briefly recall some notation. The full discussion of notation is in Section~\ref{sec:prelim}.
We model systems as quantum states of $\Hn=S$ qubits. Let $\Hd=2^\Hn$ be the dimension of the corresponding Hilbert space ${\cal H}$.
Recall $\A^\Qgrav(x)$ denotes the quantum algorithm's learning and prediction process. The algorithm $\A$ has oracle access to $\Qgrav$ and $x$ is the input it receives during the prediction phase (recall that it receives no input during the query and learn phase). This model is motivated and defined in Definitions~\ref{defn:quantum-queries} and~\ref{defn:learning-model}. In asymptotic (large $d$) statements, we write $\A_n$ to denote the algorithm $\A$ with oracle access to a $\Hd=2^\Hn$ dimensional unitary and asked to predict for a state $\ket{\psi}$ of $n$ qubits. We use $X\sim\A^\Qgrav$ to refer to random variables $X$ that $\A^\Qgrav$ produces after it queries and learns. 

\medskip
We now state the full version of Theorem~\ref{thm:fidelity-upperbd-random}.

\begin{thm}[Hardness of Learning Haar Random Unitaries]\label{thm:fidelity-upperbd-random}
For $\Hn\in\N$ let $\Qgrav$ be a Haar random unitary and $\ket{\BH}$ be a Haar random state in a Hilbert space ${\cal H}$ of dimension $\Hd=2^\Hn$.
For any quantum algorithm $\A$ that makes at most $o\(\frac{\Hd^2}{(\log\Hd)^2}\)$ quantum queries to $\Qgrav$ (Definition~\ref{defn:quantum-queries}), produces a circuit implementation of a single linear operator $\SCgrav\in\C^{\Hd\times\Hd}$ such that $\SCgrav$ approximately (up to $\frac{1}{\poly(\Hd)}$ corrections) preserves the norms of vectors, and when given $\ket{\BH}$ the algorithm outputs $\SCgrav\ket{\BH}$,
$$\Eover{\substack{\Qgrav\gets\Haar_\Hd\\
\ket{\BH}\gets\Haar_\Hd\\
\SCgrav\ket{\BH}\gets\A^\Qgrav(\ket{\BH})}}\[F\(\SCgrav\ket{\BH},\Qgrav\ket{\BH}\)\]\leq 1-\Omega\(1\)$$
where $F\(\SCgrav\ket{\BH},\Qgrav\ket{\BH}\)$ is the fidelity, i.e. for pure states the squared inner product $\left|\bra{\BH}\SCgrav^\dagger \Qgrav\ket{\BH}\right|^2$.
More generally, even when $\SCgrav$ does not necessarily preserve the norms of vectors, if  the maximum norm $\max\limits_{j\in[\Hd]}\sum\limits_{i\in[\Hd]}|\SCgrav_{ij}|^2\leq \alpha\leq 1$ (where $\alpha$ is a constant), then the above equation holds.
More precisely, for any constants $c > 0$ and $\beta\in(0,1)$,
$$
\Eover{\substack{\Qgrav\gets\Haar_\Hd\\
\ket{\BH}\gets\Haar_\Hd\\
\SCgrav\ket{\BH}\gets\A^\Qgrav(\ket{\BH})}}\[\left|\bra{\BH}\SCgrav^\dagger \Qgrav\ket{\BH}\right|^2\]\leq \frac{1+\SCcol}{2}-\(\frac{e^{\gamma_E+1}}{16\pi}\cdot e^{-c}\,(1-\beta)^2\).
$$

\end{thm}

\medskip

\paragraph{Remark.}
We note that while we defined explicit forms of quantum queries in our model (Definition~\ref{defn:quantum-queries}) -- i.e. when $\A$ queries the oracle on $\ket{\phi_i}$ (or more generally $\rho_i$) it receives $\Qgrav\ket{\phi_i}$ (more generally $\Qgrav\rho_i\Qgrav^\dagger$) -- Theorem~\ref{thm:fidelity-upperbd-random} holds for $\A$ querying and receiving qubits with the oracle via even more general protocols. It does not rely on this form for the qubits exchanged between $\A$ and the oracle, only that a \emph{bounded} number of qubits of information are exchanged.

\subsection{Proof of Theorem~\ref{thm:fidelity-upperbd-random}}
In this section we give the proof of Theorem~\ref{thm:fidelity-upperbd-random}.

We refer the reader to the technical overview (Section~\ref{sec:tech-overview}) for a detailed walk-through of the proof. Here we will simply jog the reader's memory (under the assumption that all of our readers have carefully and painstakingly worked through Section~\ref{sec:tech-overview}) by giving a high-level outline of the proof below. 
\begin{enumerate}
    \item First, we derive an expression for the squared inner product (the fidelity when $\SCgrav$ is unitary) averaged over $\ket{\BH}\gets\Haar_\Hd$, in terms of the matrix elements of $\SCgrav$ and $\Qgrav$.
    \item Then we bound this expression using the estimator error bound in Proposition~\ref{thm:cts-Fano-ineq} (which we call continuous Fano's inequality) to bound the average fidelity by the conditional entropy of $\Qgrav$ given $\SCgrav$. 
    
    \item This conditional entropy will be close to the entropy of $\Qgrav$ by the quantum communication complexity bound (Corollary~\ref{cor:I-QCC}) since $\A$ is allowed a limited number of queries. To give a numerical lower bound for $h(\Qgrav)$ in the end, we use the fact that many of the elements of a Haar random matrix are close to independent complex Gaussians (Proposition~\ref{thm:approx-Haar-U}, a restatement of Proposition~\ref{thm:approx-Haar-U-Jiang} with specific parameters). This will allow us to use the differential entropy of complex Gaussians to lower bound the differential entropy of $\Qgrav$ to obtain our result.  
    \item Proposition~\ref{thm:approx-Haar-U} only allows us to approximate $\Theta\(\frac{1}{\log\Hd}\)$ fraction of the columns of $\Qgrav$ by complex Gaussians. (This is in fact not completely surprising as the columns of a Haar random $\Qgrav$ does not have all independent columns.) To more fully utilize this result, we partition the columns of $\Qgrav$ into $\Theta\(\lfloor\frac{\Hd}{\log\Hd}\rfloor\)$ subsets of columns. For each subset, we separately define a set of complex Gaussians that approximates those columns.
    \item We then replace the elements of $\Qgrav$ by their corresponding complex Gaussians (plus an error term) and then use continuous Fano's inequality. It suffices to lower bound the conditional entropy of these complex Gaussians given $\SCgrav$.
    \item We now use the quantum communication complexity bound and some information theory facts to relate this conditional entropy to the entropy of many independent complex Gaussians which is easy to calculate numerically.
    \item Finally we instantiate any parameters and use the conditions stated in the theorem to arrive at the final bound.
\end{enumerate}

For clarity, we break up the proof of Theorem~\ref{thm:fidelity-upperbd-random} into several useful individual lemmas. As a shorthand, for matrices $U,O\in\C^{\Hd\times\Hd}$, we use:
$$F_{\SCgrav,\Qgrav}:=\Eover{\ket{\BH}\gets\Haar_\Hd}\[F(\SCgrav\ket{\BH},\Qgrav\ket{\BH})\].$$

First we prove the below lemma on averages of the squared inner product between $\SCgrav\ket{\BH}$ and $\Qgrav\ket{\BH}$ over Haar random $\ket{\BH}$.\footnote{We thank A. Harrow for discussions on this lemma.} This will be used to prove Theorem~\ref{thm:fidelity-upperbd-random} below, and also Theorem~\ref{thm:fidelity-general-learners} in Section~\ref{sec:general-learner}.
\begin{lem}\label{lem:avg-random-state}
For any (fixed) unitary $\Qgrav\in\C^{\Hd\times\Hd}$ and any (fixed) matrix $\SCgrav\in\C^{\Hd\times\Hd}$,
\begin{align*}d(d+1)F_{\SCgrav,\Qgrav}&=\left|\sum_{i,j\in[\Hd]}\SCgrav^*_{ij}\Qgrav_{ij}\right|^2+\Tr(\SCgrav^\dagger\SCgrav)\\
& \leq \left|\sum_{i,j\in[\Hd]}\SCgrav^*_{ij}\Qgrav_{ij}\right|^2+\Hd\SCcol
\end{align*} 
where $\SCcol:=\max\limits_{i\in[\Hd]}\sum\limits_{j\in[\Hd]}|\SCgrav_{ij}|^2$.
\end{lem}
\begin{proof}[Proof of Lemma~\ref{lem:avg-random-state}]
For any such $\Qgrav$ and $\SCgrav$ let $A=\SCgrav^\dagger \Qgrav$. We have:
\begin{align*}
F_{\SCgrav,\Qgrav}&=\Eover{\ket{\BH}\gets\Haar_\Hd}\[\left|\bra{\BH}\SCgrav^\dagger \Qgrav\ket{\BH}\right|^2\]\\
&=\Tr\[(A\otimes A^\dagger)\Eover{\ket{\BH}\gets\Haar_\Hd}\[ \ket{\BH}\bra{\BH}\otimes\ket{\BH}\bra{\BH}\]\].
\end{align*}

Next we make use of a result from~\cite{Har13}:
$$\Eover{\ket{\BH}\gets\Haar_\Hd}\[(\ket{\BH}\bra{\BH})^{\otimes 2}\]=\frac{1}{\Hd(\Hd+1)}\(\Id_\Hd+\SWAP_\Hd\)$$
where $\Id_\Hd$ is the $\Hd\times\Hd$ identity matrix and $\SWAP_\Hd=\sum_{i,j\in[\Hd]}\ket{i}\bra{j}\otimes\ket{j}\bra{i}$ is the $\Hd$-dimensional `SWAP operator'.
From this it follows that:
\begin{align*}
F_{\SCgrav,\Qgrav}&=\frac{1}{\Hd(\Hd+1)}\[\Tr(A\otimes A^\dagger)+\Tr\[(A\otimes A^\dagger)\SWAP_\Hd\]\].\label{eqn:fidelity-random-pf-avged-state}
\end{align*}
This can be rewritten as:
\begin{align*}
\Hd(\Hd+1)F_{\SCgrav,\Qgrav}& = |\Tr A|^2+\Tr\[(A\otimes A^\dagger)\(\sum_{i,j\in[\Hd]}\ket{i}\bra{j}\otimes\ket{j}\bra{i}\)\]\\
& = \left|\sum_{i,j\in[\Hd]} \SCgrav^*_{ij}\Qgrav_{ij}\right|^2+ \Tr(\SCgrav^\dagger\SCgrav).
\end{align*}
This proves the equality in Lemma~\ref{lem:avg-random-state}. To get the subsequent inequality, we just use $\SCcol=\max\limits_{i\in[\Hd]}\sum\limits_{j\in[\Hd]}|\SCgrav_{ij}|^2$.
\end{proof}

\begin{lem}\label{lem:fidelity-random-pf-F-sqrt-F-ineq}
For any unitary $\Qgrav$ and $\SCgrav$ such that $\max\limits_{j\in[\Hd]}\sum\limits_{i\in[\Hd]}|\SCgrav_{ij}|^2\leq \alpha\leq 1$ (where $\alpha$ is a constant),
$$F_{\SCgrav,\Qgrav}\leq \frac{1}{\Hd}\sqrt{\Hd(\Hd+1)F_{\SCgrav,\Qgrav}-\Hd\SCcol}+\frac{3\SCcol}{2\sqrt{\Hd}}+\frac{1}{2\sqrt{\Hd}}.$$
\end{lem}

The proof of this lemma is relegated to Appendix~\ref{sec:appendix-tech}.

\begin{lem}\label{lem:fidelity-rand-Vk-approx-Gaussian}
Let $\Qgrav$ be a Haar random unitary matrix. For any permutation matrix $P\in\{0,1\}^{\Hd\times\Hd}$, there exists a matrix of complex Gaussian variables $\CN=\(\CN_{ij}\)_{i,j\in[\Hd]}$, where the distribution of each $\CN_{ij}$ is $\CNdistr\(0,\frac{1}{\Hd}\)$, such that Proposition~\ref{thm:approx-Haar-U} can be applied to $\sqrt{\Hd}\cdot \CN$ and $\Qgrav P$ (showing that the first $\md(\Hd)$ columns of $\Qgrav P$ can be approximated by the first $\md(\Hd)$ columns of $\CN$).
\end{lem}
\begin{proof}[Proof of Lemma~\ref{lem:fidelity-rand-Vk-approx-Gaussian}] 
Let $\CN_0$ be a $\Hd\times\Hd$ matrix $\CN$ of independent complex Gaussian variables, e.g. from $\CNdistr(0,\frac{1}{\Hd})$. Applying the Gram-Schmidt orthogonalization process to the $\Hd$ column vectors of $\CN_0$ and then normalizing each column vector produces a $\Hd\times\Hd$ Haar random unitary matrix $\Qgrav_0$. 
Thus there exists a joint distribution of $(\Qgrav_0,\CN_0)$ where:
\begin{enumerate}
\item The marginal distribution of $\CN_0$ is a matrix of independent complex Gaussians, each sampled from $\CNdistr\(0,\frac{1}{\Hd}\)$, and the marginal distribution of $\Qgrav_0$ is a Haar random unitary.
\item For any sample $(\Qgrav_0,\CN_0)$, it holds that $\Qgrav_0$ is equal to the Gram-Schdmit orthonormalization of $\CN_0$.
\end{enumerate}

$\Qgrav$ is a random variable, distributed according to the Haar measure $\mu$.
The right-translation-invariance of the Haar measure states that for every (Borel) subset 
$S$ of the unitary group and every unitary matrix $g$, one has $\mu(Sg)=\mu(S)$. Any permutation matrix $P\in\{0,1\}^{\Hd\times\Hd}$ is its own inverse,  so $\Qgrav P$ is also a random variable distributed according to the Haar measure.
Since $\Qgrav P$ is a Haar random unitary, there exists a jointly distributed random variable $\CN=\(\CN_{ij}\)_{i,j\in[\Hd]}$ such that the joint distribution of $(\Qgrav P,\CN)$ has properties~1 and~2 above. Thus Proposition~\ref{thm:approx-Haar-U} can be applied to the matrices $\sqrt{\Hd}\cdot\CN$ and $\Qgrav P$.
\end{proof}

To prove our poor correlation result below, our approach requires (a bound on) the differential entropy of the random variables to be learned. We do not know how to explicitly calculate this (or the probability density function) for a random unitary matrix: instead our approach goes through complex Gaussian variables -- which can be related to $\Qgrav$'s elements -- and for which we can calculate the differential entropy.
 
\begin{lem}[Bound on Correlating with Approximate Gaussians]\label{lem:fidelity-rand-upperbd-O-U} 
Consider any distribution of (jointly sampled) complex-valued matrices $(\Vmat,\CN)$ such that $\Vmat=\(\Vmat_{ij}\)_{\substack{i\in[\Hd]\\ j\in[\md(\Hd)]}}$ always has orthonormal columns, and in the marginal distribution of $\CN=\(\CN_{ij}\)_{\substack{i\in[\Hd]\\j\in[\md(\Hd)]}}$, $\CN_{ij}$ are independent complex Gaussian variables each distributed according to $\CNdistr\(0,\frac{1}{\Hd}\)$. If $\Vmat$ is approximated by $\CN$, meaning that for some functions $\epsilon(\Hd)$ and $\delta(\Hd)$:
$$\Pr_{\Vmat,\CN}\[\max_{\substack{i\in[\Hd]\\ j\in[\md(\Hd)]}}\left|\Vmat_{ij}-\CN_{ij}\right|\leq\epsilon(\Hd)\]\geq 1-\delta(\Hd),$$
then for any random variable $\Qmat=\(\Qmat_{ij}\)_{\substack{i\in[\Hd]\\j\in[\md(\Hd)]}}\in\C^{\Hd\times\md(\Hd)}$ (which may be correlated with $\Vmat$ and $\CN$) and $\SCcol:=\max\limits_{\Qmat}\(\max\limits_{j\in[\md(\Hd)]}\sum\limits_{i\in[\Hd]}|\Qmat_{ij}|^2\)$, we have
$$\sum_{\substack{i\in[\Hd]\\ j\in[\md(\Hd)]}} \(\Eover{\Vmat,\CN,\Qmat}\left|\Qmat^*_{ij}\Vmat_{ij}\right|\)\leq\md(\Hd)\(\epsilon(\Hd)\sqrt{\SCcol\Hd}+\delta(\Hd)\sqrt{\SCcol}+\frac{1+\SCcol}{2(1-\delta(\Hd))}-\frac{e^{\gamma_E+1}}{16\pi}\cdot e^{\frac{-2\cdot I\(|\CN|;\Qmat\)}{\Hd\cdot\md(\Hd)}}\).$$

The techniques here first give a bound in terms of the differential entropy of $\CN$'s elements for any matrix $\CN$ that approximates $\Vmat$ (not necessarily a matrix of complex Gaussians). We then use that $\CN_{ij}$ are complex Gaussians to obtain the above bound.
\end{lem}

In this lemma, for random variables $X$ the notation $\max\limits_{X}$ denotes taking the maximum over the support of the distribution of $X$.

We will invoke this lemma with the $\Vmat_k$ and (first $\md(\Hd)$ columns of) $\CN_k$ defined above. For our final bound, the parameters quantifying the strength of the approximation, $\epsilon(\Hd)$ and $\delta(\Hd)$, will be set to the corresponding parameters in Proposition~\ref{thm:approx-Haar-U} so both will be $O\(\frac{1}{\sqrt{\Hd}}\)$. Thus the reader may think of these terms as having little contribution in the upper bound in Lemma~\ref{lem:fidelity-rand-upperbd-O-U}. We will derive a bound for any approximation, in terms of generic $\epsilon(\Hd)$ and $\delta(\Hd)$, and then specialize to our setting for Theorem~\ref{thm:fidelity-upperbd-random} at the end. 

The proof of Lemma~\ref{lem:fidelity-rand-upperbd-O-U} contains key steps of our proof, but as it is highly technical, so we have placed it in Appendix~\ref{sec:appendix-tech}. 

\medskip\noindent We now proceed to the proof of Theorem~\ref{thm:fidelity-upperbd-random}.
\begin{proof}[Proof of Theorem~\ref{thm:fidelity-upperbd-random}]

Consider any fixed unitary $\Qgrav\in\C^{\Hd\times\Hd}$ and any fixed matrix $\SCgrav\in\C^{\Hd\times\Hd}$. The proof will focus on upper bounding the terms ${\left| \SCgrav^*_{ij}\Qgrav_{ij}\right|}$ (on average) in order to upper bound $F_{\SCgrav,\Qgrav}$ (on average) as desired. We first turn the inequality in Lemma~\ref{lem:avg-random-state} into a more viable inequality.
By the triangle inequality,
\begin{equation}\label{eqn:fidelity-random-pf-sqrtF-sum-OU}
\sqrt{\Hd(\Hd+1)F_{\SCgrav,\Qgrav}-\Hd\SCcol}\leq \sum_{i,j\in[\Hd]}\left|\SCgrav^*_{ij}\Qgrav_{ij}\right|. 
\end{equation}
The left side above is related to $F_{O,U}$, the fidelity averaged over Haar random $\ket{\BH}$; by Lemma~\ref{lem:fidelity-random-pf-F-sqrt-F-ineq} and Equation~\ref{eqn:fidelity-random-pf-sqrtF-sum-OU},
\begin{equation}\label{eqn:fidelity-random-pf-upperbd-SC*Q}
F_{\SCgrav,\Qgrav}\leq \frac{1}{\Hd}\sum_{i,j\in[\Hd]}\left|\SCgrav^*_{ij}\Qgrav_{ij}\right|+\frac{3\SCcol}{2\sqrt{\Hd}}+\frac{1}{2\sqrt{\Hd}}.
\end{equation}
Recall we fixed any $\Qgrav,\SCgrav$ so Equation~\ref{eqn:fidelity-random-pf-upperbd-SC*Q} holds for every $\Qgrav,\SCgrav$. We take the expectation (i.e. the average) over the distribution of $\Qgrav\gets\Haar_\Hd$ and $\SCgrav\sim\A^\Qgrav_\Hn$ corresponding to our model.
\begin{align}
\Eover{\Qgrav,\SCgrav}\[F_{\SCgrav,\Qgrav}\]&\leq \frac{1}{\Hd}\;\Eover{\Qgrav,\SCgrav}\[\sum_{i,j\in[\Hd]}\left|\SCgrav^*_{ij}\Qgrav_{ij}\right|\]+\frac{3\SCcol}{2\sqrt{\Hd}}+\frac{1}{2\sqrt{\Hd}}\nonumber\\
&= \frac{1}{\Hd}\;\sum_{i,j\in[\Hd]} \(\Eover{\Qgrav,\SCgrav}\left|\SCgrav^*_{ij}\Qgrav_{ij}\right|\)+\frac{3\SCcol}{2\sqrt{\Hd}}+\frac{1}{2\sqrt{\Hd}}.\label{eqn:fidelity-random-pf-upperbd-SC*Q-E}
\end{align}

We now focus on upper bounding the summand on the right in order to upper bound the average fidelity on the left. We partition the terms in this sum by partitioning the unitary matrix $\Qgrav\equiv\[\Vmat_1|\Vmat_2|\ldots|\Vmat_{\qd}|\Vmat_{\qd+1}\]$ where the $V_k$ are matrices and $|$ denotes concatenation of matrices. We will consider the sum of $\Eover{\Qgrav,\SCgrav}\left|\SCgrav^*_{ij}\Qgrav_{ij}\right|$ over each $V_k$ \emph{separately}. 

To do so, we define some notation. Let $\md(\Hd)$ be the number of columns of a Haar random unitary that Proposition~\ref{thm:approx-Haar-U} shows can be (simultaneously) approximated by complex Gaussian variables.

\begin{figure}
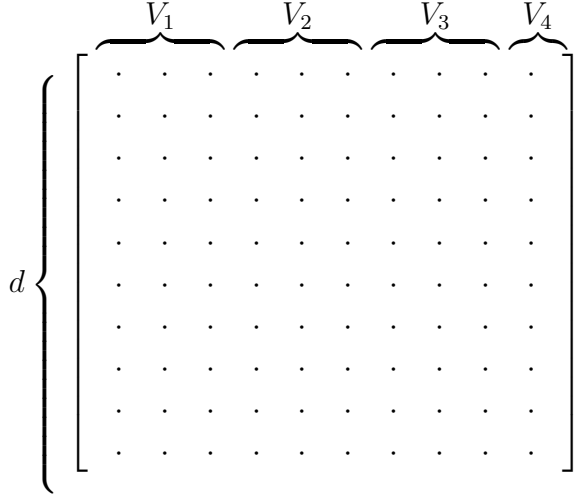

\centering
$$\vphantom{}
\begin{matrix}
\vphantom{a}\\
    \coolleftbrace{d}{\cdot \\ \cdot \\ \cdot \\ \cdot \\ \cdot \\ \cdot \\ \cdot \\ \cdot \\ \cdot \\ \cdot}
\end{matrix}
\begin{bmatrix} 
    \;\coolover{V_1}{\;\;\cdot\; & \;\;\cdot\; & \;\;\cdot\;} & \coolover{V_2}{\;\;\cdot\; & \;\;\cdot\; & \;\;\cdot\;} & \coolover{V_3}{\;\;\cdot\; & \;\;\cdot\; & \;\;\cdot\;} & \coolover{V_4}{\;\;\cdot}\;\;\; \\
    \;\;\;\cdot\; & \;\;\cdot\; & \;\;\cdot\; & \;\;\cdot\; & \;\;\cdot\; & \;\;\cdot\; & \;\;\cdot\; & \;\;\cdot\; & \;\;\cdot\; & \;\;\cdot\;\;\; \\
    \;\;\;\cdot\; & \;\;\cdot\; & \;\;\cdot\; & \;\;\cdot\; & \;\;\cdot\; & \;\;\cdot\; & \;\;\cdot\; & \;\;\cdot\; & \;\;\cdot\; & \;\;\cdot\;\;\; \\
    \;\;\;\cdot\; & \;\;\cdot\; & \;\;\cdot\; & \;\;\cdot\; & \;\;\cdot\; & \;\;\cdot\; & \;\;\cdot\; & \;\;\cdot\; & \;\;\cdot\; & \;\;\cdot\;\;\; \\
    \;\;\;\cdot\; & \;\;\cdot\; & \;\;\cdot\; & \;\;\cdot\; & \;\;\cdot\; & \;\;\cdot\; & \;\;\cdot\; & \;\;\cdot\; & \;\;\cdot\; & \;\;\cdot\;\;\; \\
    \;\;\;\cdot\; & \;\;\cdot\; & \;\;\cdot\; & \;\;\cdot\; & \;\;\cdot\; & \;\;\cdot\; & \;\;\cdot\; & \;\;\cdot\; & \;\;\cdot\; & \;\;\cdot\;\;\; \\
    \;\;\;\cdot\; & \;\;\cdot\; & \;\;\cdot\; & \;\;\cdot\; & \;\;\cdot\; & \;\;\cdot\; & \;\;\cdot\; & \;\;\cdot\; & \;\;\cdot\; & \;\;\cdot\;\;\; \\
    \;\;\;\cdot\; & \;\;\cdot\; & \;\;\cdot\; & \;\;\cdot\; & \;\;\cdot\; & \;\;\cdot\; & \;\;\cdot\; & \;\;\cdot\; & \;\;\cdot\; & \;\;\cdot\;\;\; \\
    \;\;\;\cdot\; & \;\;\cdot\; & \;\;\cdot\; & \;\;\cdot\; & \;\;\cdot\; & \;\;\cdot\; & \;\;\cdot\; & \;\;\cdot\; & \;\;\cdot\; & \;\;\cdot\;\;\; \\
    \;\;\;\cdot\; & \;\;\cdot\; & \;\;\cdot\; & \;\;\cdot\; & \;\;\cdot\; & \;\;\cdot\; & \;\;\cdot\; & \;\;\cdot\; & \;\;\cdot\; & \;\;\cdot\;\;\; \\
\end{bmatrix}$$
\caption{An example of partitioning a matrix of dimension $\Hd=10$ into matrices $V_k$ each with $\md(\Hd)=3$ columns (this value of $\md(\Hd)$ is chosen only for this illustration). The sets of indices $m_1=\{1,2,3\}, m_2=\{4,5,6\}, m_3=\{7,8,9\}, m_4=\{10\}$. There are $\qd:=\lfloor\frac{\Hd}{\md(\Hd)}\rfloor=3$ such $V_k$ and the remaining columns are in $V_{\qd+1}=V_4$.}
\end{figure}

\noindent
Let $\qd:=\lfloor\frac{\Hd}{\md(\Hd)}\rfloor$. We partition the $\Hd$ columns of $\Qgrav$ into $\qd+1$ sets of columns as follows. For every $k\in[\qd]$ define (1) the set of integers $\md_k:=\[(k-1)\cdot\md(\Hd)+1,k\cdot\md(\Hd)\]$ and define $\md_{\qd+1}:=\[\qd\cdot\md(\Hd)+1,\Hd\]$, and (2) the $\Hd\times\md(\Hd)$ matrix $\Vmat_k:=\(\Qgrav_{ij}\)_{\substack{i\in[\Hd]\\ j\in\md_k}}$ to contain the columns of $\Qgrav$ with column indices in $\md_k$, and define the matrix $\Vmat_{\qd+1}:=\(\Qgrav_{ij}\)_{\substack{i\in[\Hd]\\ j\in\md_{\qd+1}}}$ to contain the remaining columns of $\Qgrav$ ($\Vmat_{\qd+1}$ may be empty). 
We now partition the sum in Equation~\ref{eqn:fidelity-random-pf-upperbd-SC*Q-E} into $\qd+1$ parts. 
For every $k\in[\qd+1]$ define
\begin{equation}\label{eqn:fidelity-random-pf-defn-sum-EOU}
\sigma_k:=\sum_{\substack{i\in[\Hd]\\ j\in\md_k}}\(\Eover{\Qgrav,\SCgrav}\left|\SCgrav^*_{ij}\Qgrav_{ij}\right|\) 
\end{equation}
which is a measure of how much of $\Vmat_k$ has been learned in $\SCgrav$. 
Equation~\ref{eqn:fidelity-random-pf-upperbd-SC*Q-E} rewritten is

\begin{equation}\label{eqn:fidelity-random-pf-upperbd-SC*Q-two-sums}
\Eover{\Qgrav,\SCgrav}\[F_{\SCgrav,\Qgrav}\]\leq \frac{1}{\Hd}\
\(\sum_{k\in[\qd]}\sigma_k\)+\frac{\sigma_{\qd+1}}{\Hd}
+\frac{3\SCcol}{2\sqrt{\Hd}}+\frac{1}{2\sqrt{\Hd}}.
\end{equation}

Next we upper bound $\sigma_{\qd+1}$. Then we will upper bound the $\sigma_k$ for $k\in[\qd]$, which will be the bulk of the work and where we derive that the average fidelity must be greater than $ \Omega(1)$ below the maximal value of $1$.

\paragraph{Bounding $\sigma_{\qd+1}$.} For every $i\in[\Hd],j\in\md_{\qd+1}$, we apply the continuous Fano's inequality (Proposition~\ref{thm:cts-Fano-ineq}) with $X=|\Qgrav_{ij}|$, $Y=\SCgrav$, and $\widehat{X}(\SCgrav)=|\SCgrav_{ij}|$. Summing up these separate inequalities, we get
$$\sigma_{\qd+1}\leq -\frac{1}{4\pi e}\sum_{\substack{i\in[\Hd]\\ j\in\md_{\qd+1}}}e^{2h\(|\Qgrav_{ij}||\SCgrav\)}+\sum_{\substack{i\in[\Hd]\\ j\in\md_{\qd+1}}}\frac{\Eover{\Qgrav,\SCgrav}\(|\Qgrav_{ij}|^2\)}{2}+\sum_{\substack{i\in[\Hd]\\ j\in\md_{\qd+1}}}\frac{\Eover{\Qgrav,\SCgrav}\(|\SCgrav_{ij}|^2\)}{2}.$$
Using $e^{2h\(|\Qgrav_{ij}||\SCgrav\)}\geq 0$ and the norms of the columns of $\Qgrav$ and $\SCgrav$, we have
\begin{align}
\sigma_{\qd+1}&\leq \frac{\Hd-\qd\cdot\md(\Hd)}{2}+\frac{(\Hd-\qd\cdot\md(\Hd))\cdot\SCcol}{2} \nonumber\\
&=\frac{(\Hd-\qd\cdot\md(\Hd))(1+\SCcol)}{2}.\label{eqn:fidelity-random-pf-SC*Q-above-m(d)}
\end{align}

\paragraph{Bounding $\sigma_k$ for every $k\in[\qd]$.} 

We now use Lemmas~\ref{lem:fidelity-rand-Vk-approx-Gaussian} and~\ref{lem:fidelity-rand-upperbd-O-U} to show that the operator $\SCgrav$ produced by the quantum algorithm necessarily contains poor estimations of the entries of each $\Vmat_k$ (i.e. each $\sigma_k$ must be sub-maximal). This suffices to bound the average fidelity below $1$ by the linearity in Equation~\ref{eqn:fidelity-random-pf-upperbd-SC*Q-two-sums}, showing that $\Qgrav$ cannot be learned.~\footnote{The general idea behind this part of the proof is as follows: we first prove that the matrix elements of $\Vmat_k$ can be jointly approximated by a set of independent complex Gaussian variables (this will be Lemma~\ref{lem:fidelity-rand-Vk-approx-Gaussian}). Then we prove that if a set of random variables can be approximated by a set of independent complex Gaussians, any estimate of these random variables generated from limited mutual information must have poor correlation with the actual random variables. We use a result of~\cite{AruGri21} to show the mutual information obtained by any algorithm with limited quantum query access to the oracle must also be limited. 
In a nutshell, we are showing that it is hard to learn \emph{any subset} of $\md(\Hd)$ columns of $\Qgrav$.}

In what follows, we say that for a distribution of two matrices $(M,M')$, the matrix $M=(M_{ij})$ is \emph{approximated} by the matrix $M'=(M'_{ij})$ if with high probability ($\geq 1-\delta$ for small $\delta$) over the distribution of $(M,M')$, we have that for every $i,j$, $|M_{ij}-M'_{ij}|\leq\epsilon$ for small $\epsilon$. Proposition~\ref{thm:approx-Haar-U} states there is such an approximation (for $\md(\Hd)$ columns of a Haar random unitary by a matrix of complex Gaussians).~\footnote{The statement that $M$ is approximated by $M'$ is the statement that for (nearly all) samples $(M,M')$, all of the entries of $M$ can be \emph{jointly} approximated by the entries of $M'$. This simultaneous approximation is stronger than each \emph{single} entry of $M$ having approximately the same distribution as the corresponding entry of $M'$.}

For every $k\in[\qd]$ let $P_k\in\{0,1\}^{\Hd\times\Hd}$ be the permutation matrix that swaps the columns of $\Qgrav$ indexed by $\md_1$ with the columns indexed by $\md_k$ so that
$$\Qgrav P_k=\[\Vmat_k|\Vmat_2|\ldots|\Vmat_{k-1}|\Vmat_{1}|\Vmat_{k+1}|\ldots|\Vmat_{\qd+1}\].$$ 
For every $k\in[\qd]$, by Lemma~\ref{lem:fidelity-rand-Vk-approx-Gaussian}, there exists a matrix $\CN_k=\(\CN_{kij}\)_{i,j\in[\Hd]}$ of independent complex Gaussians, each with distribution $\CNdistr\(0,\frac{1}{\Hd}\)$, such that the entries in the first $\md(\Hd)$ columns of $\CN_k$ approximate the entries in $\Vmat_k$.

We now make use of Lemma~\ref{lem:fidelity-rand-upperbd-O-U} which holds for generic matrices $\Vmat$ and $\CN$ satisfying the stated conditions. Recall that we defined a quantity $\sigma_k$ for every $k\in[\qd]$ in Equation~\ref{eqn:fidelity-random-pf-defn-sum-EOU}. These quantities together with $\sigma_{\qd+1}$ constitute an upper bound on the average fidelity, as written in Equation~\ref{eqn:fidelity-random-pf-upperbd-SC*Q-two-sums}.
We are now ready to use Lemmas~\ref{lem:fidelity-rand-Vk-approx-Gaussian} and~\ref{lem:fidelity-rand-upperbd-O-U} as initially described to upper bound the $\sigma_k$. Recall that we started with a joint distribution of a Haar random unitary $\Qgrav\gets\mu$ and an operator $\SCgrav\sim\A^\Qgrav_\Hn$ produced by the algorithm during the query and learn phase. We now have a joint distribution of random variables
$$(\Qgrav,\SCgrav,\CN_1,\ldots,\CN_\qd)$$
where we used Lemma~\ref{lem:fidelity-rand-Vk-approx-Gaussian} to define the matrices $\CN_k=\(\CN_{kij}\)_{i,j\in[\Hd]}$ for $k\in[\qd]$, each depending on $\Qgrav$.
By Lemma~\ref{lem:fidelity-rand-Vk-approx-Gaussian}, the entries in the first $\md(\Hd)$ columns of each $\CN_k$ approximate the entries in $\Vmat_k$ with the parameters $\epsilon(\Hd)$ and $\delta(\Hd)$ specified by Proposition~\ref{thm:approx-Haar-U}.
Additionally, in the marginal distribution of a single matrix $\CN_k$, $\CN_{kij}$ for $i,j\in[\Hd]$ are independent complex Gaussians each with distribution $\CNdistr\(0,\frac{1}{\Hd}\)$.
Then for each $k\in[\qd]$, we can apply Lemma~\ref{lem:fidelity-rand-upperbd-O-U} to $\Vmat=\Vmat_k$ which was defined to equal $\(\Qgrav_{ij}\)_{\substack{i\in[\Hd]\\ j\in\md_k}}$, $\CN=\(\CN_{kij}\)_{\substack{i\in[\Hd]\\ j\in[\md(\Hd)]}}$ and $\Qmat=\(\SCgrav_{ij}\)_{\substack{i\in[\Hd]\\ j\in\md_k}}$.
Since the expected value (average) is the same over different marginals of $(\Qgrav,\SCgrav,\CN_1,\ldots,\CN_\qd)$,
$$\sigma_k:=\sum_{\substack{i\in[\Hd]\\ j\in\md_k}}\(\Eover{\Qgrav,\SCgrav}\left|\SCgrav^*_{ij}\Qgrav_{ij}\right|\) =\sum_{\substack{i\in[\Hd]\\ j\in[\md(\Hd)]}} \(\Eover{\Vmat,\CN,\Qmat}\left|\Qmat^*_{ij}\Vmat_{ij}\right|\).$$
Thus Lemma~\ref{lem:fidelity-rand-upperbd-O-U} gives us an upper bound on each $\sigma_k$ that depends on the mutual information $I_k:=I\(\{|\CN_{kij}|\}_{\substack{i\in[\Hd]\\j\in[\md(\Hd)]}};\(\SCgrav_{ij}\)_{\substack{i\in[\Hd]\\j\in\md_k}}\)$:
\begin{equation}\label{eqn:fidelity-random-pf-SC*Q-below-m(d)}
\sigma_k \leq\md(\Hd)\(\epsilon(\Hd)\sqrt{\SCcol\Hd}+\delta(\Hd)\sqrt{\SCcol}+\frac{1+\SCcol}{2(1-\delta(\Hd))}-\frac{e^{\gamma_E+1}}{16\pi}\cdot e^{\frac{-2\cdot I_k}{\Hd\cdot\md(\Hd)}}\).
\end{equation}

We now relate these mutual information terms to the quantum query (communication) complexity $\QCC$ between the algorithm and $\Qgrav$ while it is computing $\SCgrav$. 
Let $\prot$ denote the following interactive protocol: $\Qgrav\gets\mu$ is sampled, defining the matrices $\CN_k$, and party A is given the elements $\{|\CN_{1ij}|\}_{\substack{i\in[\Hd]\\ j\in[\md(\Hd)]}},\ldots,\{|\CN_{\qd ij}|\}_{\substack{i\in[\Hd]\\ j\in[\md(\Hd)]}}$ and (a circuit implementation of) $\Qgrav$. Party B receives nothing. Let party A represent the oracle $\Qgrav$ and party B represent the quantum learning algorithm in our model: party B runs the learning algorithm's computations. Whenever the learning algorithm queries for a state in its registers, party B sends those qubits to party A which applies $\Qgrav$ and sends back the resulting qubits. Party A and B proceed like this until the learning algorithm has finished querying and computing (a representation of) $\SCgrav$. Party B then computes and outputs the operator $\SCgrav\in\C^{\Hd\times\Hd}$.
This is a valid quantum communication protocol. Thus we can apply the quantum communication bound on the mutual information between parties A and B (Corollary~\ref{cor:I-QCC}),
\begin{equation*} 
I\(\Qgrav,\{|\CN_{1ij}|\}_{\substack{i\in[\Hd]\\ j\in[\md(\Hd)]}},\ldots,\{|\CN_{\qd ij}|\}_{\substack{i\in[\Hd]\\ j\in[\md(\Hd)]}};\SCgrav\)\leq \QCC(\prot).
\end{equation*}
By the chain rule (Lemma~\ref{lem:I-chain-rule}) and the nonnegativity (Lemma~\ref{lem:I-nonneg}) properties of mutual information, for any random variables $X_1,X_2,Y$, we have $I\(X_1,X_2;Y\)=I\(X_2;Y\)+I\(X_1;Y|X_2\)\geq I\(X_2;Y\)$. Thus for every $k\in[\qd]$,
\begin{equation}\label{eqn:fidelity-random-pf-I-QCC}
I_k:=I\(\{|\CN_{kij}|\}_{\substack{i\in[\Hd]\\j\in[\md(\Hd)]}};\(\SCgrav_{ij}\)_{\substack{i\in[\Hd]\\j\in\md_k}}\)\leq I\(\{|\CN_{kij}|\}_{\substack{i\in[\Hd]\\ j\in[\md(\Hd)]}};\SCgrav\)\leq \QCC(\prot).
\end{equation}

We are now ready to upper bound the average fidelity. Recall Equation~\ref{eqn:fidelity-random-pf-upperbd-SC*Q-two-sums}:
\begin{equation*}
\Eover{\Qgrav,\SCgrav}\[F_{\SCgrav,\Qgrav}\]\leq\frac{1}{\Hd}\
\(\sum_{k\in[\qd]}\sigma_k\)+\frac{\sigma_{\qd+1}}{\Hd}
+\frac{3\SCcol}{2\sqrt{\Hd}}+\frac{1}{2\sqrt{\Hd}}.
\end{equation*}
By the bound on $\sigma_{\qd+1}$ in Equation~\ref{eqn:fidelity-random-pf-SC*Q-above-m(d)}, the bound on $\sigma_{k}$ for $k\in[\qd]$ in Equation~\ref{eqn:fidelity-random-pf-SC*Q-below-m(d)}, and the bound on $I_k$ for $k\in[\qd]$ in Equation~\ref{eqn:fidelity-random-pf-I-QCC},
\begin{align}
\Eover{\Qgrav,\SCgrav}\[F_{\SCgrav,\Qgrav}\]&\leq\frac{\qd\cdot\md(\Hd)}{\Hd}\(\epsilon(\Hd)\sqrt{\SCcol\Hd}+\delta(\Hd)\sqrt{\SCcol}+\frac{1+\SCcol}{2(1-\delta(\Hd))}-\frac{e^{\gamma_E+1}}{16\pi}\cdot e^{\frac{-2\cdot \QCC(\pi)}{\Hd\cdot\md(\Hd)}}\)\nonumber\\
&\quad\quad+\frac{(\Hd-\qd\cdot\md(\Hd))(1+\SCcol)}{2\Hd}+\frac{3\SCcol}{2\sqrt{\Hd}}+\frac{1}{2\sqrt{\Hd}}.\label{eqn:fidelity-random-pf-before-params}
\end{align}

Finally we instantiate all of the parameters to obtain the following bound. This part is nontrivial but technical so we placed it in Appendix~\ref{sec:appendix-tech}. This uses $\md(\Hd),\epsilon(\Hd),\delta(\Hd)$ from Proposition~\ref{thm:approx-Haar-U} and $\SCcol=O(1)$. For any constants $c>0$ and $0<\beta<1$, we have that for $\QCC(\pi)\leq\frac{\eta\cdot\Hd^2}{\log\Hd}$ for some constant $\eta$ that depends on $c,\beta$ (or equivalently for any algorithm $\A$ that makes at most $\frac{\eta\cdot\Hd^2}{2(\log\Hd)^2}$ queries to $\Qgrav$):
\begin{align*}
\Eover{\Qgrav,\SCgrav}\[F_{\SCgrav,\Qgrav}\]\leq \frac{1+\SCcol}{2}-\(\frac{e^{\gamma_E+1}}{16\pi}\cdot e^{-c}\,(1-\beta)\)+O\(\frac{1}{\log\Hd}\),
\end{align*}
as each query corresponds to $2\log_{2}d$ qubits of information. For sufficiently large $\Hd$, such that the $O\(\frac{1}{\log\Hd}\)$ function above is $\leq \beta\cdot\(\frac{e^{\gamma_E+1}}{16\pi}\cdot e^{-c}\,(1-\beta)\)$, we have Theorem~\ref{thm:fidelity-upperbd-random}:
\begin{align*}
\Eover{\Qgrav,\SCgrav}\[F_{\SCgrav,\Qgrav}\]\leq \frac{1+\SCcol}{2}-\(\frac{e^{\gamma_E+1}}{16\pi}\cdot e^{-c}\,(1-\beta)^2\).
\end{align*}
Here $c$ and $\beta$ can be arbitrarily small positive constants, giving upper bounds that are
$\approx \frac{1+\SCcol}{2}-\(\frac{e^{\gamma_E+1}}{16\pi}\) = \frac{1+\SCcol}{2}-0.0963$ 

\noindent
For instance, for $c=\beta=\frac{1}{10^3}$ we have that for $\QCC(\prot)\leq c\(\frac{t^2\cdot\Hd^2}{144\log\Hd}\)=\frac{\eta\cdot\Hd^2}{\log\Hd}$ for some positive constant $\eta$, the average fidelity 
$$\Eover{\Qgrav,\SCgrav}\[F_{\SCgrav,\Qgrav}\]\leq \frac{1+\SCcol}{2}-0.0960=\frac{1+\SCcol}{2}-\Omega(1).$$
\noindent 
This concludes the proof of Theorem~\ref{thm:fidelity-upperbd-random}.
 
\end{proof}

\begin{rem}
Our $\Omega(1)$ gap from the maximal value comes from considering how well $|\SCgrav_{ij}|$ approximates $|\Qgrav_{ij}|$ (starting from the application of the continuous Fano's inequality). We expect a better gap may be obtained by considering how well $\SCgrav_{ij}$ approximates the real and imaginary parts of $\Qgrav_{ij}$ separately. The real and imaginary parts each have an approximation by (real) Gaussian variables, for which some of our inequalities would be tighter. 
\end{rem}

\begin{rem} The sharp reader may ask if we need $\md(\Hd)$, the number of columns of $\Qgrav$ we can jointly approximate by complex Gaussians, to be very large? Armed now with the idea to partition and consider \emph{subsets} of $\Qgrav$, we may no longer need $\md(\Hd)$ to be significantly large to get a fidelity bound. (Without this idea, we do need $\md(\Hd)$ to be $\Omega\(\frac{\Hd}{\log\Hd}\)$ to obtain an initial fidelity bound of $1-\Omega\(\frac{1}{\log\Hd}\)$.) 
\end{rem}

\begin{rem}(A proof with quantum analogs of our tools?)
At the heart of our work is Theorem~\ref{thm:fidelity-upperbd-random} which shows a quantum algorithm cannot learn the exponential amount of randomized information in $\Qgrav$. The information theoretic tools in our proof work with the classical entropy $h$ of $\Qgrav$'s matrix elements rather than the quantum entropy $S$. Since our starting point is the QCC result bounding the (classical) mutual information $I(\Qgrav;\SCgrav)$, it is natural to work with the related classical entropy $h$, and this is how we proceed.
The tools of classical and quantum information theory tend to have analogs, but we leave open whether there is a proof of a result similar to Theorem~\ref{thm:fidelity-upperbd-random} using the quantum analogs. For instance, a proof using a bound on the \emph{quantum} mutual information (see e.g. \cite{NieChu11}) between $\Qgrav$ and $\SCgrav$ (see e.g. \cite{NieChu11}). Similarly to classical versions of Fano's inequality, this inequality relates the error between two random variables, expressed here as a fidelity, and the conditional entropy of one variable given the other, expressed here as a von Neumann entropy. A bound on the quantum mutual information may work well with quantum Fano's inequality, or finding a way to relate the classical mutual information to the quantum von Neumann entropy.
\end{rem}

\section{Hardness of Learning Pseudorandom Unitaries}\label{sec:fidelity-upperbd-pseudorandom}

In this section we show the hardness of learning a pseudorandom unitary acting on a pseudorandom state. In our application to physical systems in general and black holes in particular, this corresponds to the more realistic scenario where the system -- say, the black hole -- dynamics are efficiently generated but remain indistinguishable from Haar random to algorithms whose computational power is bounded polynomially in $\log\dim {\cal H}$. The setup of the main theorem is analogous to the setup in Section~\ref{sec:fidelity-upperbd-random} for truly Haar random dynamics. The primary difference lies in the computational complexity of the algorithm:  the bound for Haar random unitaries is for any learning algorithm bounded in communication complexity, whereas our theorem below for pseudorandom unitaries is for any algorithm which is also bounded in computational complexity.

Before continuing, us briefly comment on a general feature of our results. The proofs and results in this section are stated for `$\poly(\Hn)$-size quantum circuits' which are more general than $\poly(\Hn)$-time quantum algorithms. The latter often refers to a single `uniform' algorithm that works for all $\Hn$. Here $\poly(\Hn)$-size quantum circuits simply refers to a family of quantum algorithms, one for each $\Hn$, where asymptotically the $\Hn$'th circuit runs in $\poly(\Hn)$-time. This is a `non-uniform' model of computation since it allows there to be a different circuit for each $\Hn$. The circuits can also use `non-uniform' advice (an arbitrary quantum state, different for each $\Hn$) for their computations. So the results in this section hold for uniform $\poly(\Hn)$-time quantum algorithms as well.

In Section~\ref{sec:fidelity-defn-PRU-PRS} we give the definitions of pseudorandom unitaries and states, first defined by \cite{JiLiu18},  and state the main result of this section, Theorem~\ref{thm:fidelity-upperbd-PRU-PRS}. In Sections~\ref{sec:fidelity-PRU} and~\ref{sec:fidelity-PRS} we prove three lemmas which are used to  generalize our proof for unitaries and states to pseudorandom unitaries and states. In Section~\ref{sec:PR-thm-proof} we use these lemmas to prove Theorem~\ref{thm:fidelity-upperbd-PRU-PRS}.

\subsection{Definitions of Pseudorandom Unitaries and States}\label{sec:fidelity-defn-PRU-PRS}
In this section we state the definitions of computationally pseudorandom unitaries and states. We specify that the security property (pseudorandomness) holds against non-uniform adversaries (also called a family of circuits).\footnote{Note that in the following definitions and in our proofs, the algorithm $\Dist$  is \textit{not} the learning algorithm $\A$. The goal of $\Dist$ is to distinguish random vs pseudorandom objects and for this purpose, it may have access to multiple copies of some (fixed) state as `advice'. This does not affect our results where the learning algorithm $\A$ as always gets access to only one copy of the (randomly sampled) state $\ket{\psi}$ in the prediction phase.} 

\begin{defn}[Pseudorandom Unitaries (PRU) \cite{JiLiu18}]\label{defn:PRU}
Let $\secp$ be the security parameter. A distribution ensemble $\PRU=\{\PRU_\secp\}_\secp$ is a \emph{pseudorandom unitary ensemble} if there is a pair of quantum algorithms $\Key=\{\Key_{\lambda}\}_{\lambda}$ and $\Qalg=\{\Qalg_{\lambda}\}_{\lambda}$ such that the following hold. 
\begin{enumerate}
\item\textbf{Efficient Computation:} $\Key$ and $\Qalg$ each run in $\poly(\secp)$-time and for every $\secp\in\N$: 
\begin{itemize}
    \item $\Key_\secp$ outputs a classical key $\key\in\{0,1\}^*$ that defines a unitary $\pru_\key$.
    \item For every $\key$ in the image of $\Key_\secp$ and any state $\ket{\phi}$, $\Qalg_\secp(\key,\ket{\phi})$ outputs $\pru_\key\ket{\phi}$.
    \item For $\key\gets\Key_\secp$ the distribution of $\pru_\key$ is $\PRU_\secp$.
\end{itemize} 
\item\textbf{Computational Pseudorandomness:} For any family of $\poly(\secp)$-size quantum circuits $\Dist=\{\Dist_\secp\}_\secp$ making quantum queries to an oracle, there exists a negligible function $\neglf$ such that for every $\secp\in\N$,
$$\left|\Pr_{\key\gets\Key_\secp}\[\Dist^{\pru_\key}_\secp=1\]-\Pr_{\pru\gets\Haar}\[\Dist^{\pru}_\secp=1\]\right|\leq\neglf(\secp)$$
where $\Haar$ is the Haar measure over the unitary group.
\end{enumerate}

\end{defn}

\begin{defn}[Pseudorandom States (PRS) \cite{JiLiu18}]\label{defn:PRS}
Let $\secp$ be the security parameter. A distribution ensemble $\PRS=\{\PRS_\secp\}_\secp$ is a \emph{pseudorandom state ensemble} if there is a pair of quantum algorithms $\Key$ and $\Qalg$ such that the following hold.
\begin{enumerate}
\item\textbf{Efficient Generation:} $\Key$ and $\Qalg$ each run in $\poly(\secp)$-time and for every $\secp\in\N$:
\begin{itemize}
    \item $\Key_\secp$ outputs a classical key $\key\in\{0,1\}^*$.
    \item For every $\key$ in the image of $\Key_\secp$, $\Qalg_\secp(\key)$ outputs a quantum state $\ket{\prs_\key}$.
    \item For $\key\gets\Key_\secp$ the distribution of $\ket{\prs_\key}$ is $\PRS_\secp$.
\end{itemize} 
\item\textbf{Computational Pseudorandomness:} For any polynomial $t=t(\secp)$ and any family of $\poly(\secp)$-size quantum circuits $\Dist=\{\Dist_\secp\}_\secp$ there exists a negligible function $\neglf$ such that for every $\secp\in\N$,
$$\left|\Pr_{\key\gets\Key_\secp}\[\Dist\(\ket{\prs_\key}^{\otimes t(\secp)}\)=1\]-\Pr_{\ket{\prs}\gets\Haar}\[\Dist\(\ket{\prs}^{\otimes t(\secp)}\)=1\]\right|\leq\neglf(\secp)$$
where $\Haar$ is the Haar measure over the unit sphere. 
\end{enumerate}
\end{defn}

One can relax the notion of PRS by relaxing the requirement that computational pseudorandomness holds for any polynomial $k=k(\secp)$ copies of the state. For our results, it will suffice to just have PRS where pseudorandomness holds for $k=2$ copies of the state. 

\begin{rem}(On the security parameter)
Cryptographic primitives are usually parameterized by a security parameter which can be different, in this case, from the size of the pseudorandom unitaries or states generated. For simplicity, we define the security parameter in this section to correspond to the size of the unitary or state that is generated so $\PRU_\Hn$ consists of unitaries acting on $\Hn$ qubit states, and $\PRS_\Hn$ consists of $\Hn$ qubit states. Since the original security parameter and $\Hn$ are polynomially related, this re-indexing does not affect the efficiency or security of the constructions.
\end{rem}

We prove the following theorem:

\begin{thm}[Hardness of Learning Pseudorandom Unitaries]\label{thm:fidelity-upperbd-PRU-PRS}
For $\Hn\in\N$ let $\PRU=\{\PRU_\Hn\}_\Hn$ be any pseudorandom unitary (PRU) ensemble and let $\PRS=\{\PRS_\Hn\}_\Hn$ be any pseudorandom state (PRS) ensemble  (Definitions~\ref{defn:PRU} and~\ref{defn:PRS}) in a Hilbert space $\mathcal{H}$ of dimension $\Hd=2^\Hn$.
For any family of $\poly(\Hn)$-size quantum circuits $\A=\{\A_\Hn\}_\Hn$ that has quantum query (oracle) access to the pseudorandom unitary (Definition~\ref{defn:quantum-queries}), learns a unitary model\footnote{``Learning a model in $\C^{\Hd\times\Hd}$'' means that $\A$ learns the description of a quantum circuit. The unitary $\Ugrav\in\C^{\Hd\times\Hd}$ in the theorem statement is defined to be the operator implemented by the circuit that $\A$ learns. Since we consider $\A$ running in time $\poly(\Hn)$, the models it may learn for $\Qgrav$ will also be $\poly(\Hn)$-size quantum circuits.}
$\Ugrav$, and given $\ket{\BH}$ outputs $\Ugrav\ket{\BH}$, the average fidelity
$$\Eover{\substack{\Qgrav\gets\PRU_\Hn\\ \ket{\BH}\gets\PRS_\Hn\\ \Ugrav\ket{\BH}\gets\A^\Qgrav_\Hn(\ket{\BH})}}\[F\(\Ugrav\ket{\BH},\Qgrav\ket{\BH}\)\]\leq 1-\Omega\(1\)$$
where $F\(\Ugrav\ket{\BH},\Qgrav\ket{\BH}\)=\left|\bra{\BH}\Ugrav^\dagger \Qgrav\ket{\BH}\right|^2$ is the squared inner product (fidelity for pure states) between $\A$'s prediction $\Ugrav\ket{\BH}$ and the fundamental time evolution $\Qgrav\ket{\BH}$.
\end{thm}

For the reader concerned about existence and constructions of PRUs and PRSs, we make a brief  remark before proceeding.

\paragraph{Constructions of PRU and PRS.} Note that a construction of pseudorandom unitaries gives a construction of pseudorandom states (e.g. $U_k\ket{0}$ for a random $k\gets\Key_\secp$). Refs. \cite{JiLiu18, BraShm20} show that PRS can be constructed from any quantum-secure one-way function, and \cite{JiLiu18} gives candidate constructions for PRUs based on similar ideas. Ref. \cite{BouFef19} gives evidence for the pseudorandomness of a different construction of states (in the context of the AdS/CFT correspondence in quantum gravity). Much of the prior work on quantum pseudorandomness focuses on statistically secure quantum $k$-designs, and has only recently begun exploring potential constructions and applications of computationally secure PRU and PRS.

\subsection{From Random to Pseudorandom Unitary}\label{sec:fidelity-PRU}
\begin{lem}\label{lem:fidelity-rand-PRU}
For any pseudorandom unitary (PRU) ensemble $\PRU$ (Definition~\ref{defn:PRU}), any distribution ensemble of states $\Phi$ with probability measure $\nu_\Phi$, and any family of $\poly(\Hn)$-size quantum circuits $\A=\{\A_\Hn\}_\Hn$, there exists a negligible function $\neglf$ such that for every $\Hn\in\N$,
$$\left|\Eover{\substack{U\gets\Haar\\ \ket{\phi}\gets\Phi\\ \ket{\phi'}\gets\A^\Qgrav_\Hn(\ket{\phi})}}\[F\(\ket{\phi'},\Qgrav\ket{\phi}\)\]-\Eover{\substack{U\gets\PRU\\ \ket{\phi}\gets\Phi\\ \ket{\phi'}\gets\A^\Qgrav_\Hn(\ket{\phi})}}\[F\(\ket{\phi'},\Qgrav\ket{\phi}\)\]\right|\leq\neglf(\Hn)$$
where $F\(\ket{\phi'},\Qgrav\ket{\phi}\)=\left|\bra{\phi'} \Qgrav\ket{\phi}\right|^2$ and $\Haar$ is the Haar measure over $\Hd=2^\Hn$ dimensional unitaries.
\end{lem}
\begin{proof}
For readability, we abbreviate $\ket{\phi'}\gets\A^\Qgrav_\Hn(\ket{\phi})$ by $\ket{\phi'}\gets\A^\Qgrav$ and $F\(\ket{\phi'},\Qgrav\ket{\phi}\)$ by $F$.

Assume towards contradiction that there exists a pseudorandom unitary ensemble $\PRU=\{\PRU_\Hn\}_\Hn$ (Definition~\ref{defn:PRU}), a distribution ensemble of states $\Phi=\{\Phi_\Hn\}_\Hn$ with probability measure $\nu_\Phi=\{\nu_\Hn\}_\Hn$, a family of $\poly(\Hn)$-size quantum circuits $\A=\{\A_\Hn\}_\Hn$, and a polynomial $\p$ such that for infinitely many $\Hn\in\N$,
$$\left|\Eover{\substack{U\gets\Haar\\ \ket{\phi}\gets\Phi\\ \ket{\phi'}\gets\A^\Qgrav_\Hn(\ket{\phi})}}\[F\(\ket{\phi'},\Qgrav\ket{\phi}\)\]-\Eover{\substack{U\gets\PRU\\ \ket{\phi}\gets\Phi\\ \ket{\phi'}\gets\A^\Qgrav_\Hn(\ket{\phi})}}\[F\(\ket{\phi'},\Qgrav\ket{\phi}\)\]\right|\geq\frac{1}{\p(\Hn)}$$
or equivalently
$$\left|\;\int\limits_{\C^\Hd}
\Eover{\substack{U\gets\Haar\\ \ket{\phi'}\gets\A^\Qgrav}}\[F\]-\Eover{\substack{U\gets\PRU\\ \ket{\phi'}\gets\A^\Qgrav}}\[F\]
\;\;d\nu_\Phi(\ket{\phi})\;\right|\geq\frac{1}{\p(\Hn)}.$$

For any such $\Hn$, we construct a circuit $\Dist_\Hn$. This will construct a non-uniform distinguisher $\Dist=\{\Dist_\Hn\}_\Hn$ that breaks the pseudorandomness of $\PRU$ (Definition~\ref{defn:PRU}).
Consider any such $\Hn$. By the triangle inequality, we have
$$\int\limits_{\C^\Hd}\;\left|
\Eover{\substack{U\gets\Haar\\ \ket{\phi'}\gets\A^\Qgrav}}\[F\]-\Eover{\substack{U\gets\PRU\\ \ket{\phi'}\gets\A^\Qgrav}}\[F\]
\right|\;d\nu_\Phi(\ket{\phi})\geq\frac{1}{\p(\Hn)}.$$
Therefore there exists an $\Hn$ qubit state $\ket{\phi_\Hn}\in\C^\Hd$ such that
\begin{equation}\label{eqn:PRU-E-F-1/poly}
\left|
\Eover{\substack{U\gets\Haar\\ \ket{\phi'}\gets\A^\Qgrav_\Hn(\ket{\phi_\Hn})}}\[F\(\ket{\phi'},\Qgrav\ket{\phi_\Hn}\)\]-\Eover{\substack{U\gets\PRU\\ \ket{\phi'}\gets\A^\Qgrav_\Hn(\ket{\phi_\Hn})}}\[F\(\ket{\phi'},\Qgrav\ket{\phi_\Hn}\)\]
\right|\geq\frac{1}{\p(\Hn)}.
\end{equation}

Let $\Dist_\Hn$ have $\ket{\phi_\Hn}^{\otimes 2}$ as advice and proceed as follows.
\begin{enumerate}
\item $\Dist^\Qgrav_\Hn$ runs $\A_\Hn$ and simulates the oracle for $\A_\Hn$: it receives each query from $\A_\Hn$, sends the query to its oracle $\Qgrav$, and returns the response to $\A_\Hn$.
\item\label{item:PRU-A-and-U-outputs} $\Dist^\Qgrav_\Hn$ gives $\ket{\phi_\Hn}$ to $\A_\Hn$ and receives back $\ket{\phi'}$. It also gives $\ket{\phi_\Hn}$ to its oracle $\Qgrav$ and receives back $\Qgrav\ket{\phi_\Hn}$.
\item $\Dist^\Qgrav_\Hn$ performs the $\SWAP$ test\footnote{The $\SWAP$ test is a procedure in quantum computation that measures how much two quantum states differ. It takes as input two states 
$\ket{\phi}$ and 
$\ket{\psi}$ and outputs a Bernoulli random variable that is $1$ with probability 
${\frac {1}{2}}-{\frac {1}{2}}{|\langle \psi |\phi \rangle |}^{2}$.} between $\ket{\phi'}$ and $\Qgrav\ket{\phi_\Hn}$ and outputs the bit from the $\SWAP$ test.
\end{enumerate}
Note that $\Dist=\{\Dist_\Hn\}_\Hn$ runs in $\poly(\Hn)$-time since $\A=\{\A_\Hn\}_\Hn$ runs in $\poly(\Hn)$-time and the $\SWAP$ test is efficient.

Next we show $\Dist_\Hn$ distinguishes between $\Qgrav\gets\Haar$ and $\Qgrav\gets\PRU$. For any unitary $\Qgrav$ as the oracle and any state $\ket{\phi'}$ output by $\A_\Hn$ in Step~\ref{item:PRU-A-and-U-outputs},
\begin{align*}
\Pr\[\SWAP(\ket{\phi'},\Qgrav\ket{\phi_\Hn})=1\]&=\frac{1}{2}-\frac{1}{2}\left|\bra{\phi'}\Qgrav\ket{\phi_\Hn}\right|^2\\
&=\frac{1}{2}-\frac{1}{2}F(\ket{\phi'},\Qgrav\ket{\phi_\Hn}).
\end{align*}
Thus for Haar random $\Qgrav\gets\mu$ and the corresponding distribution of $\ket{\phi'}$ that $\A_\Hn$ outputs in Step~\ref{item:PRU-A-and-U-outputs},
\begin{align*}
\Pr_{\Qgrav\gets\mu}\[\Dist^\Qgrav_\Hn=1\]&=\Eover{\substack{\Qgrav\gets\mu\\ \ket{\phi'}\gets\A_\Hn\text{ in Step~\ref{item:PRU-A-and-U-outputs}}}}\[\Pr\[\SWAP(\ket{\phi'},\Qgrav\ket{\phi_\Hn})=1\]\]\\
&=\frac{1}{2}-\frac{1}{2}\Eover{\substack{\Qgrav\gets\mu\\ \ket{\phi'}\gets\A_\Hn\text{ in Step~\ref{item:PRU-A-and-U-outputs}}}}\[F(\ket{\phi'},\Qgrav\ket{\phi_\Hn})\]\\
&=\frac{1}{2}-\frac{1}{2}\Eover{\substack{U\gets\Haar\\ \ket{\phi'}\gets\A^\Qgrav_\Hn(\ket{\phi_\Hn})}}\[F\(\ket{\phi'},\Qgrav\ket{\phi_\Hn}\)\].
\end{align*}
where the last line transitions from $\ket{\phi'}$ and $\Qgrav\ket{\phi_\Hn}$ obtained in Step~\ref{item:PRU-A-and-U-outputs} of $\Dist^\Qgrav_\Hn$, to $\ket{\phi'}$ and $\Qgrav\ket{\phi_\Hn}$ defined in the original context where $\A$ has $\Qgrav$ as an oracle.
This holds since $\Dist^\Qgrav_\Hn$ emulates the oracle $\Qgrav\sim\Haar$ for $\A$ and gives it $\ket{\phi_\Hn}$ to predict for.

Similarly, for pseudorandom $\Qgrav\gets\PRU$ and the corresponding distribution of $\ket{\phi'}$ that $\A_\Hn$ outputs in Step~\ref{item:PRU-A-and-U-outputs},
$$\Pr_{\Qgrav\gets\PRU}\[\Dist^\Qgrav_\Hn=1\]=\frac{1}{2}-\frac{1}{2}\Eover{\substack{U\gets\PRU\\ \ket{\phi'}\gets\A^\Qgrav_\Hn(\ket{\phi_\Hn})}}\[F\(\ket{\phi'},\Qgrav\ket{\phi_\Hn}\)\].$$
Therefore
\begin{align*}
&\left|\Pr_{\Qgrav\gets\mu}\[\Dist^\Qgrav_\Hn=1\]-\Pr_{\Qgrav\gets\PRU}\[\Dist^\Qgrav_\Hn=1\]\right|\\
&\quad=\frac{1}{2}\left|\Eover{\substack{U\gets\Haar\\ \ket{\phi'}\gets\A^\Qgrav_\Hn(\ket{\phi_\Hn})}}\[F\(\ket{\phi'},\Qgrav\ket{\phi_\Hn}\)\]-\Eover{\substack{U\gets\PRU\\ \ket{\phi'}\gets\A^\Qgrav_\Hn(\ket{\phi_\Hn})}}\[F\(\ket{\phi'},\Qgrav\ket{\phi_\Hn}\)\]\right|.
\end{align*}
By Equation~\ref{eqn:PRU-E-F-1/poly},
$$\left|\Pr_{\Qgrav\gets\mu}\[\Dist^\Qgrav_\Hn=1\]-\Pr_{\Qgrav\gets\PRU}\[\Dist^\Qgrav_\Hn=1\]\right|\geq\frac{1}{2\,\p(\Hn)}.$$
This shows that for infinitely many $\Hn\in\N$, $\Dist=\{\Dist_\Hn\}_\Hn$ distinguishes between $\Qgrav\gets\Haar$ and $\Qgrav\gets\PRU$. Since $\Dist$ runs in $\poly(\Hn)$-time, this contradicts the pseudorandomness of $\PRU$ (Definition~\ref{defn:PRU}). Thus based on the pseudorandomness of $\PRU$, Lemma~\ref{lem:fidelity-rand-PRU} holds.
\end{proof}

\subsection{From Random to Pseudorandom State}\label{sec:fidelity-PRS}

We begin by formally defining distribution ensembles of unitaries that are efficiently computable. This property will be used in Lemma~\ref{lem:fidelity-rand-PRS} below, and is satisfied by any pseudorandom unitary ensemble (Definition~\ref{defn:PRU}).
\begin{defn}(Efficiently Computable Unitaries)\label{defn:efficient-U}
Let $\cD=\{\cD_\Hn\}_\Hn$ be a distribution ensemble of unitaries where each $\cD_\Hn$ is a distribution over $2^\Hn$ dimensional unitaries. The distribution ensemble of unitaries $\cD$ is \emph{efficiently computable} if for every sequence of unitaries $\{U_\Hn\}_\Hn$ where each $U_\Hn$ has positive support in $\cD_\Hn$, there exists a sequence of $\poly(\Hn)$-size quantum circuits $\Qcir=\{\Qcir_\Hn\}_\Hn$ such that for every $\Hn$, $\Qcir_\Hn(\ket{\phi})=U_\Hn\ket{\phi}$ for every $\Hn$ qubit state $\ket{\phi}$.
\end{defn}

\begin{lem}\label{lem:fidelity-rand-PRS}
For any efficiently computable distribution ensemble of unitaries $\cD$ (Definition~\ref{defn:efficient-U}) with probability measure $\nu_\cD$, any pseudorandom state (PRS) ensemble $\PRS$ (Definition~\ref{defn:PRS}), and any family of $\poly(\Hn)$-size quantum circuits $\A=\{\A_\Hn\}_\Hn$, there exists a negligible function $\neglf$ such that for every $\Hn\in\N$,
$$\left|\Eover{\substack{\Qgrav\gets\cD\\ \ket{\BH}\gets\Haar\\ \ket{\BH'}\gets\A^\Qgrav_\Hn(\ket{\BH})}}\[F\(\ket{\BH'},\Qgrav\ket{\BH}\)\]-\Eover{\substack{\Qgrav\gets\cD\\ \ket{\BH}\gets\PRS\\ \ket{\BH'}\gets\A^\Qgrav_\Hn(\ket{\BH})}}\[F\(\ket{\BH'},\Qgrav\ket{\BH}\)\]\right|\leq\neglf(\Hn)$$
where $F\(\ket{\BH'},\Qgrav\ket{\BH}\)=\left|\bra{\BH'} \Qgrav\ket{\BH}\right|^2$ and $\Haar$ is the Haar measure over $\Hd=2^\Hn$ dimensional quantum states.
\end{lem}
\begin{proof}
The proof of Lemma~\ref{lem:fidelity-rand-PRS} is similar to the proof of Lemma~\ref{lem:fidelity-rand-PRU} with the pseudorandom states here taking place of the pseudorandom unitaries in that proof. Nonetheless, we will provide the proof below for completeness.

For readability, we abbreviate $\ket{\BH'}\gets\A^\Qgrav_\Hn(\ket{\BH})$ by $\ket{\BH'}\gets\A^\Qgrav$ and $F\(\ket{\BH'},\Qgrav\ket{\BH}\)$ by $F$.

Assume towards contradiction that there exists an efficiently computable distribution ensemble of unitaries $\cD=\{\cD_\Hn\}_\Hn$ (Definition~\ref{defn:efficient-U}) with probability measure $\nu_\cD=\{\nu_\Hn\}_\Hn$, a pseudorandom state ensemble $\PRS=\{\PRS_\Hn\}_\Hn$ (Definition~\ref{defn:PRS}), a family of $\poly(\Hn)$-size quantum circuits $\A=\{\A_\Hn\}_\Hn$, and a polynomial $\p$ such that for infinitely many $\Hn\in\N$,
$$\left|\Eover{\substack{\Qgrav\gets\cD\\ \ket{\BH}\gets\Haar\\ \ket{\BH'}\gets\A^\Qgrav_\Hn(\ket{\BH})}}\[F\(\ket{\BH'},\Qgrav\ket{\BH}\)\]-\Eover{\substack{\Qgrav\gets\cD\\ \ket{\BH}\gets\PRS\\ \ket{\BH'}\gets\A^\Qgrav_\Hn(\ket{\BH})}}\[F\(\ket{\BH'},\Qgrav\ket{\BH}\)\]\right|\geq\frac{1}{\p(\Hn)}$$
or equivalently
$$\left|\;\int\limits_{\C^{\Hd\times\Hd}}
\Eover{\substack{\ket{\BH}\gets\Haar\\ \ket{\BH'}\gets\A^\Qgrav}}\[F\]-\Eover{\substack{\ket{\BH}\gets\PRS\\ \ket{\BH'}\gets\A^\Qgrav}}\[F\]
\;\;d\nu_\cD(\Qgrav)\;\right|\geq\frac{1}{\p(\Hn)}.$$

For any such $\Hn$, we construct a circuit $\Dist_\Hn$. This will construct a non-uniform distinguisher $\Dist=\{\Dist_\Hn\}_\Hn$ that breaks the pseudorandomness of $\PRS$ (Definition~\ref{defn:PRS}).
Consider any such $\Hn$. By the triangle inequality, we have
$$\int\limits_{\C^{\Hd\times\Hd}}\;\left|
\Eover{\substack{\ket{\BH}\gets\Haar\\ \ket{\BH'}\gets\A^\Qgrav}}\[F\]-\Eover{\substack{\ket{\BH}\gets\PRS\\ \ket{\BH'}\gets\A^\Qgrav}}\[F\]
\right|\;d\nu_\cD(\Qgrav)\geq\frac{1}{\p(\Hn)}.$$
Therefore there exists a $2^\Hn$ dimensional unitary $\Qgrav_\Hn\in\C^{\Hd\times\Hd}$ such that
\begin{equation}\label{eqn:PRS-E-F-1/poly}
\left|
\Eover{\substack{\ket{\BH}\gets\Haar\\ \ket{\BH'}\gets\A^{\Qgrav_\Hn}_\Hn(\ket{\BH})}}\[F\(\ket{\BH'},\Qgrav_\Hn\ket{\BH}\)\]-\Eover{\substack{\ket{\BH}\gets\PRS\\ \ket{\BH'}\gets\A^{\Qgrav_\Hn}_\Hn(\ket{\BH})}}\[F\(\ket{\BH'},\Qgrav_\Hn\ket{\BH}\)\]
\right|\geq\frac{1}{\p(\Hn)}.
\end{equation}
Since $\cD$ is efficiently computable (Definition~\ref{defn:efficient-U}), there exists a sequence $\Qcir=\{\Qcir_\Hn\}_\Hn$ of $\poly(\Hn)$-size quantum circuits where $\Qcir_\Hn$ computes $\Qgrav_\Hn$ for the infinitely many such $\Hn$.

Let $\Dist_\Hn\(\ket{\BH}^{\otimes 2}\)$, given as input two copies of a state $\ket{\BH}$, proceed as follows.
\begin{enumerate}
\item $\Dist_\Hn$ runs $\A_\Hn$ and simulates the oracle for $\A_\Hn$: it receives each query $\ket{\phi_i}$ from $\A_\Hn$, runs $\Qcir_\Hn(\ket{\phi_i})$ to obtain $\Qgrav_\Hn\ket{\phi_i}$, and returns $\Qgrav_\Hn\ket{\phi_i}$ to $\A_\Hn$.
\item\label{item:PRS-A-and-U-outputs} $\Dist_\Hn$ gives one copy of $\ket{\BH}$ to $\A_\Hn$ and receives back $\ket{\BH'}$. It uses the second copy to run $\Qcir_\Hn(\ket{\BH})$ and obtains $\Qgrav_\Hn\ket{\BH}$.
\item $\Dist_\Hn$ performs the $\SWAP$ test between $\ket{\BH'}$ and $\Qgrav_\Hn\ket{\BH}$ and outputs the bit from the $\SWAP$ test.
\end{enumerate}
Note that $\Dist=\{\Dist_\Hn\}_\Hn$ runs in $\poly(\Hn)$-time since $\A=\{\A_\Hn\}_\Hn$ runs in $\poly(\Hn)$-time and $\Dist$ runs $\Qcir=\{\Qcir_\Hn\}_\Hn$, which is $\poly(\Hn)$-size, at most $\poly(\Hn)$ times since $\A$ can only make polynomially many queries. Additionally the $\SWAP$ test is efficient.

Next we show $\Dist_\Hn$ distinguishes between $\ket{\BH}\gets\Haar$ and $\ket{\BH}\gets\PRS$. For any input state $\ket{\BH}$ and any state $\ket{\BH'}$ output by $\A_\Hn$ in Step~\ref{item:PRS-A-and-U-outputs},
\begin{align*}
\Pr\[\SWAP(\ket{\BH'},\Qgrav_\Hn\ket{\BH})=1\]&=\frac{1}{2}-\frac{1}{2}\left|\bra{\BH'}\Qgrav_\Hn\ket{\BH}\right|^2\\
&=\frac{1}{2}-\frac{1}{2}F(\ket{\BH'},\Qgrav_\Hn\ket{\BH}).
\end{align*}
Thus for Haar random $\ket{\BH}\gets\mu$ and the corresponding distribution of $\ket{\BH'}$ that $\A_\Hn$ outputs in Step~\ref{item:PRS-A-and-U-outputs},
\begin{align*}
\Pr_{\ket{\BH}\gets\mu}\[\Dist_\Hn\(\ket{\BH}^{\otimes 2}\)=1\]&=\Eover{\substack{\ket{\BH}\gets\mu\\ \ket{\BH'}\gets\A_\Hn\text{ in Step~\ref{item:PRS-A-and-U-outputs}}}}\[\Pr\[\SWAP(\ket{\BH'},\Qgrav_\Hn\ket{\BH})=1\]\]\\
&=\frac{1}{2}-\frac{1}{2}\Eover{\substack{\ket{\BH}\gets\mu\\ \ket{\BH'}\gets\A_\Hn\text{ in Step~\ref{item:PRS-A-and-U-outputs}}}}\[F(\ket{\BH'},\Qgrav_\Hn\ket{\BH})\]\\
&=\frac{1}{2}-\frac{1}{2}\Eover{\substack{\ket{\BH}\gets\Haar\\ \ket{\BH'}\gets\A^{\Qgrav_\Hn}_\Hn(\ket{\BH})}}\[F\(\ket{\BH'},\Qgrav_\Hn\ket{\BH}\)\].
\end{align*}
where the last line transitions from $\ket{\BH'}$ and $\Qgrav_\Hn\ket{\BH}$ obtained in Step~\ref{item:PRS-A-and-U-outputs} of $\Dist_\Hn$, to $\ket{\BH'}$ and $\Qgrav_\Hn\ket{\BH}$ defined in the original context where $\A$ has $\Qgrav_\Hn$ as an oracle.
This holds since $\Dist_\Hn$ emulates the oracle $\Qgrav_\Hn$ for $\A$ and gives it $\ket{\BH}\sim\Haar$ to predict for.

Similarly, for pseudorandom $\ket{\BH}\gets\PRS$ and the corresponding distribution of $\ket{\BH'}$ that $\A_\Hn$ outputs in Step~\ref{item:PRS-A-and-U-outputs},
$$\Pr_{\ket{\BH}\gets\PRS}\[\Dist_\Hn\(\ket{\BH}^{\otimes 2}\)=1\]=\frac{1}{2}-\frac{1}{2}\Eover{\substack{\ket{\BH}\gets\PRS\\ \ket{\BH'}\gets\A^{\Qgrav_\Hn}_\Hn(\ket{\BH})}}\[F\(\ket{\BH'},\Qgrav_\Hn\ket{\BH}\)\].$$

Therefore
\begin{align*}
&\left|\Pr_{\ket{\BH}\gets\mu}\[\Dist_\Hn\(\ket{\BH}^{\otimes 2}\)=1\]-\Pr_{\ket{\BH}\gets\PRS}\[\Dist_\Hn\(\ket{\BH}^{\otimes 2}\)=1\]\right|\\
&\quad=\frac{1}{2}\left|\Eover{\substack{\ket{\BH}\gets\Haar\\ \ket{\BH'}\gets\A^{\Qgrav_\Hn}_\Hn(\ket{\BH})}}\[F\(\ket{\BH'},\Qgrav_\Hn\ket{\BH}\)\]-\Eover{\substack{\ket{\BH}\gets\PRS\\ \ket{\BH'}\gets\A^{\Qgrav_\Hn}_\Hn(\ket{\BH})}}\[F\(\ket{\BH'},\Qgrav_\Hn\ket{\BH}\)\]\right|.
\end{align*}
By Equation~\ref{eqn:PRS-E-F-1/poly},
$$\left|\Pr_{\ket{\BH}\gets\mu}\[\Dist_\Hn\(\ket{\BH}^{\otimes 2}\)=1\]-\Pr_{\ket{\BH}\gets\PRS}\[\Dist_\Hn\(\ket{\BH}^{\otimes 2}\)=1\]\right|\geq\frac{1}{2\,\p(\Hn)}.$$
This shows that for infinitely many $\Hn\in\N$, $\Dist=\{\Dist_\Hn\}_\Hn$ distinguishes between $\ket{\BH}\gets\Haar$ and $\ket{\BH}\gets\PRS$. Since $\Dist$ runs in $\poly(\Hn)$-time, this contradicts the pseudorandomness of $\PRS$ (Definition~\ref{defn:PRS}). Thus based on the pseudorandomness of $\PRS$, Lemma~\ref{lem:fidelity-rand-PRS} holds.
\end{proof}

\subsection{Proof of Theorem~\ref{thm:fidelity-upperbd-PRU-PRS}}\label{sec:PR-thm-proof}
Our final lemma in this section establishes that the difference in predictions in the random and pseudorandom settings is negligible.

\begin{lem}\label{lem:neglPRU} For all quantities defined as in the statement of Theorem~\ref{thm:fidelity-upperbd-PRU-PRS}, there exist negligible functions $\eta$ and $\eta'$ such that for every $n \in \mathbb{N}$,
\begin{equation}
\left|\Eover{\substack{U\gets\Haar\\ \ket{\phi}\gets\Haar\\ \Ugrav\ket{\phi}\gets\A^\Qgrav_\Hn(\ket{\phi})}}\[F\(\Ugrav\ket{\phi},\Qgrav\ket{\phi}\)\]-\Eover{\substack{\Qgrav\gets\PRU\\ \ket{\BH}\gets\PRS\\ \Ugrav\ket{\BH}\gets\A^\Qgrav_\Hn(\ket{\BH})}}\[F\(\Ugrav\ket{\BH},\Qgrav\ket{\BH}\)
\]\right|\leq\neglf(\Hn)+\neglf'(\Hn).
\end{equation}
\end{lem}

\begin{proof}[Proof of Lemma~\ref{lem:neglPRU}]
Consider any PRU ensemble $\PRU$, any PRS ensemble $\PRS$, and any family of $\poly(\Hn)$-size quantum circuits $\A=\{\A_\Hn\}_\Hn$ that given a state $\ket{\BH}$ outputs $\Ugrav\ket{\BH}$ where $\Ugrav$ is a unitary.

By Lemma~\ref{lem:fidelity-rand-PRU} for the PRU ensemble $\PRU$, the distribution ensemble of Haar random states as $\Phi$ (so $\nu_\Phi=\Haar$), and the family of quantum circuits $\A=\{\A_\Hn\}_\Hn$, there exists a negligible function $\neglf$ such that for every $\Hn\in\N$,
$$\left|\Eover{\substack{U\gets\Haar\\ \ket{\phi}\gets\Haar\\ \Ugrav\ket{\phi}\gets\A^\Qgrav_\Hn(\ket{\phi})}}\[F\(\Ugrav\ket{\phi},\Qgrav\ket{\phi}\)\]-\Eover{\substack{U\gets\PRU\\ \ket{\phi}\gets\Haar\\ \Ugrav\ket{\phi}\gets\A^\Qgrav_\Hn(\ket{\phi})}}\[F\(\Ugrav\ket{\phi},\Qgrav\ket{\phi}\)\]\right|\leq\neglf(\Hn).$$

Note that any PRU ensemble (Definition~\ref{defn:PRU}) is an efficiently computable distribution ensemble of unitaries (Definition~\ref{defn:efficient-U}).
By Lemma~\ref{lem:fidelity-rand-PRS} for the distribution ensemble of unitaries $\cD=\PRU$, the PRS ensemble $\PRS$, and the family of quantum circuits $\A=\{\A_\Hn\}_\Hn$, there exists a negligible function $\neglf'$ such that for every $\Hn\in\N$,
$$\left|\Eover{\substack{\Qgrav\gets\PRU\\ \ket{\BH}\gets\Haar\\ \Ugrav\ket{\BH}\gets\A^\Qgrav_\Hn(\ket{\BH})}}\[F\(\Ugrav\ket{\BH},\Qgrav\ket{\BH}\)\]-\Eover{\substack{\Qgrav\gets\PRU\\ \ket{\BH}\gets\PRS\\ \Ugrav\ket{\BH}\gets\A^\Qgrav_\Hn(\ket{\BH})}}\[F\(\Ugrav\ket{\BH},\Qgrav\ket{\BH}\)\]\right|\leq\neglf'(\Hn).$$
Therefore for every $\Hn\in\N$,
\begin{equation}\label{eqn:PR-fidelity-diff-rand-PR-negl}
\left|\Eover{\substack{U\gets\Haar\\ \ket{\phi}\gets\Haar\\ \Ugrav\ket{\phi}\gets\A^\Qgrav_\Hn(\ket{\phi})}}\[F\(\Ugrav\ket{\phi},\Qgrav\ket{\phi}\)\]-\Eover{\substack{\Qgrav\gets\PRU\\ \ket{\BH}\gets\PRS\\ \Ugrav\ket{\BH}\gets\A^\Qgrav_\Hn(\ket{\BH})}}\[F\(\Ugrav\ket{\BH},\Qgrav\ket{\BH}\)
\]\right|\leq\neglf(\Hn)+\neglf'(\Hn).
\end{equation}
\end{proof}

We now combine Lemmas~\ref{lem:fidelity-rand-PRU}, \ref{lem:fidelity-rand-PRS}, and \ref{lem:neglPRU} to show the main result of Section~\ref{sec:fidelity-upperbd-pseudorandom}: the hardness of learning a pseudorandom unitary applied to a pseudorandom state.

\begin{proof}[Proof of Theorem~\ref{thm:fidelity-upperbd-PRU-PRS}]
By Theorem~\ref{thm:fidelity-upperbd-random},
$$\Eover{\substack{\Qgrav\gets\Haar\\
\ket{\phi}\gets\Haar\\
\Ugrav\ket{\phi}\gets\A^\Qgrav_\Hn(\ket{\phi})}}\[F\(\Ugrav\ket{\phi},\Qgrav\ket{\phi}\)\]\leq 1-\Omega\(1\).$$
Combining this with Equation~\ref{eqn:PR-fidelity-diff-rand-PR-negl}, we have the desired bound
$$\Eover{\substack{\Qgrav\gets\PRU\\ \ket{\BH}\gets\PRS\\ \Ugrav\ket{\BH}\gets\A^\Qgrav_\Hn(\ket{\BH})}}\[F\(\Ugrav\ket{\BH},\Qgrav\ket{\BH}\)\]\leq 1-\Omega\(1\).$$
\end{proof}

\section{Hardness of Learning for Algorithms Predicting with Quantum Channels}\label{sec:general-learner} 

In this section we relax the requirement that the algorithm predicts using a single linear operator and show the hardness of learning unitary operators for quantum algorithms of bounded complexity that are allowed to use \emph{any} quantum operation -- quantum channel -- to predict the time evolution. We will therefore no longer assume that the learning algorithm produces a single linear operator as its hypothesis of the fundamental dynamics as was assumed for Theorems~\ref{thm:fidelity-upperbd-random} and~\ref{thm:fidelity-upperbd-PRU-PRS}. We now only assume the algorithm $\A$ produces a quantum channel $\qop$ as its hypothesis of the $\Qgrav$ and given the quantum state $\ket{\BH}$, it outputs $\rho=\qop(\ket{\BH}\bra{\BH})$ as its prediction for the time evolution of $\ket{\BH}$ (which in the fundamental description is $\Qgrav\ket{\BH}$). The state $\rho$ predicted by $\A$ may be a mixed state even if the initial state $\ket{\psi}$ is pure and may also be in a different, e.g. larger, Hilbert space. Below we discuss the motivation for allowing such predictions in the context of the black hole information paradox.

Our result in this section, Theorem~\ref{thm:fidelity-general-learners}, is the hardness of learning a randomly sampled unitary from a distribution with sufficiently high differential entropy, for a quantum algorithm predicting a general quantum channel and making a bounded number of queries to $\Qgrav$.  We show the techniques for Theorem~\ref{thm:fidelity-upperbd-random} work for predictions coming from such general quantum maps (which includes non-unitary maps and decohering channels) between the fundamental Hilbert space and any Hilbert space used by the algorithm. This also illustrates the derivation of a bound for general unitary ensembles as described above, not just Haar random or pseudorandom unitaries: the only necessary input is the differential entropy of $U$.

\begin{thm}[Hardness of Learning for Algorithms Predicting with Quantum Channels]\label{thm:fidelity-general-learners}
For $\Hn\in\N$ let $\mathcal{H}$ be a Hilbert space of dimension $\Hd=2^\Hn$, $\mathcal{U}$ be any distribution of unitary operators acting on $\Hil$, and $\Qgrav$ be a unitary randomly sampled from $\mathcal{U}$.
In the following, let $\Hil'$ denote any other Hilbert space and let $\mathcal{M}$ denote any fixed\footnote{By fixed, we mean that the map $\mathcal{M}$ is independent of the instance of $\Qgrav$, $\ket{\BH}$, and $\qop$ that is sampled or output; this is important also for consistency of our results with those of~\cite{AkeEng22}.} CPTP map from $\Hil'$ to $\Hil$. For any quantum algorithm $\A$ with quantum query access to $\Qgrav$ that produces a (circuit implementation of a) CPTP map $\qop:\Hil\to\Hil'$, and then given a state $\ket{\BH}$ outputs $\qop(\ket{\BH}\bra{\BH})$, for $\ket{\BH}$ a Haar random state in $\Hil$, the average fidelity between the image of $\A$'s prediction under ${\cal M}$, i.e. ${\cal M}\left (\qop(\ket{\BH}\bra{\BH}) \right)$ and  $\Qgrav\ket{\BH}\in\Hil$ is
$$\Eover{\substack{\Qgrav\gets\mathcal{U}\\
\ket{\BH}\gets\Haar\\
\qop(\ket{\BH}\bra{\BH})\gets\A^\Qgrav_\Hn(\ket{\BH})}}\[F\(\mathcal{M}\(\qop(\ket{\BH}\bra{\BH})\),\Qgrav\ket{\BH}\bra{\BH}\Qgrav^\dagger\)\]\leq 1-\frac{1}{8\pi e}\cdot\frac{\Hd^2}{e^{\frac{2}{\Hd^4}\[-h\(\left|\Qgrav\otimes\Qgrav\right|\)+\QCC\]}}
$$
where $h(|\Qgrav\otimes\Qgrav|)$ is the joint differential entropy of the elements of $\Qgrav\otimes\Qgrav$ after taking their norms, and the quantum communication complexity $\QCC=2\ell\Hn$ where $\ell$ is the number of queries $\A$ makes to $\Qgrav$.\end{thm}

\medskip

In the context of the black hole information problem or AdS/CFT, $\Hil$ should be thought of as the fundamental Hilbert space, $\Hil'$ the Hilbert space of the effective description, ${\cal M}$ maps the latter into the former, and $\qop$ a computationally-bounded reconstruction map.

To measure the fidelity between $\A$'s prediction $\qop(\ket{\BH}\bra{\BH})$ in $\Hil'$ and the fundamental time evolution $\Qgrav\ket{\BH}$ in $\Hil$, we have to map $\A$'s prediction to $\Hil$ via some map ${\cal M}$ and then measure correlation as the fidelity between its image and $\Qgrav\ket{\BH}\bra{\BH}\Qgrav^\dagger$, now both in $\Hil$. $\mathcal{M}$ can be \emph{any} CPTP map, with no bound on its complexity. It can be, for instance, any (combination) of the following CPTP maps:
\begin{itemize}
\item any isometry (a single linear operator ${\cal M}:\Hil'\to\Hil$ such that ${\cal M}^\dagger {\cal M}=I$).
\item if $\Hil'=\Hil\otimes\Hil''$ (or can be mapped to such via a CPTP map):
\begin{itemize}
\item the partial trace over $\Hil''$ (which in general increases rank).
\item measurement over $\Hil''$ (which for instance would include post-selection.
\end{itemize}
\end{itemize}
One may wonder why using an auxiliary map $\mathcal{M}$ does not artificially inflate the fidelity between $\A$'s prediction and the fundamental time evolution.
The important property here is  $\mathcal{M}$ is \emph{fixed}: independent of $\Qgrav$ and $\ket{\BH}$ (the sources of complexity (randomness)) and the model $\qop$ that $\A$ learns (the model of a bounded observer). This prevents $\mathcal{M}$ from artificially bridging the gap `in complexity' by depending on the particular instance of $\Qgrav$, $\ket{\BH}$, and $\qop$. What underlies the proof of Theorem~\ref{thm:fidelity-general-learners} is that if $\Qgrav$ has high differential entropy (i.e. less negative $h$), then any algorithm $\A$ making a bounded number of queries (i.e. bounded $\QCC$) to $\Qgrav$ can only learn $\qop$ that has limited (sub-maximal) correlation with $\Qgrav$. More precisely, $\qop$ has an operator-sum representation and the correlation between each operator and $\Qgrav$ is limited.
Post-processing outputs of $\qop$ (by applying a fixed CPTP map $\mathcal{M}$) will not increase the correlation (only decreasing it if a poor choice of channel is used as $\mathcal{M}$).

By allowing the quantum learning algorithm $\A$ to use any quantum channel $\qop$ instead of a single linear operator $\SCgrav$, we leave open the possibility that $\A$ outputs any quantum state $\rho$ as its prediction; $\rho$ may even be a \textbf{mixed} state in a \textbf{larger} Hilbert space, even though the fundamental time evolution $\Qgrav\ket{\BH}$ is a pure state in $\Hil$. A learning algorithm $\A$ may not necessarily perform better by using a channel $\qop$ to predict evolution by a unitary operator $\Qgrav$.
Our purpose in proving a fidelity bound for any algorithm $\A$ that predicts using \emph{any quantum channel} is to more closely model Hawking's semiclassical prediction for the time evolution of an evaporating black hole.

If Hawking's calculation is indeed a result of coarse-graining over complexity, as suggested by prior works~\cite{EngWal17b, EngPen21b, AkeEng22}, then in our model, the semiclassical prediction is the ``simple'' description produced by a general quantum learning algorithm (quantum channels) and quantum gravity dynamics are modeled as a highly random (or apparently random) unitary, then a fidelity bound for such algorithms would suggest that the inconsistency of Hawking's prediction can be traced back to a bound in computational complexity. 

Let us now discuss the assumption of high differential entropy in this and more general contexts.

\paragraph{Extending to bounds for random and pseudorandom $\Qgrav$.} We have focused on the most general result that we can obtain -- i.e. without subscribing to a particular choice of measure or ensemble for the operator $\Qgrav$. Our results are purely in terms of differential entropy, and we may expect that operators $U\otimes U$ of high differential entropy will yield bounds of similar order to Theorems~\ref{thm:fidelity-upperbd-random} and~\ref{thm:fidelity-upperbd-PRU-PRS}. While we have gone through a rigorous check of this, we expect that one may be able to compute an explicit bound on the average fidelity when $\Qgrav$ is a Haar random or pseudorandom unitary (an analog of Theorem~\ref{thm:fidelity-upperbd-random} for quantum channels $\qop$ instead of single operator models $\SCgrav$ as was assumed there), using some combination of the approximation theorem of~\cite{Jia10} (Proposition~\ref{thm:approx-Haar-U}) and our techniques in the proofs above  to replace the differential entropies of elements of $\Qgrav$ by known differential entropies.

\medskip We now proceed to the proof of Theorem~\ref{thm:fidelity-general-learners}.

\subsection{Proof of Theorem~\ref{thm:fidelity-general-learners}}

\medskip

We will use the following result to show our bound for algorithms that predict general quantum channels.
By Kraus' theorem, a quantum operation, mapping one quantum state to another, is a completely positive trace preserving (CPTP) map that has the following representation. This is also known as the operator-sum representation. For a reference see e.g.~\cite{NieChu11}.

\begin{prop}[Kraus' theorem]\label{thm:Kraus}
Let $\Hil_1$ be any Hilbert space of dimension $\Hd_1$, $\Hil_2$ be any Hilbert space of dimension $\Hd_2$, and $\qop$ be any quantum operation mapping density operators in $\Hil_1$ to density operators in $\Hil_2$. There exist $M\leq\Hd_1\Hd_2$ matrices $\{\Kraus_{k}\}_{k\in[M]}$ called the \emph{Kraus operators} such that for any $\rho\in\Hil_1$, $\qop$ maps $\rho$ to
$$\qop(\rho)=\sum_{k\in[M]}\Kraus_{k}\rho \Kraus_{k}^\dagger.$$
Conversely, any map of this form such that $\sum_{k}\Kraus_{k}^\dagger \Kraus_{k}=I$ is a quantum operation.
\end{prop}
\noindent The number of Kraus operators does not matter for us, so we will denote $\{\Kraus_k\}_{k\in[M]}$ by $\{\Kraus_k\}$.

If $M=1$ then $\qop$ takes any $\rho=\ket{\psi}\bra{\psi}$ corresponding to a pure state to $\qop(\rho)$ which is also a pure state.
If $M>1$ then $\qop$ can take $\rho=\ket{\psi}\bra{\psi}$ corresponding to a pure state to $\qop(\rho)$ which is a mixed state.

\begin{proof}[Proof of Theorem~\ref{thm:fidelity-general-learners}]

Let $\Hil$ be any Hilbert space and $\cal{U}$ be any distribution of unitaries in $\Hil$. Let $\A$ be any quantum algorithm, given oracle access to $\Qgrav\gets\mathcal{U}$.  During the ``query and learn'' phase, we assume $\A^{\Qgrav}_\Hn$ learns, as its model of $\Qgrav$, the description of a quantum operation (CPTP map) $\qop:\Hil\to\Hil'$ where $\Hil'$ can be any Hilbert space, not necessarily $\Hil$. Let $\{E_i\}_{i}$ be any set of Kraus operators for the CPTP map $\qop$, guaranteed to exist by Proposition~\ref{thm:Kraus}. 
Given $\ket{\BH}$ in the prediction phase, $\A$ applies the CPTP map $\qop$ and outputs
$$\rho=\qop(\ket{\BH}\bra{\BH})=\sum_{i\in[N]} E_i\ket{\BH}\bra{\BH}E_i^\dagger.$$
Consider any CPTP map $\mathcal{M}$ from $\Hil'$ to $\Hil$ that is fixed, i.e. independent of the instances of $\Qgrav$, $\ket{\BH}$, and $\qop$. Let $\{M_j\}_j$ be a fixed set of Kraus operators for $\mathcal{M}$. We now apply $\mathcal{M}$ to the prediction $\rho$ in order to define and assess its fidelity with the fundamental time evolution $\Qgrav\ket{\BH}$. Let
$$\rho':=\mathcal{M}(\rho)=\sum_{j}M_j\(\sum_{i\in[N]} E_i\ket{\BH}\bra{\BH}E_i^\dagger\)M_j^\dagger.$$
$(\mathcal{M}\circ\qop)$ is a CPTP map (a ``concatenated channel'') from $\Hil$ to $\Hil$, with Kraus operators $\{\Kraus_k\in\C^{\Hd\times\Hd}\}_k:=\{M_j E_i\}_{j,i}$. (The number of $\Kraus_k$ may be larger than as guaranteed in Proposition~\ref{thm:Kraus}. This does not matter for this proof.) The operators $\{\Kraus_k\}_k$ depend on and are defined by the random variables $\Qgrav\gets\mathcal{U}$ and $\qop\sim\A^\Qgrav_\Hn$, and the Kraus operators $\{E_i\}_i$ chosen for $\qop$.

\medskip
We now proceed to consider the fidelity between $\rho'$ and $\Qgrav\ket{\BH}\bra{\BH}\Qgrav^\dagger$. Consider any unitary $\Qgrav\in\C^{\Hd\times\Hd}$, CPTP map $\qop$ with Kraus operators $\{E_i\in\C^{\Hd'\times\Hd}\}$, and state $\ket{\BH}\in\C^{\Hd}$. As defined above, let $\{\Kraus_k\in\C^{\Hd\times\Hd}\}$ be Kraus operators for the CPTP map $(\mathcal{M}\circ\qop)$ and
$$\rho'=(\mathcal{M}\circ\qop)(\ket{\BH}\bra{\BH})=\sum\limits_k\Kraus_k\ket{\BH}\bra{\BH}\Kraus_k^\dagger.$$
Let $F(\rho',\Qgrav\ket{\BH}\bra{\BH}\Qgrav^\dagger)$ be the quantum state fidelity between $\rho'$ and $\Qgrav\ket{\BH}\bra{\BH}\Qgrav^\dagger$. The fidelity is symmetric, $$F(\rho',\Qgrav\ket{\BH}\bra{\BH}\Qgrav^\dagger)=F(\Qgrav\ket{\BH}\bra{\BH}\Qgrav^\dagger,\rho').$$
For any pure state $\tau=\ket{\phi}\bra{\phi}$ and density operator $\sigma$,
$$F(\tau,\sigma)=\left(\operatorname {Tr} {\sqrt {|\phi \rangle \langle \phi |\sigma |\phi \rangle \langle \phi |}}\right)^{2}=\langle \phi |\sigma |\phi \rangle \left(\operatorname {Tr} {\sqrt {|\phi \rangle \langle \phi |}}\right)^{2}=\langle \phi |\sigma |\phi \rangle.$$
Thus
\begin{align*}
F(\rho',\Qgrav\ket{\BH}\bra{\BH}\Qgrav^\dagger)&=\bra{\BH}\Qgrav^\dagger\rho'\Qgrav\ket{\BH}\\
&=\bra{\BH}\Qgrav^\dagger\(\sum_k \Kraus_k\ket{\BH}\bra{\BH}\Kraus_k^\dagger\)\Qgrav\ket{\BH}\\
&=\sum_k \bra{\BH}\Qgrav^\dagger\Kraus_k\ket{\BH}\bra{\BH}\Kraus_k^\dagger\Qgrav\ket{\BH}\\
&=\sum_k\left|\bra{\BH}\Kraus_k^\dagger\Qgrav\ket{\BH}\right|^2.
\end{align*}

Taking the average over Haar random $\ket{\BH}\gets\Haar$, the above equation becomes
\begin{equation}\label{eqn:fidelity-general-learner-avg-state-sum_k}
\Eover{\ket{\BH}\gets\Haar}\[F(\rho',\Qgrav\ket{\BH}\bra{\BH}\Qgrav^\dagger)\]=\sum_k \Eover{\ket{\BH}\gets\Haar}\left|\bra{\BH}\Kraus_k^\dagger\Qgrav\ket{\BH}\right|^2.
\end{equation}
Consider any $k$.
By Lemma~\ref{lem:avg-random-state} for unitary $\Qgrav$ and matrix $\Kraus_k$, we have the following average, over Haar random $\ket{\BH}$, of the squared inner product between $\Kraus_k\ket{\BH}$ and $\Qgrav\ket{\BH}$,
$$\Hd(\Hd+1)\(\Eover{\ket{\BH}\gets\Haar}\left|\bra{\BH}\Kraus_k^\dagger\Qgrav\ket{\BH}\right|^2\)=\left|\sum_{i,j\in[\Hd]}\Kraus^*_{kij}\Qgrav_{ij}\right|^2+\Tr(\Kraus_k^\dagger\Kraus_k).$$
Using the above equation for every $k$, and then $\sum\limits_k\Kraus_k^\dagger\Kraus_k=I$, Equation~\ref{eqn:fidelity-general-learner-avg-state-sum_k} becomes
\begin{align*}
\Eover{\ket{\BH}\gets\Haar}\[F(\rho',\Qgrav\ket{\BH}\bra{\BH}\Qgrav^\dagger)\]&=\frac{1}{\Hd(\Hd+1)}\sum_k\left|\sum_{i,j\in[\Hd]}\Kraus^*_{kij}\Qgrav_{ij}\right|^2+\sum_k\frac{\Tr(\Kraus_k^\dagger\Kraus_k)}{\Hd(\Hd+1)}\\
&=\frac{1}{\Hd(\Hd+1)}\sum_k\left|\sum_{i,j\in[\Hd]}\Kraus^*_{kij}\Qgrav_{ij}\right|^2+\frac{1}{(\Hd+1)}.
\end{align*}

We now take the average of this last equation over the distribution of $\Qgrav$ and $\{\Kraus_k\}$ that we have in our context ($\Qgrav\gets\mathcal{U}$ and $\{\Kraus_k\}$ is defined by the CPTP map $\qop$ that $\A^\Qgrav_\Hn$ learns, the Kraus operators for $\qop$, and the fixed Kraus operators $\{M_j\}$ for $\mathcal{M}$).
\begin{equation}\label{eqn:fidelity-general-learner-avg-F-upperbd-sumOk*U}
\Eover{\substack{\Qgrav\gets\mathcal{U}\\\qop\sim\A^\Qgrav_\Hn\\\ket{\BH}\gets\Haar}}\[F(\rho',\Qgrav\ket{\BH}\bra{\BH}\Qgrav^\dagger)\]=\frac{1}{\Hd(\Hd+1)}\sum_k\(\Eover{\Qgrav,\{\Kraus_k\}}\left|\sum_{i,j\in[\Hd]}\Kraus^*_{kij}\Qgrav_{ij}\right|^2\)+\frac{1}{(\Hd+1)}.
\end{equation}
To upper bound the right side, we focus on upper bounding, for each $k$, the inner product between $\Kraus_{k}$ and $\Qgrav$. Consider any $k$. Let the random variable $\Knorm_k:=\frac{1}{\Hd}\sum_{i,j\in[\Hd]} |\Kraus_{kij}|^2$.
\begin{align}
\Eover{\Qgrav,\{\Kraus_k\}}\left|\sum_{i,j\in[\Hd]}\Kraus^*_{kij}\Qgrav_{ij}\right|^2&=\Eover{\Qgrav,\{\Kraus_k\}}\[\(\sum_{i,j\in[\Hd]}\Kraus^*_{kij}\Qgrav_{ij}\)\(\sum_{\ell,\emm\in[\Hd]} \overline{\Kraus^*_{k\ell\emm}\Qgrav_{\ell\emm}}\)\]\nonumber\\
&=\Eover{\Qgrav,\{\Kraus_k\}}\[\sum_{i,j,\ell,\emm\in[\Hd]} \Kraus^*_{kij}\Qgrav_{ij}\Kraus_{k\ell\emm}\Qgrav^*_{\ell\emm}\]\nonumber\\
&\leq\Eover{\Qgrav,\{\Kraus_k\}}\[\sum_{i,j,\ell,\emm\in[\Hd]}\left|\Kraus_{kij}\Qgrav_{ij}\Kraus_{k\ell\emm}\Qgrav_{k\ell\emm}\right|\]\nonumber\\
&=\Knorm_k\sum_{i,j,\ell,\emm\in[\Hd]}\Eover{\Qgrav,\{\Kraus_k\}}\left|\frac{\Kraus_{kij}}{\sqrt{\Knorm_k}}\Qgrav_{ij}\frac{\Kraus_{k\ell\emm}}{\sqrt{\Knorm_k}}\Qgrav_{\ell\emm}\right|\label{eqn:fidelity-general-learner-upperbd-sum-Ok*U}
\end{align}

The scaling of the entries of $\Kraus_k$ by $\frac{1}{\sqrt{\Knorm_k}}$ in the last line above will allow us to ``balance out'' the norms of $\Kraus_k$ and $\Qgrav$ in the calculations to follow. This is needed to get a nontrivial bound on the average fidelity.

Next for any $i,j,\ell,\emm\in[\Hd]$, by continuous Fano's inequality (Proposition~\ref{thm:cts-Fano-ineq}) with $X=\left|\Qgrav_{ij}\Qgrav_{\ell\emm}\right|$, $Y=\{\Kraus_k\}_{k}$, and $\widehat{X}(Y)=\frac{\left|\Kraus_{kij}\Kraus_{k\ell\emm}\right|}{\Knorm_k}$,

$$\Eover{\Qgrav,\{\Kraus_k\}}\left|\frac{\Kraus_{kij}}{\sqrt{\Knorm_k}}\Qgrav_{ij}\frac{\Kraus_{k\ell\emm}}{\sqrt{\Knorm_k}}\Qgrav_{\ell\emm}\right|\leq\Eover{\Qgrav,\{\Kraus_k\}}\[\frac{\left|\Qgrav_{ij}\Qgrav_{\ell\emm}\right|^2}{2}\]+\Eover{\Qgrav,\{\Kraus_k\}}\[\frac{\left|\Kraus_{kij}\Kraus_{k\ell\emm}\right|^2}{{2}\Knorm_k^2}\]-\frac{e^{2h\(\left|\Qgrav_{ij}\Qgrav_{\ell\emm}\right|\middle|\{\Kraus_k\}\)}}{4\pi e}.$$
Using the above equation for every $i,j,\ell,\emm\in[\Hd]$, Equation~\ref{eqn:fidelity-general-learner-upperbd-sum-Ok*U} becomes

\begin{align*}
\Eover{\Qgrav,\{\Kraus_k\}}&\left|\sum_{i,j\in[\Hd]}\Kraus^*_{kij}\Qgrav_{ij}\right|^2\\
&\quad\quad\leq\Knorm_k\sum_{i,j,\ell,\emm\in[\Hd]}\(\Eover{\Qgrav,\{\Kraus_k\}}\frac{\left|\Qgrav_{ij}\Qgrav_{\ell\emm}\right|^2}{2}+\Eover{\Qgrav,\{\Kraus_k\}}\frac{\left|\Kraus_{kij}\Kraus_{k\ell\emm}\right|^2}{2\Knorm_k^2}-\frac{1}{4\pi e}e^{2h\(\left|\Qgrav_{ij}\Qgrav_{\ell\emm}\right|\middle|\{\Kraus_k\}\)}\)\\
&\quad\quad=\Knorm_k\(\frac{\Hd^2}{2}+\frac{\Hd^2}{2}-\frac{1}{4\pi e}\sum_{i,j,\ell,\emm\in[\Hd]}e^{2h\(\left|\Qgrav_{ij}\Qgrav_{\ell\emm}\right|\middle|\{\Kraus_k\}\)}\)\\
&\quad\quad=\Knorm_k\Hd^2-\frac{\Knorm_k}{4\pi e}\sum_{i,j,\ell,\emm\in[\Hd]}e^{2h\(\left|\Qgrav_{ij}\Qgrav_{\ell\emm}\right|\middle|\{\Kraus_k\}\)}.
\end{align*}

Returning to Equation~\ref{eqn:fidelity-general-learner-avg-F-upperbd-sumOk*U}, we now use the above equation for every $k$ to obtain the following. We then use $\sum_k\Knorm_k=1$ (which follows from the definition of the $\Knorm_k$ and $\sum_k\Kraus_k\Kraus_k^\dagger=I\in\C^{\Hd\times\Hd}$).
\begin{align}
\Eover{\substack{\Qgrav\gets\mathcal{U}\\\qop\sim\A^\Qgrav_\Hn\\\ket{\BH}\gets\Haar}}&\[F(\rho',\Qgrav\ket{\BH}\bra{\BH}\Qgrav^\dagger)\]\nonumber\\
&\quad\quad\leq\frac{1}{\Hd(\Hd+1)}\(\sum_k\Knorm_k\)\(\Hd^2-\frac{1}{4\pi e}\sum_{i,j,\ell,\emm\in[\Hd]}e^{2h\(\left|\Qgrav_{ij}\Qgrav_{\ell\emm}\right|\middle|\{\Kraus_k\}\)}\)+\frac{1}{(\Hd+1)} 
\nonumber\\
&\quad\quad=1-\frac{1}{\Hd(\Hd+1)}\(\frac{1}{4\pi e}\sum_{i,j,\ell,\emm\in[\Hd]}e^{2h\(\left|\Qgrav_{ij}\Qgrav_{\ell\emm}\right|\middle|\{\Kraus_k\}\)}\). 
\label{eqn:fidelity-general-learner-EF-leq-1-sum-e^2h}
\end{align}
It suffices to lower bound the sum of (exponentiated) differential entropies in order to upper bound the average fidelity.

By Jensen's inequality applied to the convex exponential function,
\begin{equation}\label{eqn:fidelity-general-learner-sum-e^2h-Jensens}
\sum_{i,j,\ell,\emm\in[\Hd]}e^{2h\(\left|\Qgrav_{ij}\Qgrav_{\ell\emm}\right|\middle|\{\Kraus_k\}\)} \geq\Hd^4\cdot e^{\frac{2}{\Hd^4}\[\sum\limits_{i,j,\ell,\emm\in[\Hd]} h\(\left|\Qgrav_{ij}\Qgrav_{\ell\emm}\right|\middle|\{\Kraus_k\}\)\]}. 
\end{equation}
By the conditional chain rule in Lemma~\ref{lem:h-cond-chain-rule}, and then by Lemma~\ref{lem:I-expr-h},
\begin{align}
\sum\limits_{i,j,\ell,\emm\in[\Hd]} h\(\left|\Qgrav_{ij}\Qgrav_{\ell\emm}\right|\middle|\{\Kraus_k\}\)&\geq h\(\{\left|\Qgrav_{ij}\Qgrav_{\ell\emm}\right|\}_{i,j,\ell,\emm\in[\Hd]}\middle|\{\Kraus_k\}\)\nonumber\\
&=h\(\{\left|\Qgrav_{ij}\Qgrav_{\ell\emm}\right|\}_{i,j,\ell,\emm\in[\Hd]}\)-I\(\{\left|\Qgrav_{ij}\Qgrav_{\ell\emm}\right|\}_{i,j,\ell,\emm\in[\Hd]};\{\Kraus_k\}\).\label{eqn:fidelity-general-learner-sum-h-lowerbd}
\end{align}

We define a protocol $\prot$ between two parties, party $A$ and party $B$ who are not bounded in runtime. Sample $\Qgrav\gets\mathcal{U}$ and give $\(\Qgrav,\{|\Qgrav_{ij}\Qgrav_{\ell\emm}|\}_{i,j,\ell,\emm\in[\Hd]}\)$ to party $A$ (and nothing to party $B$). We assume $A$ is given or can compute a sequence of quantum gates that approximates applying $U\in\C^{\Hd\times\Hd}$ (say within exponential errors). Party $B$ runs the quantum learning algorithm $\A$ and whenever it queries for a state (a set of qubits in its registers) $\ket{\phi_i}$, party $B$ sends those qubits to party $A$ who applies $\Qgrav$ via its sequence of gates and sends the qubits back, in state $\Qgrav\ket{\phi_i}$. (As mentioned, one can consider the learning algorithm $\A$ to query the oracle on qubits in its registers that may be either pure or mixed states, represented by density operators $\sigma_i$. The protocol here does not change -- party A still applies the gates for $\Qgrav$ and returns the qubits, now having density operator $\Qgrav\sigma_i\Qgrav^\dagger$.) This continues until $\A$ finishes its querying and computing, and has learned its CPTP map $\qop$ as its model of $\Qgrav$. ($\A$ may learn a description of $\qop$, e.g. its implementation via a  quantum circuit, defining the CPTP map $\qop$.) Party $B$ uses the description of $\qop$ to compute a set of Kraus operators $\{E_i\}_i$ for $\qop$. We assume there is a procedure to do this; recall party $B$ is not limited in runtime. Finally party $B$ uses the fixed $M_j$ to compute and output $\{\Kraus_k\}_k:=\{M_jE_i\}_{j,i}$.
By Corollary~\ref{cor:I-QCC}, 
\begin{equation}\label{eqn:fidelity-general-learner-I-QCC}
I\(\{\Kraus_k\}_k;\(\Qgrav,\{|\Qgrav_{ij}\Qgrav_{\ell\emm}|\}_{i,j,\ell,\emm\in[\Hd]}\)\)\leq\QCC(\prot).
\end{equation}

By the chain rule (Lemma~\ref{lem:I-chain-rule}) and the nonnegativity (Lemma~\ref{lem:I-nonneg}) properties of mutual information,
\begin{align*}
&I\(\(\Qgrav,\{|\Qgrav_{ij}\Qgrav_{\ell\emm}|\}_{i,j,\ell,\emm\in[\Hd]}\);\{\Kraus_k\}_k\)\\
&\quad\quad\quad\quad=I\(\{|\Qgrav_{ij}\Qgrav_{\ell\emm}|\}_{i,j,\ell,\emm\in[\Hd]};\{\Kraus_k\}_k\)+I\(\Qgrav;\{\Kraus_k\}_k|\{|\Qgrav_{ij}\Qgrav_{\ell\emm}|\}_{i,j,\ell,\emm\in[\Hd]}\)\\
&\quad\quad\quad\quad\geq I\(\{|\Qgrav_{ij}\Qgrav_{\ell\emm}|\}_{i,j,\ell,\emm\in[\Hd]};\{\Kraus_k\}_k\).
\end{align*}

Combining this with the $\QCC$ bound above (Equation~\ref{eqn:fidelity-general-learner-I-QCC}),
\begin{equation*}
\QCC(\prot)\geq I\(\{|\Qgrav_{ij}\Qgrav_{\ell\emm}|\}_{i,j,\ell,\emm\in[\Hd]};\{\Kraus_k\}_k\).
\end{equation*}
Using this and Equation~\ref{eqn:fidelity-general-learner-sum-h-lowerbd}, Equation~\ref{eqn:fidelity-general-learner-sum-e^2h-Jensens} becomes
\begin{equation}
\sum_{i,j,\ell,\emm\in[\Hd]}e^{2h\(\left|\Qgrav_{ij}\Qgrav_{\ell\emm}\right|\middle|\{\Kraus_k\}\)}\geq\Hd^4\cdot e^{\frac{2}{\Hd^4}\[h\(\{\left|\Qgrav_{ij}\Qgrav_{\ell\emm}\right|\}_{i,j,\ell,\emm\in[\Hd]}\)-\QCC(\prot).\]}.
\end{equation}
Using this in Equation~\ref{eqn:fidelity-general-learner-EF-leq-1-sum-e^2h},
\begin{align*}
\Eover{\substack{\Qgrav\gets\mathcal{U}\\\qop\sim\A^\Qgrav_\Hn\\\ket{\BH}\gets\Haar}}&\[F(\rho',\Qgrav\ket{\BH}\bra{\BH}\Qgrav^\dagger)\]\leq 1-\frac{1}{4\pi e}\(\frac{\Hd^3}{(\Hd+1)}\)\cdot e^{\frac{2}{\Hd^4}\[h\(\{\left|\Qgrav_{ij}\Qgrav_{\ell\emm}\right|\}_{i,j,\ell,\emm\in[\Hd]}\)-\QCC(\prot)\]}.
\end{align*}
We note that $\{\left|\Qgrav_{ij}\Qgrav_{\ell\emm}\right|\}_{i,j,\ell,\emm\in[\Hd]}$ are the matrix elements of $\Qgrav\otimes\Qgrav$ after taking their norms, and $\QCC(\prot)=2\ell\Hn$ where $\ell$ is the number of queries that $\A$ makes to $\Qgrav$.
Recalling that
$\rho'=\mathcal{M}\(\qop(\ket{\BH}\bra{\BH})\)$, this proves Theorem~\ref{thm:fidelity-general-learners}.
\end{proof}

\section*{Acknowledgments}
It is a pleasure to thank S. Aaronson, C. Akers, D. Harlow, A. Harrow, and E. Verheijden for helpful discussions. LY is supported by an NSF Graduate Research Fellowship and NCNS-215414. NE is supported in part by NSF grant no. PHY-2011905, by the U.S. Department of Energy Early Career Award DE-SC0021886, by the John Templeton Foundation and the Gordon and Betty Moore Foundation via the Black Hole Initiative, by the Sloan Foundation, and by funds from the MIT department of physics.

\appendix
\section{Technical Parts of Proof of Theorem~\ref{thm:fidelity-upperbd-random}}\label{sec:appendix-tech}

In this section we provide any proofs or details thereof that were too long or technical to include in the main body of the paper.

We first give an explicit definition of the expected value (which we use in all of our proofs), since this expected value notation is not commonly used in physics and since we work with continuous random variables. 
\begin{defn}\label{defn:expected-value}(Expected Value)
If $X$ is a real-valued random variable defined on a probability space $(\Omega,\Sigma,P)$, then the expected value of $X$, denoted by $\E[X]$, is defined as the Lebesgue integral
$$\E[X]=\int_{\Omega} X\,dP.$$
When $X$ depends on other random variables, e.g. $Y$, we write those random variables below $\E$, e.g. $\Eover{Y}[X]$.
\end{defn}

\noindent The remainder of this section contains parts of the proof of Theorem~\ref{thm:fidelity-upperbd-random}. 

\begin{proof}[Proof of Lemma~\ref{lem:fidelity-random-pf-F-sqrt-F-ineq}]
Consider any $\Qgrav$ and $\SCgrav$. Note that 
\begin{align*}
&F_{\SCgrav,\Qgrav}\leq\frac{1}{\Hd}\sqrt{\Hd(\Hd+1)F_{\SCgrav,\Qgrav}-\Hd\SCcol}+\frac{\SCcol}{\sqrt{\Hd}}\\
\iff&\Hd F_{\SCgrav,\Qgrav}-\SCcol\sqrt{\Hd}\leq\sqrt{\Hd(\Hd+1)F_{\SCgrav,\Qgrav}-\Hd\SCcol}\\
\iff&\Hd^2 F_{\SCgrav,\Qgrav}^2-2\SCcol\sqrt{\Hd}\Hd F_{\SCgrav,\Qgrav}+\SCcol^2\Hd\leq \Hd^2 F_{\SCgrav,\Qgrav}+\Hd F_{\SCgrav,\Qgrav}-\Hd\SCcol.
\end{align*}
Since $0\leq F_{\SCgrav,\Qgrav}\leq 1$ (because $\alpha\leq 1$), we have $0\leq \Hd F_{\SCgrav,\Qgrav}$ and $\Hd^2 F_{\SCgrav,\Qgrav}^2\leq \Hd^2 F_{\SCgrav,\Qgrav}$. So the last inequality above is true if $-2\SCcol\sqrt{\Hd}\Hd F_{\SCgrav,\Qgrav}+\SCcol^2\Hd\leq -\Hd\SCcol \iff 1+\SCcol\leq 2\sqrt{\Hd}F_{\SCgrav,\Qgrav}$.

Since either $F_{\SCgrav,\Qgrav}\leq\frac{1+\SCcol}{2\sqrt{\Hd}}$ or $F_{\SCgrav,\Qgrav}\geq\frac{1+\SCcol}{2\sqrt{\Hd}}$, by the above, either
$F_{\SCgrav,\Qgrav}\leq\frac{1}{2\sqrt{\Hd}}+\frac{\SCcol}{2\sqrt{\Hd}}$ or $F_{\SCgrav,\Qgrav}\leq\frac{1}{\Hd}\sqrt{\Hd(\Hd+1)F_{\SCgrav,\Qgrav}-\Hd\SCcol}+\frac{\SCcol}{\sqrt{\Hd}}$.
Thus the desired inequality holds:
$$F_{\SCgrav,\Qgrav}\leq \frac{1}{\Hd}\sqrt{\Hd(\Hd+1)F_{\SCgrav,\Qgrav}-\Hd\SCcol}+\frac{3\SCcol}{2\sqrt{\Hd}}+\frac{1}{2\sqrt{\Hd}}.$$
\end{proof}

\begin{proof}[Proof of Lemma~\ref{lem:fidelity-rand-upperbd-O-U}]

Consider any such distribution of $(\Vmat,\CN,\Qmat)$. We will consider expected values (averages) and probabilities over this joint distribution. 
Let $\Eapp$ be the event that ${\displaystyle \max_{\substack{i\in[\Hd]\\ j\in[\md(\Hd)]}}\left|\Vmat_{ij}-\CN_{ij}\right|\leq\epsilon(\Hd)}$ so we have ${\displaystyle \Pr_{\Vmat,\CN,\Qmat}[\neg\Eapp]\leq\delta(\Hd)}$. ($\Eapp$ only depends on $\Vmat,\CN$; the inclusion of $\Qmat$ is not necessary but will be useful.) The notation $|\Eapp$ denotes conditioning on the event $\Eapp$.
\begin{align*}
\sum_{\substack{i\in[\Hd]\\ j\in[\md(\Hd)]}} \(\Eover{\Vmat,\CN,\Qmat}\left|\Qmat^*_{ij}\Vmat_{ij}\right|\)
&\leq \Pr_{\Vmat,\CN,\Qmat}\[\Eapp\]\Eover{\Vmat,\CN,\Qmat|\Eapp}\[\sum_{\substack{i\in[\Hd]\\ j\in[\md(\Hd)]}} |\Qmat^*_{ij}\Vmat_{ij}|\]\\
&\quad\quad + \Pr_{\Vmat,\CN,\Qmat}\[\neg\Eapp\]\(\max_{\Vmat,\Qmat} \sum_{\substack{i\in[\Hd]\\ j\in[\md(\Hd)]}} |\Qmat^*_{ij}\Vmat_{ij}|\)\\
&\leq \Eover{\Vmat,\CN,\Qmat|\Eapp}\[\sum_{\substack{i\in[\Hd]\\ j\in[\md(\Hd)]}} |\Qmat^*_{ij}\Vmat_{ij}|\]+\delta(\Hd)\cdot\md(\Hd)\cdot\sqrt{\SCcol}
\end{align*}
where in the first line we maximize over the distribution, and the second line uses the Cauchy-Schwarz inequality to bound by the column norms of $\Qmat$ and $\Vmat$.
Furthermore, since whenever $\Eapp$ occurs, $|\Vmat_{ij}|\leq \epsilon(\Hd)+|\CN_{ij}|$, we have
\begin{align}
\sum_{\substack{i\in[\Hd]\\ j\in[\md(\Hd)]}}& \(\Eover{\Vmat,\CN,\Qmat}\left|\Qmat^*_{ij}\Vmat_{ij}\right|\)\leq\delta(\Hd)\cdot\md(\Hd)\cdot\sqrt{\SCcol}+\Eover{\Vmat,\CN,\Qmat|\Eapp}\(\sum_{\substack{i\in[\Hd]\\ j\in[\md(\Hd)]}} |\Qmat^*_{ij}|(\epsilon(\Hd)+|\CN_{ij}|)\) \nonumber\\
&=\delta(\Hd)\cdot\md(\Hd)\cdot\sqrt{\SCcol}+\Eover{\Vmat,\CN,\Qmat|\Eapp}\[\sum_{\substack{i\in[\Hd]\\ j\in[\md(\Hd)]}} |\Qmat^*_{ij}|\]\cdot\epsilon(\Hd)+\Eover{\Vmat,\CN,\Qmat|\Eapp}\[\sum_{\substack{i\in[\Hd]\\ j\in[\md(\Hd)]}} |\Qmat^*_{ij}||\CN_{ij}|\] \nonumber\\
&\leq\delta(\Hd)\cdot\md(\Hd)\cdot\sqrt{\SCcol}+\md(\Hd)\cdot\sqrt{\SCcol\Hd}\cdot\epsilon(\Hd)+\frac{1}{1-\delta(\Hd)}\cdot\Eover{\Vmat,\CN,\Qmat}\[\sum_{\substack{i\in[\Hd]\\ j\in[\md(\Hd)]}} |\Qmat^*_{ij}||\CN_{ij}|\]. \label{eqn:fidelity-random-pf-SC*Q-below-m(d)-CN}
\end{align}

Now we use continuous Fano's inequality (Proposition~\ref{thm:cts-Fano-ineq}) to bound the summand here. We also use that for a complex Gaussian variable $Z=\frac{X+iY}{\sqrt{2d}}$ sampled from $\CNdistr(0,\frac{1}{\Hd})$, we have $\E\(|Z|^2\)=\frac{\E(X^2+Y^2)}{2\Hd}=\frac{1}{\Hd}$.
\begin{align*}
\sum_{\substack{i\in[\Hd]\\ j\in[\md(\Hd)]}}\( \Eover{\Vmat,\CN,\Qmat}\[|\Qmat^*_{ij}||\CN_{ij}|\]\)&\leq -\frac{1}{4\pi e}\sum_{\substack{i\in[\Hd]\\ j\in[\md(\Hd)]}}e^{2h\(|\CN_{ij}||\Qmat\)}+\sum_{\substack{i\in[\Hd]\\ j\in[\md(\Hd)]}}\frac{\Eover{\Vmat,\CN,\Qmat}\(|\CN_{ij}|^2\)}{2}+\sum_{\substack{i\in[\Hd]\\ j\in[\md(\Hd)]}}\frac{\Eover{\Vmat,\CN,\Qmat}\(|\Qmat_{ij}|^2\)}{2}\\
&\leq -\frac{1}{4\pi e}\sum_{\substack{i\in[\Hd]\\ j\in[\md(\Hd)]}}e^{2h\(|\CN_{ij}||\Qmat\)}+\frac{\md(\Hd)}{2}+\frac{\md(\Hd)\cdot\SCcol}{2}
\end{align*}
Using this bound on the summand and $\frac{1}{1-\delta(\Hd)}\geq 1$, Equation~\ref{eqn:fidelity-random-pf-SC*Q-below-m(d)-CN} becomes
\begin{equation}\label{eqn:fidelity-random-pf-SC*Q-below-m(d)-e^h}
\sum_{\substack{i\in[\Hd]\\ j\in[\md(\Hd)]}} \(\Eover{\Vmat,\CN,\Qmat}\left|\Qmat^*_{ij}\Vmat_{ij}\right|\)\leq\md(\Hd)\(\epsilon(\Hd)\sqrt{\SCcol\Hd}+\delta(\Hd)\sqrt{\SCcol}+\frac{1+\SCcol}{2(1-\delta(\Hd))}\)-\frac{1}{4\pi e}\sum_{\substack{i\in[\Hd]\\ j\in[\md(\Hd)]}}e^{2h\(|\CN_{ij}||\Qmat\)}.
\end{equation}

Now we seek to bound this summand in terms of the mutual information between $\CN$ and $\Qmat$.
By Jensen's inequality applied to the convex exponential function,
\begin{equation}\label{eqn:fidelity-random-pf-lowerbd-sum-e^h-jensen}
\frac{1}{\Hd\cdot\md(\Hd)} \sum_{\substack{i\in[\Hd]\\ j\in[\md(\Hd)]}}e^{2h\(|\CN_{ij}||\Qmat\)}\geq e^{\frac{2}{\Hd\cdot\md(\Hd)}\[\sum_{i\in[\Hd],j\in[\md(\Hd)]} h\(|\CN_{ij}||\Qmat\)\]}.
\end{equation}
For any independent random variables $X_1,\ldots,X_{\ell'}$, for every $\ell\in[\ell']$, we have $h(X_\ell|\Qmat)=h(X_{\ell}|\Qmat,X_1,\ldots,X_{\ell-1})$. Since the $\CN_{ij}$ are independent, by this equation and the conditional chain rule (Lemma~\ref{lem:h-cond-chain-rule}),
\begin{equation*} 
\sum_{\substack{i\in[\Hd]\\ j\in[\md(\Hd)]}} h\(|\CN_{ij}||\Qmat\)=h\(\{|\CN_{ij}|\}_{\substack{i\in[\Hd]\\ j\in[\md(\Hd)]}}\,\middle|\,\Qmat\).
\end{equation*}
For notational convenience, we denote the norms of the elements of $\CN$ by $|\CN|:=\{|\CN_{ij}|\}_{\substack{i\in[\Hd]\\ j\in[\md(\Hd)]}}$.
By the relationship between conditional entropy and mutual information (Lemma~\ref{lem:I-expr-h}),
\begin{equation*} 
h\(|\CN|\,\middle|\,\Qmat\)=h\(|\CN|\)-I\(|\CN|;\Qmat\).
\end{equation*}
By the two equations above, Equation~\ref{eqn:fidelity-random-pf-lowerbd-sum-e^h-jensen} becomes
$$\sum_{\substack{i\in[\Hd]\\ j\in[\md(\Hd)]}}e^{2h\(|\CN_{ij}||\Qmat\)}\geq \Hd\cdot\md(\Hd)\cdot e^{\frac{2}{\Hd\cdot\md(\Hd)}\[h\(|\CN|\)-I\(|\CN|;\Qmat\)\]}.$$

For a complex Gaussian variable $Z=\frac{X+iY}{\sqrt{2d}}$ sampled from $\CNdistr\(0,\frac{1}{\Hd}\)$, we have $h\(|Z|\)=h\(\frac{|X+iY|}{\sqrt{2\Hd}}\)=h\(|X+iY|\)+\log\left|\frac{1}{\sqrt{2d}}\right|$ by Lemma~\ref{lem:h-scale-constant}. Note that $|X+iY|=\sqrt{X^2+Y^2}$ is a Rayleigh (or more generally Chi) random variable so, by calculation using its PDF, we have its differential entropy is $1+\log \frac{1}{\sqrt{2}}+\frac{\gamma_E}{2}$ where $\gamma_E$ is Euler's constant. Therefore
$$h\(|Z|\)=1-\log 2+\frac{\gamma_E}{2}-\frac{\log\Hd}{2}.$$
Since the $\CN_{ij}$ are independent complex Gaussians with distribution $\CNdistr\(0,\frac{1}{\Hd}\)$,
\begin{align*}
\sum_{\substack{i\in[\Hd]\\ j\in[\md(\Hd)]}}e^{2h\(|\CN_{ij}||\Qmat\)}&\geq \Hd\cdot\md(\Hd)\cdot e^{\frac{2}{\Hd\cdot\md(\Hd)}\[\Hd\cdot\md(\Hd)\cdot\(1-\log 2+\frac{\gamma_E}{2}-\frac{\log\Hd}{2}\)-I\(|\CN|;\Qmat\)\]} \nonumber\\
&=\md(\Hd)\cdot e^{2-2\log 2+\gamma_E}\cdot e^{\frac{-2\cdot I\(|\CN|;\Qmat\)}{\Hd\cdot\md(\Hd)}}. 
\end{align*}
Using this lower bound on the summand in Equation~\ref{eqn:fidelity-random-pf-SC*Q-below-m(d)-e^h}, we get Lemma~\ref{lem:fidelity-rand-upperbd-O-U}:
\begin{equation*} 
\sum_{\substack{i\in[\Hd]\\ j\in[\md(\Hd)]}} \(\Eover{\Vmat,\CN,\Qmat}\left|\Qmat^*_{ij}\Vmat_{ij}\right|\)\leq\md(\Hd)\(\epsilon(\Hd)\sqrt{\SCcol\Hd}+\delta(\Hd)\sqrt{\SCcol}+\frac{1+\SCcol}{2(1-\delta(\Hd))}-\frac{e^{\gamma_E+1}}{16\pi}\cdot e^{\frac{-2\cdot I\(|\CN|;\Qmat\)}{\Hd\cdot\md(\Hd)}}\).
\end{equation*}
\end{proof}

\paragraph{Instantiating the parameters in the proof of Theorem~\ref{thm:fidelity-upperbd-random}.} Here we explain the technical details omitted from the proof of Theorem~\ref{thm:fidelity-upperbd-random} for clarity. 
After Equation~\ref{eqn:fidelity-random-pf-before-params}, we instantiate our parameters to derive our final bound. We will use the parameters $\md(\Hd),\epsilon(\Hd),\delta(\Hd)$ given by Proposition~\ref{thm:approx-Haar-U}.
Since $\delta(\Hd)\leq \frac{2}{3}$ in Proposition~\ref{thm:approx-Haar-U}, we have $\frac{1}{1-\delta(\Hd)} \leq 1+3\delta(\Hd)$, which yields
\begin{align*}
\Eover{\Qgrav,\SCgrav}\[F_{\SCgrav,\Qgrav}\]&\leq\frac{\qd\cdot\md(\Hd)}{\Hd}\(\epsilon(\Hd)\sqrt{\SCcol\Hd}+\delta(\Hd)\sqrt{\SCcol}+3\delta(\Hd)\(\frac{1+\SCcol}{2}\)-\frac{e^{\gamma_E+1}}{16\pi}\cdot e^{\frac{-2\cdot \QCC(\pi)}{\Hd\cdot\md(\Hd)}}\)\\
&\quad\quad+\(\frac{1+\SCcol}{2}\)+\frac{3\SCcol}{2\sqrt{\Hd}}+\frac{1}{2\sqrt{\Hd}}.
\end{align*}
Recall that we partitioned the columns of $\Qgrav$ into $\qd:=\lfloor\frac{\Hd}{\md(\Hd)}\rfloor$ sets. Using $\frac{\Hd}{\md(\Hd)}-1\leq \qd \leq\frac{\Hd}{\md(\Hd)}$, we have
\begin{align*}
\Eover{\Qgrav,\SCgrav}\[F_{\SCgrav,\Qgrav}\]&\leq\(\epsilon(\Hd)\sqrt{\SCcol\Hd}+\delta(\Hd)\sqrt{\SCcol}+3\delta(\Hd)\(\frac{1+\SCcol}{2}\)\)+\(\frac{\md(\Hd)}{\Hd}-{1}\)\cdot\frac{e^{\gamma_E+1}}{16\pi}\cdot e^{\frac{-2\cdot \QCC(\pi)}{\Hd\cdot\md(\Hd)}}\\
&\quad\quad+\(\frac{1+\SCcol}{2}\)+\frac{3\SCcol}{2\sqrt{\Hd}}+\frac{1}{2\sqrt{\Hd}}.
\end{align*}
Let $t$ denote a parameter in $(0,6]$ which we will set shortly.
By Proposition~\ref{thm:approx-Haar-U} we can set $\md(\Hd)=\frac{t^2\Hd}{72\log\Hd}$, $\epsilon(\Hd)=\frac{3t}{\sqrt{\Hd}}$, and $\delta(\Hd)=O\(\frac{1}{\Hd}\)$. For the right side above to be sub-maximal, we need $\epsilon(\Hd)\sqrt{\SCcol\Hd}=3t\sqrt{\SCcol}$ to be a small constant. Hence we will prove our bound for $\SCcol=O(1)$. With these parameters, our bound becomes
\begin{align*}
\Eover{\Qgrav,\SCgrav}\[F_{\SCgrav,\Qgrav}\]&\leq 3t\sqrt{\SCcol}+O\(\frac{1}{\Hd}\)+\(\frac{t^2}{72\log\Hd}-{1}\)\cdot\frac{e^{\gamma_E+1}}{16\pi}\cdot e^{\frac{-2\cdot \QCC(\pi)}{\Hd\cdot\md(\Hd)}}\\
&\quad\quad+\(\frac{1+\SCcol}{2}\)+O\(\frac{1}{\sqrt{\Hd}}\).
\end{align*}
 
As in the statement of Theorem~\ref{thm:fidelity-upperbd-random} let $c$ denote any positive constant.\footnote{Contrary to popular practice, we will not set this $c$ equal to $1$. \\(You can if you wish, though the bound is best for small $c$, e.g. $c\leq \frac{1}{10^3}$.)} For algorithms with bounded query (communication) complexity with the oracle, specifically for $\QCC(\prot)\leq c\(\frac{\Hd\cdot\md(\Hd)}{2}\)$, we have
\begin{align*}
\Eover{\Qgrav,\SCgrav}\[F_{\SCgrav,\Qgrav}\]&\leq 3t\sqrt{\SCcol}+\(\frac{t^2}{72\log\Hd}-{1}\)\cdot\frac{e^{\gamma_E+1}}{16\pi}\cdot e^{-c}+\(\frac{1+\SCcol}{2}\)+O\(\frac{1}{\sqrt{\Hd}}\)\\
&= 3t\sqrt{\SCcol}+O\(\frac{1}{\log\Hd}\)-\frac{e^{\gamma_E+1}}{16\pi}\cdot e^{-c}+\(\frac{1+\SCcol}{2}\).
\end{align*}

It only remains to give a bound on $3t\sqrt{\SCcol}$ where, per the statement of Theorem~\ref{thm:fidelity-upperbd-random}, $\sqrt{\SCcol}$ is the maximal value for the column norms of $\SCgrav$. For any positive constant $\beta < 1$, let $t=\frac{\beta}{3}\cdot\(\frac{e^{\gamma_E+1}}{16\pi}\cdot e^{-c}\)\cdot\min\{1,\frac{1}{\sqrt{\SCcol}}\}$.

Then for any $\SCcol=O(1)$ (regardless of whether $\SCcol > 1$ or $\SCcol \leq 1$), we have $3t\sqrt{\SCcol}\leq\beta\cdot\(\frac{e^{\gamma_E+1}}{16\pi}\cdot e^{-c}\)$. Thus
\begin{align*}
\Eover{\Qgrav,\SCgrav}\[F_{\SCgrav,\Qgrav}\]\leq \frac{1+\SCcol}{2}-\(\frac{e^{\gamma_E+1}}{16\pi}\cdot e^{-c}\,(1-\beta)\)+O\(\frac{1}{\log\Hd}\).
\end{align*}
From this we obtain the bound in Theorem~\ref{thm:fidelity-upperbd-random}.

\medskip
Lastly we state the original theorem in~\cite{Jia10} for approximating a Haar random unitary by complex Gaussians.  
\begin{prop}[Theorem A.2 in~\cite{Jia10}]\label{thm:approx-Haar-U-Jiang}
For any $\Hd\geq 2$, let $Z=\(Z_{ij}\)_{i,j\in[\Hd]}$ be a $\Hd\times\Hd$ matrix where $Z_{ij}\gets\CNdistr(0,1)$ are independent complex Gaussian variables\footnote{Equivalently, let $Z_{ij}=\frac{X_{ij}+iY_{ij}}{\sqrt{2}}$ where $X_{ij},Y_{ij}\gets\Ndistr(0,1)$ are independent standard Gaussian variables.}. Let $U=\(U_{ij}\)_{i,j\in[\Hd]}$ be the $\Hd\times\Hd$ matrix that results from performing the Gram-Schmidt procedure\footnote{The Gram–Schmidt procedure is a simple process that takes a finite, linearly independent set of vectors and generates an orthogonal set that spans the same subspace.} on the $\Hd$ columns vectors of $Z$ and normalizing them. $U$ is distributed according to the Haar measure on the unitary group.
For any $r\in(0,1/4)$, $s>0$, $t>0$, and $\md\leq\frac{r\Hd}{2}$, 
\begin{align*}
\Pr_{Z}&\[\max_{i\in[\Hd],j\in[\md]}\left|\sqrt{\Hd}\cdot\Qgrav_{ij}-Z_{ij}\right|\geq rs+2t\]\\
&\quad\leq 4\md e^{-\Hd r^2/8}+\md\Hd e^{-s^2}+\frac{6\md\Hd}{t}\(1+\frac{t^2}{12(\md+t\sqrt{\Hd})}\)^{-\Hd}
\end{align*}
where ${\displaystyle \Pr_{Z}}$ is the probability averaged over $Z$.
\end{prop}

Theorem~\ref{thm:fidelity-upperbd-random} uses Proposition~\ref{thm:approx-Haar-U} which we show follows from using certain parameters in Proposition~\ref{thm:approx-Haar-U-Jiang}.
\begin{proof}[Proof of Proposition~\ref{thm:approx-Haar-U}]
Consider any $t\in(0,6]$ and $\Hd>\max\{e^{1/t^4},e^4\}$ such that $\frac{\sqrt{\Hd}}{\log\Hd}\geq \frac{72}{t}$.
Let the parameters in Proposition~\ref{thm:approx-Haar-U-Jiang} be as follows: $r=\frac{1}{\log\Hd}$, $s=(\log\Hd)^{3/4}$, $t=t$, and $\md=\md(\Hd)=\frac{t^2\Hd}{72\log\Hd}$. Note the conditions are satisfied: $\Hd>1$ implies $r>0$ and $s>0$, $\Hd> e^4$ implies $r=\frac{1}{\log\Hd}< 1/4$, and $t\leq 6$ implies $\md(\Hd)\leq\frac{\Hd}{2\log\Hd}=\frac{r\Hd}{2}$.

Since $\Hd\geq e^{1/t^4}$, $rs\leq t$. Thus if
$$\max_{i\in[\Hd],j\in[\md(\Hd)]}\left|\Qgrav_{ij}-\frac{Z_{ij}}{\sqrt{\Hd}}\right|\geq\epsilon(\Hd)=\frac{3t}{\sqrt{d}}\enspace,$$
then
$$\max_{i\in[\Hd],j\in[\md]}\left|\sqrt{\Hd}\cdot\Qgrav_{ij}-Z_{ij}\right|\geq 3t\geq rs+2t\enspace.$$
This shows
$$\Pr_{Z}\[\max_{i\in[\Hd],j\in[\md(\Hd)]}\left|\Qgrav_{ij}-\frac{Z_{ij}}{\sqrt{\Hd}}\right|\geq\epsilon(\Hd)\]\leq 4\md e^{-\Hd r^2/8}+\md\Hd e^{-s^2}+\frac{6\md\Hd}{t}\(1+\frac{t^2}{12(\md+t\sqrt{\Hd})}\)^{-\Hd}.$$
Define $\delta(\Hd)$ to be the right side above. It suffices to show $\delta(\Hd)=O\(\frac{1}{\Hd}\)$. By the choice of $r$ and $s$, $4\md e^{-\Hd r^2/8}=O\(\frac{1}{\Hd}\)$ and $\md\Hd e^{-s^2}=O\(\frac{1}{\Hd}\)$. We focus on the remaining term. Since $\frac{\sqrt{\Hd}}{\log\Hd}\geq \frac{72}{t}$ and by the choice of $\md$, we have $\frac{1}{\md+t\sqrt{\Hd}}\geq\frac{1}{2\md}$. Using this and the choice of $\md$,
$$1+\frac{t^2}{12(\md+t\sqrt{\Hd})}\geq 1+\frac{t^2}{12\cdot 2\md}=1+\frac{3\log\Hd}{\Hd}.$$
Thus
$$\(1+\frac{t^2}{12(\md+t\sqrt{\Hd})}\)^{-\Hd}\leq\(1+\frac{3\log\Hd}{\Hd}\)^{-\Hd}=O\(\frac{1}{\Hd^3}\)$$
and again by the choice of $\md$, and $t\leq 6$,
$$\frac{6\md\Hd}{t}\(1+\frac{t^2}{12(\md+t\sqrt{\Hd})}\)^{-\Hd}=O\(\frac{1}{\Hd\log\Hd}\)=O\(\frac{1}{\Hd}\).$$
This shows $\delta(\Hd)=O\(\frac{1}{\Hd}\)$, as desired.
\end{proof}

\section{A Fidelity Bound implies a Distinguishing Measurement}\label{sec:appendix-POVM}  

In this section we explain how our bounds on the fidelity between two states imply the existence of a quantum measurement (specifically, a positive operator-valued measure, POVM) that distinguishes with significant probability between these two states. The relationship between fidelity and distinguishing measurements is a known result (see e.g.~\cite{NieChu11}); we include it here for completeness.

We note that the fidelity between two \emph{random variables} $X$ and $Y$, each taking a value in $\{1,\ldots,\ell\}$ with probabilities $\vec{p}=(p_1,\ldots,p_\ell)$ and $\vec{q}=(q_1,\ldots,q_\ell)$ respectively, is defined to be
$${\displaystyle F({\vec {p}},{\vec {q}})\equiv \left(\sum _{k\in[\ell]}{\sqrt {p_{k}q_{k}}}\right)^{2}}.$$
The fidelity between two \emph{quantum states (density operators)}, $ F(\rho,\sigma )$, is a generalization of this notion. These two notions of fidelity have the following relation.
\begin{lem}[Fidelity of States and Fidelity of Measurement Outcomes~\cite{NieChu11}]\label{lem:fidelity-states-measurement}
For any two quantum states (density operators) $\rho$ and $\sigma$, 
$${\displaystyle F(\rho ,\sigma )=\min _{\{E_{k}\}}F({\vec {p}},{\vec {q}})}$$
where $F(\rho,\sigma)$ is the fidelity between $\rho$ and $\sigma$,
the minimum is taken over all possible POVMs $\{E_{k}\}$, and $F({\vec {p}},{\vec {q}})$ is the fidelity between the \emph{measurement outcomes} of $\{E_{k}\}$ applied to $\rho$ (which has probability vector $\vec{p}$) and of $\{E_{k}\}$ applied to $\sigma$ (which has probability vector $\vec{q}$).
\end{lem} 

Let $\{E_k\}_{k\in[\ell]}$ be the minimizing POVM in Lemma~\ref{lem:fidelity-states-measurement} and let $\vec{p}\equiv(p_1,\ldots,p_\ell)$ denote the probabilities of obtaining each of the $\ell$ possible measurement outcomes when $\{E_k\}_{k\in[\ell]}$ is applied to $\rho$. Similarly let $\vec{q}\equiv(q_1,\ldots,q_\ell)$ denote the outcome probabilities for $\{E_k\}_{k\in[\ell]}$ applied to $\sigma$. 
A natural way to compare the measurement outcomes for $\rho$ and $\sigma$ is to use the \emph{total variation distance} between the distribution of $\rho$'s measurement outcome and $\sigma$'s measurement outcome, which is related to how well one can guess whether a given state is $\rho$ or $\sigma$.

The definition of the total variation distance between two distributions with probability vectors $\vec{p}$ and $\vec{q}$ is
$$d_{\TV}(\vec{p},\vec{q})\equiv\frac{1}{2}\sum_{k\in[\ell]}|p_k-q_k|.$$
The total variation distance is $1$ if and only if the measurement outcomes that occur for $\rho$ and $\sigma$ are disjoint, making them perfectly distinguishable. For instance, if the measurement outcomes of applying $\{E_k\}_{k\in[\ell]}$ are denoted by $\{1,\ldots,\ell\}$, then an example of disjoint measurement outcomes is $\{E_k\}$ applied to $\rho$ always gives an even element of $\{1,\ldots,\ell\}$ and $\{E_k\}$ applied to $\sigma$ always gives an odd element of $\{1,\ldots,\ell\}$. The distribution of the measurement outcomes for $\rho$ and $\sigma$ are exactly the same if and only if the total variation distance is $0$.

Concretely, to distinguish whether a state is either $\rho$ or $\sigma$ with equal probability (e.g. either the algorithm's prediction $\qop\(\ket{\psi}\bra{\psi}\)$ or the fundamental time evolution $\Qgrav\ket{\psi}$), apply the minimal POVM $\{E_k\}$ from Lemma~\ref{lem:fidelity-states-measurement} to the state to get a measurement outcome $i\in\{1,\ldots,\ell\}$. If a measurement outcome is a more likely outcome for $\rho$ (i.e. if $p_i > q_i$), then guess the state is $\rho$; if $q_i \geq p_i$, then guess it is $\sigma$. This guess will be correct with probability $\frac{1}{2} + \frac{1}{2}\cdot d_{\TV}(\vec{p},\vec{q})$, an ``advantage'' of $\frac{1}{2}\cdot d_{\TV}(\vec{p},\vec{q})$ over a random guess. 

\medskip\medskip
Using Lemma~\ref{lem:fidelity-states-measurement} and e.g. relating $F(\vec{p},\vec{q})$ to the Hellinger distance, and the Hellinger distance to the total variation distance, we have the following.
\begin{lem} 
For any $\rho$ and $\sigma$, let $1-\delta\equiv F(\rho,\sigma)$. Then
$$\frac{\delta}{2}\leq d_{\TV}(\vec{p},\vec{q})\leq \sqrt{2\delta}$$
where $\vec{p}$ and $\vec{q}$ are the probability vectors corresponding to applying the minimal POVM of Lemma~\ref{lem:fidelity-states-measurement} to $\rho$ and $\sigma$ respectively. 
\end{lem}
\noindent Thus:
\begin{itemize}
    \item If $F(\rho,\sigma)\geq 1-\frac{1}{\exp(S)}$ (corresponding to our definition of successful learning, Equation~\ref{eqn:intro-def-success}), then the total variation distance $d_{\TV}(\vec{p},\vec{q})$ for \emph{any} quantum measurement (POVM) is at \emph{most} $\frac{\sqrt{2}}{\sqrt{\exp(S)}}$ so it is not possible to distinguish with any significant probability whether the state is $\rho$ or $\sigma$ by applying a POVM.
    \item If $F(\rho,\sigma)\leq 1-\frac{1}{\poly(S)}$ (corresponding to failure in learning, Equation~\ref{eqn:intro-def-failure}), then there exists a POVM $\{E_k\}$ such that the total variation distance is at \emph{least} $\frac{1}{2\poly(S)}$ so the above procedure distinguishes whether the state is $\rho$ or $\sigma$ with a significant advantage.
\end{itemize}
(This, in part, motivates our definition of successful learning and failure, or hardness, of learning.)

\section{Information Theory}\label{sec:appendix-info-theory}
In this section we state the results about differential entropy and mutual information that we use in our proofs, following the presentation in \cite{CovTho06}.

\paragraph{Differential Entropy.}
Note that the joint entropy of several variables $X_1,\ldots,X_\ell$ is the same as the entropy of the single vector-valued random variable $(X_1,\ldots,X_\ell)$. Furthermore, for any $k\in[\ell]$, letting $(X_{k+1},\ldots,X_\ell)$ be a vector-valued random variable, we have
\begin{equation}\label{eqn:joint-h-equiv-h}
    h(X_1,\ldots,X_\ell)=h(X_1,\ldots,X_k,(X_{k+1},\ldots,X_\ell)).
\end{equation}

\begin{lem}\label{lem:h-scale-constant}
For a random variable $X$ and constant $a\in\R$,
$h(aX)=h(X)+\log |a|.$
\end{lem}

\begin{lem}\label{lem:h-chain-rule}(Chain Rule)
For any jointly distributed random variables $X_1,\ldots,X_\ell$,
$$h(X_1,\ldots,X_\ell)=\sum_{k=1}^\ell h(X_k|X_1,\ldots,X_{k-1}).$$
\end{lem}

\begin{lem}\label{lem:h-cond-chain-rule}(Conditional Chain Rule)
For any jointly distributed random variables $X_1,\ldots,X_\ell,Y$,
$$h(X_1,\ldots,X_\ell|Y)=\sum_{k=1}^\ell h(X_k|Y,X_1,\ldots,X_{k-1}).$$
\end{lem}
\begin{proof}
By Lemma~\ref{lem:h-chain-rule}, we have
$$h(Y,(X_1,\ldots,X_\ell))=h(Y)+h((X_1,\ldots,X_\ell)|Y)$$
$$h(Y,X_1,\ldots,X_\ell)=h(Y)+h(X_1|Y)+h(X_2|Y,X_1)+\ldots+h(X_\ell|Y,X_1,\ldots,X_{\ell-1}).$$
Since the joint entropy of several variables is the same as the entropy of their corresponding vector (Equation~\ref{eqn:joint-h-equiv-h}),
$$h(X_1,\ldots,X_\ell|Y)=h(X_1|Y)+h(X_2|Y,X_1)+\ldots+h(X_\ell|Y,X_1,\ldots,X_{\ell-1}).$$
\end{proof}

\paragraph{Mutual Information.}
The mutual information can be defined in terms of the Kullback–Leibler (KL) divergence.
\begin{defn}(Kullback–Leibler (KL) divergence)\label{defn:KL-divergence}
For probability density functions $p$ and $q$ of a continuous random variable, the Kullback–Leibler (KL) divergence is
$$\KLD(p\|q)=\int_{-\infty }^{\infty }p(x)\log \({\frac {p(x)}{q(x)}}\)\,dx.$$
The conventions $0\log\frac{0}{0}=0$, $0\log\frac{0}{q}=0$, and $p\log\frac{p}{0}=\infty$ are used.

\end{defn}

Let $X,Y,Z$ denote any jointly distributed random variables.
\begin{defn}(Mutual Information)\label{defn:I}
$$I(X;Y)=\KLD(p_{X,Y}\|p_{X}\otimes p_{Y})$$

\end{defn}

\begin{defn}(Conditional Mutual Information)\label{defn:cond-I}
\begin{align*}
    I(X;Y|Z)&=\Eover{Z}[\KLD(p_{X,Y|Z}\|p_{X|Z}\otimes p_{Y|Z})]\\
    &=\Eover{z\sim p_{Z}}[I(X;Y|Z=z)]
\end{align*}
\end{defn}

\begin{lem}(Nonnegativity of $\KLD$ and $I$)\label{lem:I-nonneg}
For any probability density functions $p$ and $q$,
$$\KLD(p\|q)\ge 0.$$
This implies
$$I(X;Y)\ge 0.$$
Then by definition
$$I(X;Y|Z)\ge 0$$
since for every $z\in\mathcal{Z}$, we can let $(X_z,Y_z)\equiv(X,Y)_{|Z=z}$ and then $I(X;Y|Z=z)=I(X_z;Y_z)\ge 0$. Thus $I(X;Y|Z)=\Eover{z\sim p_{Z}}[I(X;Y|Z=z)]\ge 0$.
\end{lem}

\begin{lem}(Chain Rule for Mutual Information)\label{lem:I-chain-rule}
$$I(X_1,\ldots,X_\ell;Y)=\sum_{k=1}^{\ell} I(X_k;Y|X_1,\ldots,X_{k-1}).$$
\end{lem}

\bibliographystyle{jhep}
\bibliography{all}

\providecommand{\href}[2]{#2}\begingroup\raggedright\begin{thebibliography}{10}

\bibitem{HarHay13}
D.~Harlow and P.~Hayden, \emph{{Quantum Computation vs. Firewalls}},
  \href{https://doi.org/10.1007/JHEP06(2013)085}{\emph{JHEP} {\bfseries 1306}
  (2013) 085} [\href{https://arxiv.org/abs/1301.4504}{{\ttfamily 1301.4504}}].

\bibitem{KimTan20}
I.~Kim, E.~Tang and J.~Preskill, \emph{{The ghost in the radiation: Robust
  encodings of the black hole interior}},
  \href{https://doi.org/10.1007/JHEP06(2020)031}{\emph{JHEP} {\bfseries 06}
  (2020) 031} [\href{https://arxiv.org/abs/2003.05451}{{\ttfamily
  2003.05451}}].

\bibitem{StaSus14}
D.~Stanford and L.~Susskind, \emph{{Complexity and Shock Wave Geometries}},
  \href{https://doi.org/10.1103/PhysRevD.90.126007}{\emph{Phys. Rev. D}
  {\bfseries 90} (2014) 126007}
  [\href{https://arxiv.org/abs/1406.2678}{{\ttfamily 1406.2678}}].

\bibitem{RobSta14}
D.~A. Roberts, D.~Stanford and L.~Susskind, \emph{{Localized shocks}},
  \href{https://doi.org/10.1007/JHEP03(2015)051}{\emph{JHEP} {\bfseries 03}
  (2015) 051} [\href{https://arxiv.org/abs/1409.8180}{{\ttfamily 1409.8180}}].

\bibitem{SusZha14}
L.~Susskind and Y.~Zhao, \emph{{Switchbacks and the Bridge to Nowhere}},
  \href{https://arxiv.org/abs/1408.2823}{{\ttfamily 1408.2823}}.

\bibitem{Ali15}
M.~Alishahiha, \emph{{Holographic Complexity}},
  \href{https://doi.org/10.1103/PhysRevD.92.126009}{\emph{Phys. Rev. D}
  {\bfseries 92} (2015) 126009}
  [\href{https://arxiv.org/abs/1509.06614}{{\ttfamily 1509.06614}}].

\bibitem{BroRob15}
A.~R. Brown, D.~A. Roberts, L.~Susskind, B.~Swingle and Y.~Zhao,
  \emph{{Holographic Complexity Equals Bulk Action?}},
  \href{https://doi.org/10.1103/PhysRevLett.116.191301}{\emph{Phys. Rev. Lett.}
  {\bfseries 116} (2016) 191301}
  [\href{https://arxiv.org/abs/1509.07876}{{\ttfamily 1509.07876}}].

\bibitem{BroRob15b}
A.~R. Brown, D.~A. Roberts, L.~Susskind, B.~Swingle and Y.~Zhao,
  \emph{{Complexity, action, and black holes}},
  \href{https://doi.org/10.1103/PhysRevD.93.086006}{\emph{Phys. Rev. D}
  {\bfseries 93} (2016) 086006}
  [\href{https://arxiv.org/abs/1512.04993}{{\ttfamily 1512.04993}}].

\bibitem{LehMye16}
L.~Lehner, R.~C. Myers, E.~Poisson and R.~D. Sorkin, \emph{{Gravitational
  action with null boundaries}},
  \href{https://doi.org/10.1103/PhysRevD.94.084046}{\emph{Phys. Rev. D}
  {\bfseries 94} (2016) 084046}
  [\href{https://arxiv.org/abs/1609.00207}{{\ttfamily 1609.00207}}].

\bibitem{CouFis16}
J.~Couch, W.~Fischler and P.~H. Nguyen, \emph{{Noether charge, black hole
  volume, and complexity}},
  \href{https://doi.org/10.1007/JHEP03(2017)119}{\emph{JHEP} {\bfseries 03}
  (2017) 119} [\href{https://arxiv.org/abs/1610.02038}{{\ttfamily
  1610.02038}}].

\bibitem{ChaMar16}
S.~Chapman, H.~Marrochio and R.~C. Myers, \emph{{Complexity of Formation in
  Holography}}, \href{https://doi.org/10.1007/JHEP01(2017)062}{\emph{JHEP}
  {\bfseries 01} (2017) 062}
  [\href{https://arxiv.org/abs/1610.08063}{{\ttfamily 1610.08063}}].

\bibitem{CarMye16}
D.~Carmi, R.~C. Myers and P.~Rath, \emph{{Comments on Holographic Complexity}},
  \href{https://doi.org/10.1007/JHEP03(2017)118}{\emph{JHEP} {\bfseries 03}
  (2017) 118} [\href{https://arxiv.org/abs/1612.00433}{{\ttfamily
  1612.00433}}].

\bibitem{CarCha17}
D.~Carmi, S.~Chapman, H.~Marrochio, R.~C. Myers and S.~Sugishita, \emph{{On the
  Time Dependence of Holographic Complexity}},
  \href{https://doi.org/10.1007/JHEP11(2017)188}{\emph{JHEP} {\bfseries 11}
  (2017) 188} [\href{https://arxiv.org/abs/1709.10184}{{\ttfamily
  1709.10184}}].

\bibitem{EngWal17a}
N.~Engelhardt and A.~C. Wall, \emph{{No Simple Dual to the Causal Holographic
  Information?}}, \href{https://doi.org/10.1007/JHEP04(2017)134}{\emph{JHEP}
  {\bfseries 04} (2017) 134}
  [\href{https://arxiv.org/abs/1702.01748}{{\ttfamily 1702.01748}}].

\bibitem{EngWal18}
N.~Engelhardt and A.~C. Wall, \emph{{Coarse Graining Holographic Black Holes}},
  \href{https://doi.org/10.1007/JHEP05(2019)160}{\emph{JHEP} {\bfseries 05}
  (2019) 160} [\href{https://arxiv.org/abs/1806.01281}{{\ttfamily
  1806.01281}}].

\bibitem{BroGha19}
A.~R. Brown, H.~Gharibyan, G.~Penington and L.~Susskind, \emph{{The
  Python\textquoteright{}s Lunch: geometric obstructions to decoding Hawking
  radiation}}, \href{https://doi.org/10.1007/JHEP08(2020)121}{\emph{JHEP}
  {\bfseries 08} (2020) 121}
  [\href{https://arxiv.org/abs/1912.00228}{{\ttfamily 1912.00228}}].

\bibitem{EngPen21a}
N.~Engelhardt, G.~Penington and A.~Shahbazi-Moghaddam, \emph{{A world without
  pythons would be so simple}},
  \href{https://doi.org/10.1088/1361-6382/ac2de5}{\emph{Class. Quant. Grav.}
  {\bfseries 38} (2021) 234001}
  [\href{https://arxiv.org/abs/2102.07774}{{\ttfamily 2102.07774}}].

\bibitem{EngPen21b}
N.~Engelhardt, G.~Penington and A.~Shahbazi-Moghaddam, \emph{{Finding Pythons
  in Unexpected Places}},  \href{https://arxiv.org/abs/2105.09316}{{\ttfamily
  2105.09316}}.

\bibitem{AkeEng22}
C.~Akers, N.~Engelhardt, D.~Harlow, G.~Penington and S.~Vardhan, \emph{{The
  black hole interior from non-isometric codes and complexity}},
  \href{https://arxiv.org/abs/2207.06536}{{\ttfamily 2207.06536}}.

\bibitem{HayPre07}
P.~Hayden and J.~Preskill, \emph{{Black holes as mirrors: quantum information
  in random subsystems}},
  \href{https://doi.org/10.1088/1126-6708/2007/09/120}{\emph{JHEP} {\bfseries
  09} (2007) 120} [\href{https://arxiv.org/abs/0708.4025}{{\ttfamily
  0708.4025}}].

\bibitem{ChoSha21}
J.~Choi et~al., \emph{{Preparing random states and benchmarking with many-body
  quantum chaos}},
  \href{https://doi.org/10.1038/s41586-022-05442-1}{\emph{Nature} {\bfseries
  613} (2023) 468} [\href{https://arxiv.org/abs/2103.03535}{{\ttfamily
  2103.03535}}].

\bibitem{HoCho22}
W.~W. Ho and S.~Choi, \emph{Exact emergent quantum state designs from quantum
  chaotic dynamics},
  \href{https://doi.org/10.1103/physrevlett.128.060601}{\emph{Physical Review
  Letters} {\bfseries 128} (2022) }.

\bibitem{RobYos17}
D.~A. Roberts and B.~Yoshida, \emph{{Chaos and complexity by design}},
  \href{https://doi.org/10.1007/JHEP04(2017)121}{\emph{JHEP} {\bfseries 04}
  (2017) 121} [\href{https://arxiv.org/abs/1610.04903}{{\ttfamily
  1610.04903}}].

\bibitem{Haa91}
F.~Haake, \emph{Quantum signatures of chaos}. Springer, 1991.

\bibitem{Pag93a}
D.~N. Page, \emph{Average entropy of a subsystem},
  \href{https://doi.org/10.1103/physrevlett.71.1291}{\emph{Physical Review
  Letters} {\bfseries 71} (1993) 1291}.

\bibitem{Haw75}
S.~W. Hawking, \emph{Particle creation by black holes}, {\emph{Commun. Math.
  Phys.} {\bfseries 43} (1975) 199}.

\bibitem{EngWal17b}
N.~Engelhardt and A.~C. Wall, \emph{{Decoding the Apparent Horizon:
  Coarse-Grained Holographic Entropy}},
  \href{https://doi.org/10.1103/PhysRevLett.121.211301}{\emph{Phys. Rev. Lett.}
  {\bfseries 121} (2018) 211301}
  [\href{https://arxiv.org/abs/1706.02038}{{\ttfamily 1706.02038}}].

\bibitem{Jay57a}
E.~T. Jaynes, \emph{Information theory and statistical mechanics},
  \href{https://doi.org/10.1103/PhysRev.106.620}{\emph{Phys. Rev.} {\bfseries
  106} (1957) 620}.

\bibitem{Jay57b}
E.~T. Jaynes, \emph{{Information Theory and Statistical Mechanics. II}},
  \href{https://doi.org/10.1103/PhysRev.108.171}{\emph{Phys. Rev.} {\bfseries
  108} (1957) 171}.

\bibitem{Mal97}
J.~Maldacena, \emph{The large {$N$} limit of superconformal field theories and
  supergravity}, {\emph{Adv. Theor. Math. Phys.} {\bfseries 2} (1998) 231}
  [\href{https://arxiv.org/abs/hep-th/9711200}{{\ttfamily hep-th/9711200}}].

\bibitem{EngWal14}
N.~Engelhardt and A.~C. Wall, \emph{{Quantum Extremal Surfaces: Holographic
  Entanglement Entropy beyond the Classical Regime}},
  \href{https://doi.org/10.1007/JHEP01(2015)073}{\emph{JHEP} {\bfseries 01}
  (2015) 073} [\href{https://arxiv.org/abs/1408.3203}{{\ttfamily 1408.3203}}].

\bibitem{BouCha19}
R.~Bousso, V.~Chandrasekaran and A.~Shahbazi-Moghaddam, \emph{{From black hole
  entropy to energy-minimizing states in QFT}},
  \href{https://doi.org/10.1103/PhysRevD.101.046001}{\emph{Phys. Rev. D}
  {\bfseries 101} (2020) 046001}
  [\href{https://arxiv.org/abs/1906.05299}{{\ttfamily 1906.05299}}].

\bibitem{Pag93b}
D.~N. Page, \emph{{Information in black hole radiation}},
  \href{https://doi.org/10.1103/PhysRevLett.71.3743}{\emph{Phys. Rev. Lett.}
  {\bfseries 71} (1993) 3743}
  [\href{https://arxiv.org/abs/hep-th/9306083}{{\ttfamily hep-th/9306083}}].

\bibitem{Pen19}
G.~Penington, \emph{{Entanglement Wedge Reconstruction and the Information
  Paradox}}, \href{https://doi.org/10.1007/JHEP09(2020)002}{\emph{JHEP}
  {\bfseries 09} (2020) 002}
  [\href{https://arxiv.org/abs/1905.08255}{{\ttfamily 1905.08255}}].

\bibitem{AEMM}
A.~Almheiri, N.~Engelhardt, D.~Marolf and H.~Maxfield, \emph{{The entropy of
  bulk quantum fields and the entanglement wedge of an evaporating black
  hole}}, \href{https://doi.org/10.1007/JHEP12(2019)063}{\emph{JHEP} {\bfseries
  12} (2019) 063} [\href{https://arxiv.org/abs/1905.08762}{{\ttfamily
  1905.08762}}].

\bibitem{EngFol20}
N.~Engelhardt and A.~Folkestad, \emph{{Holography abhors visible trapped
  surfaces}}, \href{https://doi.org/10.1007/JHEP07(2021)066}{\emph{JHEP}
  {\bfseries 07} (2021) 066}
  [\href{https://arxiv.org/abs/2012.11445}{{\ttfamily 2012.11445}}].

\bibitem{BanDou98}
T.~Banks, M.~R. Douglas, G.~T. Horowitz and E.~J. Martinec, \emph{{AdS dynamics
  from conformal field theory}},
  \href{https://arxiv.org/abs/hep-th/9808016}{{\ttfamily hep-th/9808016}}.

\bibitem{HamKab05}
A.~Hamilton, D.~N. Kabat, G.~Lifschytz and D.~A. Lowe, \emph{{Local bulk
  operators in AdS/CFT: A Boundary view of horizons and locality}},
  \href{https://doi.org/10.1103/PhysRevD.73.086003}{\emph{Phys.Rev.} {\bfseries
  D73} (2006) 086003} [\href{https://arxiv.org/abs/hep-th/0506118}{{\ttfamily
  hep-th/0506118}}].

\bibitem{HamKab06}
A.~Hamilton, D.~N. Kabat, G.~Lifschytz and D.~A. Lowe, \emph{{Holographic
  representation of local bulk operators}},
  \href{https://doi.org/10.1103/PhysRevD.74.066009}{\emph{Phys.Rev.} {\bfseries
  D74} (2006) 066009} [\href{https://arxiv.org/abs/hep-th/0606141}{{\ttfamily
  hep-th/0606141}}].

\bibitem{HamKab06b}
A.~Hamilton, D.~N. Kabat, G.~Lifschytz and D.~A. Lowe, \emph{{Local bulk
  operators in AdS/CFT: A Holographic description of the black hole interior}},
  \href{https://doi.org/10.1103/PhysRevD.75.106001,
  10.1103/PhysRevD.75.129902}{\emph{Phys. Rev.} {\bfseries D75} (2007) 106001}
  [\href{https://arxiv.org/abs/hep-th/0612053}{{\ttfamily hep-th/0612053}}].

\bibitem{HeeMar}
I.~Heemskerk, D.~Marolf, J.~Polchinski and J.~Sully, \emph{{Bulk and
  Transhorizon Measurements in AdS/CFT}},
  \href{https://doi.org/10.1007/JHEP10(2012)165}{\emph{JHEP} {\bfseries 10}
  (2012) 165} [\href{https://arxiv.org/abs/1201.3664}{{\ttfamily 1201.3664}}].

\bibitem{BouFre12}
R.~Bousso, B.~Freivogel, S.~Leichenauer, V.~Rosenhaus and C.~Zukowski,
  \emph{{Null Geodesics, Local CFT Operators and AdS/CFT for Subregions}},
  \href{https://doi.org/10.1103/PhysRevD.88.064057}{\emph{Phys.Rev.} {\bfseries
  D88} (2013) 064057} [\href{https://arxiv.org/abs/1209.4641}{{\ttfamily
  1209.4641}}].

\bibitem{SekSus09}
Y.~Sekino and L.~Susskind, \emph{{Census Taking in the Hat: FRW/CFT Duality}},
  \href{https://doi.org/10.1103/PhysRevD.80.083531}{\emph{Phys. Rev.}
  {\bfseries D80} (2009) 083531}
  [\href{https://arxiv.org/abs/0908.3844}{{\ttfamily 0908.3844}}].

\bibitem{MalShe15}
J.~Maldacena, S.~H. Shenker and D.~Stanford, \emph{{A bound on chaos}},
  \href{https://doi.org/10.1007/JHEP08(2016)106}{\emph{JHEP} {\bfseries 08}
  (2016) 106} [\href{https://arxiv.org/abs/1503.01409}{{\ttfamily
  1503.01409}}].

\bibitem{CotGur16}
J.~S. Cotler, G.~Gur-Ari, M.~Hanada, J.~Polchinski, P.~Saad, S.~H. Shenker
  et~al., \emph{{Black Holes and Random Matrices}},
  \href{https://doi.org/10.1007/JHEP05(2017)118}{\emph{JHEP} {\bfseries 05}
  (2017) 118} [\href{https://arxiv.org/abs/1611.04650}{{\ttfamily
  1611.04650}}].

\bibitem{PirSun20}
L.~Piroli, C.~S\"underhauf and X.-L. Qi, \emph{{A Random Unitary Circuit Model
  for Black Hole Evaporation}},
  \href{https://doi.org/10.1007/JHEP04(2020)063}{\emph{JHEP} {\bfseries 04}
  (2020) 063} [\href{https://arxiv.org/abs/2002.09236}{{\ttfamily
  2002.09236}}].

\bibitem{Val84}
L.~G. Valiant, \emph{A theory of the learnable},
  \href{https://doi.org/10.1145/1968.1972}{\emph{Commun. {ACM}} {\bfseries 27}
  (1984) 1134}.

\bibitem{Ans22}
A.~Anshu, \emph{Some {Recent} {Progress} in {Learning} {Theory}: The {Quantum}
  {Side}}, {\emph{Harvard Data Science Review} {\bfseries 4} (2022) }.

\bibitem{AruWol17}
S.~Arunachalam and R.~de~Wolf, \emph{A survey of quantum learning theory},
  {\emph{CoRR} {\bfseries abs/1701.06806} (2017) }
  [\href{https://arxiv.org/abs/1701.06806}{{\ttfamily 1701.06806}}].

\bibitem{Sus20}
L.~Susskind, \emph{Horizons protect church-turing},  2020.
\newblock 10.48550/ARXIV.2003.01807.

\bibitem{AruGri21}
S.~Arunachalam, A.~B. Grilo and A.~Sundaram, \emph{Quantum hardness of learning
  shallow classical circuits},
  \href{https://doi.org/10.1137/20M1344202}{\emph{{SIAM} J. Comput.} {\bfseries
  50} (2021) 972}.

\bibitem{HuaChe22}
H.-Y. Huang, S.~Chen and J.~Preskill, \emph{Learning to predict arbitrary
  quantum processes},  2022.
\newblock 10.48550/ARXIV.2210.14894.

\bibitem{Jia10}
T.~Jiang, \emph{The entries of haar-invariant matrices from the classical
  compact groups},
  \href{https://doi.org/10.1007/s10959-009-0241-7}{\emph{Journal of Theoretical
  Probability} {\bfseries 23} (2010) 1227}.

\bibitem{Jia09}
T.~Jiang, \emph{Approximation of haar distributed matrices and limiting
  distributions of eigenvalues of jacobi ensembles}, {\emph{Probability theory
  and related fields} {\bfseries 144} (2009) 221}.

\bibitem{NieChu11}
M.~A. Nielsen and I.~L. Chuang, \emph{Quantum Computation and Quantum
  Information: 10th Anniversary Edition}. Cambridge University Press, USA,
  10th~ed., 2011.

\bibitem{CotHun17}
J.~Cotler, N.~Hunter-Jones, J.~Liu and B.~Yoshida, \emph{{Chaos, Complexity,
  and Random Matrices}},
  \href{https://doi.org/10.1007/JHEP11(2017)048}{\emph{JHEP} {\bfseries 11}
  (2017) 048} [\href{https://arxiv.org/abs/1706.05400}{{\ttfamily
  1706.05400}}].

\bibitem{SheSta13}
S.~H. Shenker and D.~Stanford, \emph{{Black holes and the butterfly effect}},
  \href{https://doi.org/10.1007/JHEP03(2014)067}{\emph{JHEP} {\bfseries 03}
  (2014) 067} [\href{https://arxiv.org/abs/1306.0622}{{\ttfamily 1306.0622}}].

\bibitem{JiLiu18}
Z.~Ji, Y.~Liu and F.~Song, \emph{Pseudorandom quantum states},  in
  \emph{Advances in Cryptology - {CRYPTO} 2018 - 38th Annual International
  Cryptology Conference, Santa Barbara, CA, USA, August 19-23, 2018,
  Proceedings, Part {III}}, H.~Shacham and A.~Boldyreva, eds., vol.~10993 of
  \emph{Lecture Notes in Computer Science}, pp.~126--152, Springer, 2018,
  \href{https://doi.org/10.1007/978-3-319-96878-0\_5}{DOI}.

\bibitem{BouFef19}
A.~Bouland, B.~Fefferman and U.~Vazirani, \emph{{Computational
  pseudorandomness, the wormhole growth paradox, and constraints on the AdS/CFT
  duality}},  \href{https://arxiv.org/abs/1910.14646}{{\ttfamily 1910.14646}}.

\bibitem{BouFef22}
A.~Bouland, B.~Fefferman, S.~Ghosh, U.~Vazirani and Z.~Zhou, \emph{Quantum
  pseudoentanglement},  2022.
\newblock 10.48550/ARXIV.2211.00747.

\bibitem{HuaBro21}
H.-Y. Huang et~al., \emph{{Quantum advantage in learning from experiments}},
  \href{https://doi.org/10.1126/science.abn7293}{\emph{Science} {\bfseries 376}
  (2022) abn7293} [\href{https://arxiv.org/abs/2112.00778}{{\ttfamily
  2112.00778}}].

\bibitem{BraShm20}
Z.~Brakerski and O.~Shmueli, \emph{Scalable pseudorandom quantum states},  in
  \emph{Advances in Cryptology - {CRYPTO} 2020 - 40th Annual International
  Cryptology Conference, {CRYPTO} 2020, Santa Barbara, CA, USA, August 17-21,
  2020, Proceedings, Part {II}}, D.~Micciancio and T.~Ristenpart, eds.,
  vol.~12171 of \emph{Lecture Notes in Computer Science}, pp.~417--440,
  Springer, 2020, \href{https://doi.org/10.1007/978-3-030-56880-1\_15}{DOI}.

\bibitem{Kha93}
M.~Kharitonov, \emph{Cryptographic hardness of distribution-specific learning},
   in \emph{Proceedings of the Twenty-Fifth Annual {ACM} Symposium on Theory of
  Computing, May 16-18, 1993, San Diego, CA, {USA}}, S.~R. Kosaraju, D.~S.
  Johnson and A.~Aggarwal, eds., pp.~372--381, {ACM}, 1993,
  \href{https://doi.org/10.1145/167088.167197}{DOI}.

\bibitem{Kha92}
M.~Kharitonov, \emph{Cryptographic lower bounds for learnability of boolean
  functions on the uniform distribution},  in \emph{Proceedings of the Fifth
  Annual {ACM} Conference on Computational Learning Theory, {COLT} 1992,
  Pittsburgh, PA, USA, July 27-29, 1992}, D.~Haussler, ed., pp.~29--36, {ACM},
  1992, \href{https://doi.org/10.1145/130385.130388}{DOI}.

\bibitem{Aar06}
S.~Aaronson, \emph{The learnability of quantum states}, {\emph{Electron.
  Colloquium Comput. Complex.} {\bfseries {TR06-106}} (2006) }
  [\href{https://arxiv.org/abs/TR06-106}{{\ttfamily TR06-106}}].

\bibitem{ChuLin18}
K.-M. Chung and H.-H. Lin, \emph{Sample efficient algorithms for learning
  quantum channels in pac model and the approximate state discrimination
  problem},  2018.
\newblock 10.48550/ARXIV.1810.10938.

\bibitem{Lin16}
Y.~Lindell, ``How to simulate it - a tutorial on the simulation proof
  technique.'' Cryptology ePrint Archive, Paper 2016/046, 2016.

\bibitem{Har13}
A.~W. Harrow, \emph{The church of the symmetric subspace},  2013.
\newblock 10.48550/ARXIV.1308.6595.

\bibitem{Tou15}
D.~Touchette, \emph{Quantum information complexity},  in \emph{Proceedings of
  the Forty-Seventh Annual {ACM} on Symposium on Theory of Computing, {STOC}
  2015, Portland, OR, USA, June 14-17, 2015}, R.~A. Servedio and R.~Rubinfeld,
  eds., pp.~317--326, {ACM}, 2015,
  \href{https://doi.org/10.1145/2746539.2746613}{DOI}.

\bibitem{KerLau16}
I.~Kerenidis, M.~Lauri{\`{e}}re, F.~L. Gall and M.~Rennela, \emph{Information
  cost of quantum communication protocols},
  \href{https://doi.org/10.26421/QIC16.3-4-1}{\emph{Quantum Inf. Comput.}
  {\bfseries 16} (2016) 181}.

\bibitem{CovTho06}
T.~M. Cover and J.~A. Thomas, \emph{Elements of Information Theory (Wiley
  Series in Telecommunications and Signal Processing)}. Wiley-Interscience,
  USA, 2006.

\end{thebibliography}\endgroup

\end{document}